\documentclass[12pt]{article}
\usepackage{amssymb}
\usepackage{amsmath}
\usepackage{amsfonts}
\usepackage{bbm}
\usepackage{geometry}
\usepackage{setspace}
\usepackage[font={small}]{caption}
\usepackage{chbibref}
\usepackage{float}
\usepackage{color}
\usepackage{appendix}
\usepackage{natbib}
\usepackage{hyperref}
\usepackage{enumerate}
\usepackage{graphicx}
\usepackage{subcaption}
\usepackage{pdflscape}
\usepackage{mathtools}
\usepackage{multirow}
\usepackage{tikz}
\usepackage{tikz-network}
\usepackage{booktabs}
\usepackage{mathabx}
\usepackage[normalem]{ulem}
\hypersetup{
    colorlinks,
    citecolor=blue,
    filecolor=black,
    linkcolor=blue,
    urlcolor=black
}
\newtheorem{theorem}{Theorem}
\newtheorem{sectiontheorem}{Theorem}[section]

\newtheorem{assumption}{Assumption} 
\newtheorem{sectionassumption}{Assumption}[section]

\newtheorem{lemma}{Lemma}
\newtheorem{sectionlemma}{Lemma}[section]

\newenvironment{proof}[1][Proof]{\noindent \textbf{#1.} }{\  \rule{0.5em}{0.5em}}
\setbibref{References}{\sffamily}
\allowdisplaybreaks[1]

\definecolor{red}{RGB}{213,94,0}
\definecolor{green}{RGB}{0,158,115}

\newcommand{\indep}{\perp \!\!\! \perp}
\allowdisplaybreaks[3]

\begin{document}
\title{Triadic Network Formation%
\footnote{We are grateful to Paola Conconi, Banu Demir, Bo Honor\'{e} and Martin Weidner for fruitful discussions.}}
\author{
    Chris Muris\footnote{McMaster University. Email: \href{mailto:muerisc@mcmaster.ca}{\texttt{muerisc@mcmaster.ca}}.}
    \, and
    Cavit Pakel\footnote{University of Oxford. Email: \href{mailto:cavit.pakel@economics.ox.ac.uk}{\texttt{cavit.pakel@economics.ox.ac.uk}}.}
}

\maketitle

\begin{abstract}
We study estimation and inference for triadic link formation with dyad-level fixed effects in a nonlinear binary choice logit framework. Dyad-level effects provide a richer and more realistic representation of heterogeneity across pairs of dimensions (e.g. importer--exporter, importer--product, exporter--product), yet their sheer number creates a severe incidental parameter problem. We propose a novel ``hexad logit'' estimator and establish its consistency and asymptotic normality. Identification is achieved through a conditional likelihood approach that eliminates the fixed effects by conditioning on sufficient statistics, in the form of hexads---wirings that involve two nodes from each part of the network.
Our central finding is that dyad-level heterogeneity fundamentally changes how information accumulates. Unlike under node-level heterogeneity, where informative wirings automatically grow with link formation, under dyad-level heterogeneity the network may generate infinitely many links yet asymptotically zero informative wirings. We derive explicit sparsity thresholds that determine when consistency holds and when asymptotic normality is attainable. 
These results have important practical implications, as they reveal that there is a limit to how granular or disaggregate a dataset one can employ under dyad-level heterogeneity.
\end{abstract}

\section{Introduction}

Networks are prevalent in economics, and they appear in many contexts: firms link to suppliers, countries trade in products, and researchers collaborate on projects. Increasingly detailed datasets capture such interactions across multiple dimensions, revealing structures that go beyond simple dyads. Triadic relationships in particular increasingly feature in empirical work; e.g., in production decisions (\citealt{BDM22}), import-export link formation of multinational corporations (\citealt{CLMT24}), and also in investigations of the effect of economic reliance on a trade partner on political alignment (\citealt{KLR24}). At the same time, sparsity appears to be a pervasive feature of network data.\footnote{A typical example is firm-level production network data, which is the focus of a well-established literature in economics (\citealt{CTS19}). These networks cover the universe of firms in a country and are typically characterised by the overwhelming majority of the firms having very few connections (\citealt{BMUM18} and \citealt{BMY19}).} In modelling such networks, it is standard to account for unobserved heterogeneity using fixed effects.

In this paper we focus on the formation of sparse triadic networks, and develop econometric tools for estimation in a nonlinear binary choice logit model of link formation with dyad-level heterogeneity. Crucially, in such models it is not the links themselves but specific configurations of links---wirings---that carry information about the structural parameters. Our key finding is that under dyad-level heterogeneity, once sparsity exceeds a threshold the network may generate infinitely many links yet fail to produce enough informative wirings for inference. This is a fundamental departure from the behaviour under node-level heterogeneity: indeed we show that under node-level heterogeneity informative wirings accumulate automatically with link formation. We characterise precisely how sparse is `too sparse', deriving thresholds that determine when consistency holds and when asymptotic normality can be attained.

Our findings carry important practical implications. Dyad-level effects (e.g. importer--exporter, importer--product, exporter--product) offer a richer and more realistic representation of heterogeneity than classical node-level effects (e.g. importer, exporter, product). Yet our results reveal that this flexibility requires sparsity to remain within certain limits, which may not hold in more granular or disaggregate data. When these conditions are not met, one option is to revert to node-level heterogeneity, though such a modelling choice may not always be appropriate given the context.  

To formalise our theoretical framework and contributions, our link formation model is given by
\begin{align}
    Y_{ijk} = 1\{A_{ij} + B_{jk} + C_{ik} + X_{ijk}'\beta_0 \geq \varepsilon_{ijk}\},
    \label{eq:triadicformation}
\end{align}
where $i,j,k$ are the indices for the three network dimensions, $X_{ijk}$ is the vector of observable covariates with the corresponding parameter vector $\beta_0$, and $\varepsilon_{ijk}$ is the triad-specific logistic random shock. $A_{ij}$, $B_{jk}$ and $C_{ik}$ are the dyad-level fixed effects. Our interest is in estimating $\beta_0$. Under sparse network asymptotics, fixed effects are generally not point-identified, which in turn leads to failure to identify $\beta_0$. We use a conditional likelihood approach based on sufficient statistics for the fixed effects in order to point-identify $\beta_0$. For dyadic link formation, sufficient statistics have been proposed in the form of informative wirings between groups of four nodes, called tetrads (\citealt{Charbonneau17}, \citealt{Graham17}, \citealt{Jochmans18}). In the triadic setting we propose a conditional likelihood estimator based on hexad subnetworks---the \textit{hexad logit estimator}---and establish the conditions under which it is consistent and asymptotically normal.

To provide some more precise discussion on the collection of informative wirings, let $N$ be the number of nodes in each of the three parts of the network and define $\rho_N = \mathrm P (Y_{ijk}=1)$, the unconditional probability of link formation. The rate at which $\rho_N$ tends to zero measures the level of sparsity. A necessary condition for statistical analysis is that the average node degree in the network remains bounded from below by zero as $N \to \infty$, which ensures asymptotic non-emptiness of the network---but still allows the network to be sparse. Consequently, as $N\to\infty$, the network accumulates links. Under node-level heterogeneity, this asymptotic accumulation of links leads to a corresponding accumulation of informative wirings. This holds both in dyadic link formation models, as previously shown, and in triadic link formation, as we establish in this paper.

We show that under dyad-level heterogeneity accumulation of links does not necessarily translate into accumulation of informative wirings: specifically, if sparsity exceeds a certain threshold such that $\rho_N$ tends to zero faster than $N^{-3/2}$, the network will generate infinitely many links and yet produce an asymptotically vanishing number of informative wirings. The reason for this phenomenon lies in the nature of informative wirings. The richer heterogeneity structure translates into wirings that are more intricate, depending on a larger number of links and arranged in highly specific patterns. As a result, they are relatively less likely to appear in data than their counterparts under node-level heterogeneity. Put differently, noticing that the average degree per node in our setting is $O(N^2 \rho_N )$, collecting sufficiently many informative wirings requires networks that are less sparse than just non-empty.

Collecting informative wirings is equivalent to accumulating information on $\beta_0$. The situation identified here therefore has important consequences on the asymptotic behaviour of the hexad logit estimator $\widehat \beta$. First, we show that $\rho_N = O(N^{-\delta})$ with $0\leq \delta<3/2$ is indeed a necessary condition for $\widehat \beta \to_p \beta_0$. Furthermore, we also show that asymptotic normality does not necessarily follow unless $\rho_N = O(N^{-\delta})$ with $0\leq \delta<1$. As will be discussed in more detail, this condition is necessary to ensure that informative wirings accumulate faster than links as the network grows. We show that this is automatically satisfied for the dyadic and triadic models with node-level heterogeneity, but not for the triadic model under dyad-level heterogeneity. This requirement to accumulate informative wirings at a faster rate relates to the conditional likelihood score's interpretation as a U-statistic. In particular, unless $\rho_N = O(N^{-\delta})$ with $0\leq \delta<1$, the score behaves analogous to a highly degenerate U-statistic. Interestingly, this degeneracy occurs in a discrete fashion in the sense that as soon as the U-statistic becomes degenerate this degeneracy is of the highest possible order.\footnote{We note that this degeneracy is different from the one identified by \citet{Graham17} for the score of the tetrad logit estimator where he can still obtain asymptotic normality. That appears to be a typical feature of the score function in network data, and also appears in our analysis. However, the degeneracy spelled out here is of a different nature.} The resulting asymptotic distribution will be \textit{Gaussian chaos}. 

\textbf{Literature review.}
Our work contributes to three strands of literature:

First, it builds on the work on binary logit link formation in sparse dyadic networks. \citet{Charbonneau17}, \citet{Graham17}, and \citet{Jochmans18} identify sufficient statistics in the form of tetrad wirings, laying the foundation for conditional likelihood estimation. Recent extensions include distribution regressions (\citealt{Szini25}) and ordered choice models in dyadic settings (\citealt{MPZ25}). 
Our contribution is to establish the corresponding theory for triadic link formation in a tripartite setting with a full set of dyad-level fixed effects, where the sufficient statistics take the form of hexads.
For surveys of the broader literature on econometrics of networks, see \citet{dePaula17}, \citet{Graham17}, and \citet{GdP20}.

Second, our approach is connected to the literature on sufficient statistics and the incidental parameter problem. The idea of eliminating fixed effects through conditioning goes back to \citet{Rasch60}, \citet{Andersen70}, and \citet{Chamberlain80}. In panel data, identification issues arise in short panels (\citealt{ArellanoBonhomme11}); in networks, the analogue is sparsity. \citet{DBW25} provide a recent review and also a general result for obtaining sufficient statistics in the binary logit setting. See also \citet{BonhommeDano24} who consider the functional differencing approach.

Third, our work relates to the emerging literature on nonlinear three-way models. While linear three-way specifications with heterogeneity have been widely studied (\citealt{BMW17}), nonlinear counterparts are more recent. Recent examples are due to  \citet{Stammann23}, who analyses dense three-way nonlinear panels with dyad-specific effects and proposes bias correction methods, and \citet{WeidnerZylkin21} and \citet{YZ23} who consider the gravity model.

\section{Model} \label{sec:model}

We model the formation of connections among three nodes from three different partite sets (also called \textit{parts}). 
Let $i \in \{1,\ldots,N_1\}$, $j \in \{1,\ldots,N_2\}$, and $k \in \{1,\ldots,N_3\}$ denote nodes from these three sets, respectively. 
We assume $N=N_1=N_2=N_3$ for notational simplicity.
The key unit of analysis is a \emph{triad} $(i,j,k)$ which is an ordered 3-tuple with node $i$ belonging to the first part, $j$ to the second part, and $k$ to the third part. 
The set of all $N^3$ such triads is given by $\mathbb T_N = \{1,\ldots,N\}^3$.

The binary variable $Y_{ijk}$ denotes formation of a hyperedge joining the nodes of the triad $(i,j,k)$. The hyperedge that connects the nodes of the triad $(i,j,k)$ is denoted by $[i,j,k]$.
The set of nodes and the associated hyperedges form a hypergraph that we call a \emph{network}.
Formally, we consider $3$-partite $3$-uniform networks in the sense that hyperedges connect exactly three nodes, and no two nodes of a hyperedge belong to the same part. Figure \ref{fig:depicting_hypergraphs} visualises one such network with $N=4$ nodes in each of the three parts. In this example we observe the hyperedges $[i_1,j_2,k_1]$, $[i_2,j_1,k_2]$ and $[i_3,j_4,k_1]$.

\begin{figure}
    \centering
    \begin{tikzpicture}

        \SetVertexStyle[FillOpacity = 0.5]
    
        \Vertex[x = 0, y = 3, label = $i_1$]{i1}
        \Vertex[x = 0, y = 2, label = $i_2$]{i2}
        \Vertex[x = 0, y = 1, label = $i_3$]{i3}
        \Vertex[x = 0, y = 0, label = $i_4$]{i4}

        \Vertex[x = 2, y = 3, label = $j_1$, color = green]{j1}
        \Vertex[x = 2, y = 2, label = $j_2$, color = green]{j2}
        \Vertex[x = 2, y = 1, label = $j_3$, color = green]{j3}
        \Vertex[x = 2, y = 0, label = $j_4$, color = green]{j4}

        \Vertex[x = 4, y = 3, label = $k_1$, color = red]{k1}
        \Vertex[x = 4, y = 2, label = $k_2$, color = red]{k2}
        \Vertex[x = 4, y = 1, label = $k_3$, color = red]{k3}
        \Vertex[x = 4, y = 0, label = $k_4$, color = red]{k4}

        \Edge(i1)(j2)
        \Edge(j2)(k1)
        \Edge[color = blue](i3)(j4)
        \Edge[color = blue](j4)(k1)
        \Edge[color = brown](i2)(j1)
        \Edge[color = brown](j1)(k2)
        
    \end{tikzpicture}
    \caption{A network with $N=4$ nodes in each of the three parts. Across the $4^3$ triads, only three of them decide to form a hyperedge; these are $(i_1,j_2,k_1)$, $(i_2,j_1,k_2)$, and $(i_3,j_4,k_1)$, as indicated by the hyperedges in solid lines.}
    \label{fig:depicting_hypergraphs}
\end{figure}
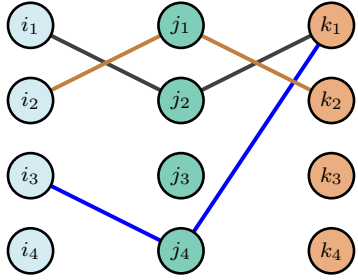

We model a triad's decision to form a hyperedge as a binary choice logit model with additive fixed effects. 
Let $X_{ijk} = h(X_i,X_j,X_k) \in \mathbb R^P$ be a triad-specific covariate (a known transformation of the node-specific observable characteristics) and $\beta_0$ be the associated $P \times 1$ parameter vector.
Let also $\varepsilon_{ijk}$ be a triad-specific logistic error with the CDF $\Lambda(u) = \exp(u)/(1+\exp(u))$.
We model the link formation decision as
\begin{align}
    Y_{ijk} = 1\{A_{ij} + B_{jk} + C_{ik} + X_{ijk}' \beta_0 - \varepsilon_{ijk} \geq 0\},
    \label{eq:ofe}
\end{align}
where $A_{ij},B_{jk},C_{ik}$ are dyad-specific fixed effects. 
For instance, in a setting where firm $i$ imports input $j$ from the exporting country $k$, these correspond to firm-input, input-exporter and firm-exporter effects, respectively. 
This overlapping heterogeneity structure provides the most general additive fixed effect structure in the given framework.
Finally, we assume that the random shocks $\varepsilon_{ijk}$ are independent across all triads, conditional on covariates and fixed effects.

In what follows, we use bold notation when collecting objects across the network. In particular, $\mathbf Y = \left(Y_{ijk}\right)_{\mathbb T_N}$ and 
$\mathbf X = \left(X_{ijk}\right)_{\mathbb T_N}$ collect the link formation information and observable characteristics across all triads, respectively. $\mathbf{Y}=\mathbf{y}$ denotes a particular realisation of the network.
The fixed effects for the triad $(i,j,k)$ are denoted  $\mathrm{F}_{ijk} = (A_{ij},B_{jk},C_{ik})$ and fixed effects across all triads are gathered in $\mathbf{F} = \left(\mathrm{F}_{ijk}\right)_{\mathbb T_N}$.

The next assumption formalises the discussion made so far.

\begin{assumption}[Likelihood]\label{a:logit}
The conditional likelihood of the network $\mathbf Y = \mathbf y$ is
\begin{equation*}
    \mathrm{P}\left(\mathbf Y = \mathbf y | \mathbf X, \mathbf{F}\right) 
    = 
    \prod_{(i,j,k) \in \mathbb T_N} \mathrm{P} \left(\left. Y_{ijk} = y_{ijk} \right| X_{ijk}, \mathrm{F}_{ijk} \right),
\end{equation*}
where 
\begin{align*}
    \mathrm{P} \left(\left. Y_{ijk} = y \right| X_{ijk}, \mathrm{F}_{ijk} \right) 
    = 
    \begin{cases}
    \Lambda\left( X_{ijk}' \beta_0 + A_{ij} + B_{jk} + C_{ik} \right) & \text{ if }y = 1, \\
    1-\Lambda\left( X_{ijk}' \beta_0 + A_{ij} + B_{jk} + C_{ik} \right) & \text{ if }y=0.
    \end{cases}
\end{align*}
\end{assumption}

\section{Identification}

In sparse networks, the fixed effects will generally not be point-identified, as connections involving different node-pairs are typically observed in very few instances, if at all. This is analogous to the identification failure in short nonlinear panels (see, e.g., \citealt{Chamberlain10} and \citealt{ArellanoBonhomme11}), and is an instance of the incidental parameter issue. 
As in the panel data literature, for models with logistic random shocks it is usually possible to overcome this problem by a conditional likelihood approach.
In what follows, we develop such a conditional likelihood approach that eliminates the fixed effects through carefully constructed sufficient statistics, thereby achieving identification of the common parameter $\beta_0$. 
Our sufficient statistics will be based on hexad subnetworks. 

\subsection{Hexad subnetworks}
\label{sec:hexads}

A conditional likelihood function that applies to the full network can in principle be constructed by conditioning on a sufficient statistic for the unobserved heterogeneity of the entire network. However, the resulting conditional likelihood would be computationally intractable, even for moderately large networks. Focussing on lower dimensional subnetworks, on the other hand, maintains the statistical properties needed for consistent estimation of $\beta_0$ while remaining computationally practical. This principle has been successfully applied in dyadic network models using tetrad subnetworks \citep{Charbonneau17, Graham17, Jochmans18}. To handle triadic networks, we use hexad subnetworks.

A hexad is a collection of six nodes. Our analysis is based on a specific type of hexad, consisting of two nodes from each of the three parts of the original full hypergraph. Let $\sigma=(i_1,j_1,k_1,i_2,j_2,k_2)$ be a generic hexad consisting of nodes $(i_1,i_2)$ from part 1, $(j_1,j_2)$ form part 2, and $(k_1,k_2)$ form part 3. Let $\Sigma$ be the set of all possible hexads consisting of two nodes from each of the three parts. There are $m_N = |\Sigma| = N^3(N-1)^3$ hexads in this structure. Notice that each hexad subnetwork contains $2^3$ distinct triads (and therefore anywhere from 0 to $2^3$ hyperedges). Let
$X_\sigma$ and $\mathrm{F}_\sigma$ collect the covariates and fixed effects for all triads in the subnetwork generated by the hexad $\sigma$. Hexad subnetworks provide enough structure to find sufficient statistics for the fixed effects $\mathrm{F}_\sigma$, as we will show in the next section. The resulting hexad-specific conditional likelihood functions can then be combined to obtain a (composite) conditional likelihood function for estimation of $\beta_0$.

Let $d_\sigma(\cdot)$ return the degree of a chosen node in $\sigma$; that is, the number of hyperedges that this node belongs to in the subnetwork generated by $\sigma$. The degree sequence for $\sigma=(i_1,j_1,k_1,i_2,j_2,k_2)$ is then given by 
$D_{\sigma} = (d_\sigma(i_1),d_\sigma(j_1),d_\sigma(k_1),d_\sigma(i_2),d_\sigma(j_2),d_\sigma(k_2)).$ The degree sequence is a key quantity in obtaining a conditional likelihood function; see \citet{Graham17}.

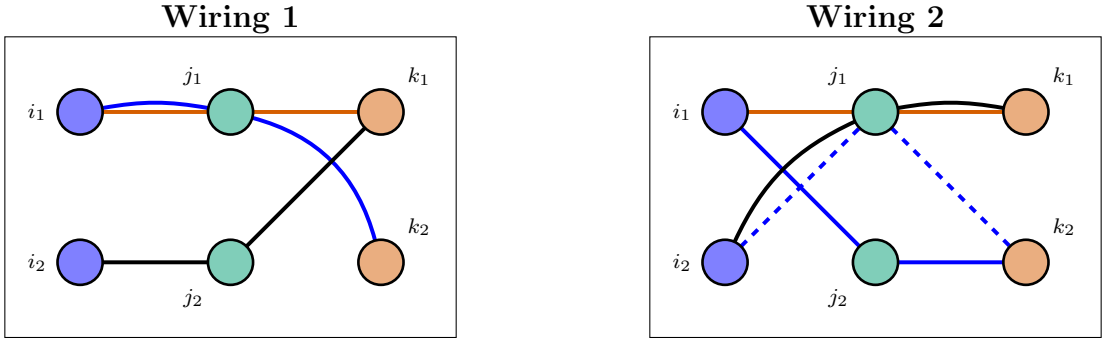
\begin{figure}[t]
    \centering
    
    \begin{subfigure}[b]{0.45\linewidth}
        \textbf{Wiring 1}
        \centering
        
        \begin{tikzpicture}
        \SetVertexStyle[FillOpacity = 0.5]
        \draw (0.0,0.0) rectangle (6,4);
        \Vertex[x = 1, y = 3, color = blue, label =$i_1$, position = left]{i1}
        \Vertex[x = 1, y = 1, color = blue, label =$i_2$, position = left]{i2}
        \Vertex[x = 3, y = 3, color = green, label =$j_1$, position = above left]{j1}
        \Vertex[x = 3, y = 1, color = green, label =$j_2$, position = below left]{j2}
        \Vertex[x = 5, y = 3, color = red, label =$k_1$, position = above right]{k1}
        \Vertex[x = 5, y = 1, color = red, label =$k_2$, position = above right]{k2}
        
        \Edge[color=red, bend = 0](i1)(j1)
        \Edge[color=red, bend = 0](j1)(k1)
        
        \Edge[color=blue, bend = 10](i1)(j1)
        \Edge[color=blue, bend = 30](j1)(k2)
        
        \Edge[color=black, bend = 0](i2)(j2)
        \Edge[color=black, bend = -0](j2)(k1)
        
        \end{tikzpicture}
        
    \end{subfigure}
    \hfill
    \begin{subfigure}[b]{0.45\linewidth}
        \textbf{Wiring 2}
        \centering
        \begin{tikzpicture}
        \SetVertexStyle[FillOpacity = 0.5]
        \draw (0.0,0.0) rectangle (6,4);
        \Vertex[x = 1, y = 3, color = blue, label =$i_1$, position = left]{i1}
        \Vertex[x = 1, y = 1, color = blue, label =$i_2$, position = left]{i2}
        \Vertex[x = 3, y = 3, color = green, label =$j_1$, position = above left]{j1}
        \Vertex[x = 3, y = 1, color = green, label =$j_2$, position = below left]{j2}
        \Vertex[x = 5, y = 3, color = red, label =$k_1$, position = above right]{k1}
        \Vertex[x = 5, y = 1, color = red, label =$k_2$, position = above right]{k2}

        \Edge[color=red, bend = 0](i1)(j1)
        \Edge[color=red, bend = 0](j1)(k1)
        
        \Edge[color=blue, bend = 0](i1)(j2)
        \Edge[color=blue, bend = -0](j2)(k2)
        
        \Edge[color=black, bend = 20](i2)(j1)
        \Edge[color=black, bend = 10](j1)(k1)

        \Edge[color=blue, bend = 0, style=dashed](i2)(j1)
        \Edge[color=blue, bend = 0, style=dashed](j1)(k2)
        
        \end{tikzpicture}
    \end{subfigure}

    \caption{Two examples of hexad wirings. Wiring 1 has degree sequence $(2,2,2,1,1,1)$ whereas Wiring 2 has degree sequence $(2,3,2,2,1,2)$.}
    \label{fig:subhypergraphs_illustration}
\end{figure}
As an illustration of the concepts introduced thus far, consider the hexads in Figure \ref{fig:subhypergraphs_illustration}. 
Hyperedges are illustrated as line or curve segments of the same colour.
In Wiring 1 there are three hyperedges: $[i_1,j_1,k_2]$, $[i_1,j_1,k_1]$ and $[i_2,j_2,k_1]$. 
Wiring 2, on the other hand, contains four hyperedges: $[i_1,j_1,k_1]$, $[i_1,j_2,k_2]$, $[i_2,j_1,k_1]$ and $[i_2,j_1,k_2]$. 
The degree sequences of Wirings 1 and 2 are $(2,2,2,1,1,1)$ and $(2,3,2,2,1,2)$, respectively. 

\subsection{Informative wirings and sufficiency}
\label{sec:sufficiency}

Finding a sufficient statistic for $\mathrm{F}_\sigma$ requires finding degree sequences that admit at least two different wirings that contain the same fixed effects. Consequently, such degree sequences can be used to isolate variation attributable to $\beta_0$ rather than to the fixed effects, thereby obtaining identification of $\beta_0$. This intuition has been used successfully in \citet{Charbonneau17}, \citet{Graham17}  and \citet{Jochmans18} in dyadic link formation. 

We are specifically interested in using the minimal degree sequence that permits identification---by minimal we mean that there exists no other identifying degree sequence where at least one node has a lower degree (and no nodes have a higher degree). In Appendix \ref{sec:simpleds} we show that in the given setting this is achieved by the degree sequence $(2,2,2,2,2,2)$, which yields two wirings that can be used to construct a conditional likelihood function.
We call these two wirings \emph{informative wirings} or \textit{identifying wirings}.
For a generic hexad $\sigma=(i_1,j_1,k_1,i_2,j_2,k_2)$ these two informative wirings correspond to the following collections of hyperedges:
\begin{align*}
    \text{Informative wiring \#1:} \quad \left\{[i_1,j_1,k_1]\,,\,[i_1,j_2,k_2]\,,\,[i_2,j_1,k_2]\,,\,[i_2,j_2,k_1]\right\};
    \\
    \text{Informative wiring \#2:} \quad \left\{[i_2,j_2,k_2]\,,\,[i_2,j_1,k_1]\,,\,[i_1,j_2,k_1]\,,\,[i_1,j_1,k_2]\right\}.
\end{align*}
These are illustrated in Figure \ref{fig:subhypergraphs_222222}: the first informative wiring corresponds to the left panel whereas the second is given in the right panel. Notice that the two informative wirings do not share any common hyperedges despite having the same degree sequence. However, and crucially, they both contain all the fixed effects in $\mathrm{F}_\sigma$. This can also be confirmed by noticing that both wirings contain all  possible dyadic links: since our fixed effects are dyad-specific, this confirms that both wirings contain all the fixed effects in $\mathrm{F}_\sigma$. Consequently, the differences between the wirings should be attributable to covariates, enabling point-identification of $\beta_0$.

Letting $\overline{Y}_{abc} = 1-Y_{abc}$, the indicator functions for these wirings are given by
\begin{align*}
    S_{\sigma, 1}
    = 
    Y_{i_1 j_1 k_1}
    Y_{i_1 j_2 k_2}
    Y_{i_2 j_1 k_2}
    Y_{i_2 j_2 k_1}
    \overline Y_{i_2 j_2 k_2}
    \overline Y_{i_1 j_2 k_1}
    \overline Y_{i_2 j_1 k_1}
    \overline Y_{i_1 j_1 k_2},
    \\
    S_{\sigma, 2}
    = 
    \overline Y_{i_1 j_1 k_1}
    \overline Y_{i_1 j_2 k_2}
    \overline Y_{i_2 j_1 k_2}
    \overline Y_{i_2 j_2 k_1}
    Y_{i_2 j_2 k_2}
    Y_{i_1 j_2 k_1}
    Y_{i_2 j_1 k_1}
    Y_{i_1 j_1 k_2}.
\end{align*}
In other words $S_{\sigma,1}$ indicates whether the hexad $\sigma$ admits the first informative wiring or not (and similarly for $S_{\sigma,2}$). Notice that a hexad cannot admit both informative wirings at once. Therefore, the binary variable
\begin{align}
    S_\sigma = S_{\sigma,1} + S_{\sigma,2},
    \label{eq:Stildedefn}
\end{align}
indicates whether the hexad $\sigma$ is informative or not.

\begin{figure}[t]
    \centering
    
    \begin{subfigure}[b]{0.45\linewidth}
        \centering
        \begin{tikzpicture}
        \SetVertexStyle[FillOpacity = 0.5]
        \draw (0.0,0.0) rectangle (6,4);
        \Vertex[x = 1, y = 3, color = blue, label =$i_1$, position = left]{i1}
        \Vertex[x = 1, y = 1, color = blue, label =$i_2$, position = left]{i2}
        \Vertex[x = 3, y = 3, color = green, label =$j_1$, position = above left]{j1}
        \Vertex[x = 3, y = 1, color = green, label =$j_2$, position = below left]{j2}
        \Vertex[x = 5, y = 3, color = red, label =$k_1$, position = above right]{k1}
        \Vertex[x = 5, y = 1, color = red, label =$k_2$, position = above right]{k2}

        \Edge[color=red, bend = 0](i1)(j1)
        \Edge[color=red, bend = 0](j1)(k1)
        
        \Edge[color=blue, bend = -10](i1)(j2)
        \Edge[color=blue, bend = -10](j2)(k2)
        
        \Edge[color=black, bend = 0](i2)(j1)
        \Edge[color=black, bend = 0](j1)(k2)

        \Edge[color=green, style={dashed}, bend = -10](i2)(j2)
        \Edge[color=green, style={dashed}, bend = -10](j2)(k1)
        
        \end{tikzpicture}
        \caption*{$S_{\sigma,1} = 1$}
    \end{subfigure}
    \hfill
    \begin{subfigure}[b]{0.45\linewidth}
        \centering
                \begin{tikzpicture}
        \SetVertexStyle[FillOpacity = 0.5]
        \draw (0.0,0.0) rectangle (6,4);
        \Vertex[x = 1, y = 3, color = blue, label =$i_1$, position = left]{i1}
        \Vertex[x = 1, y = 1, color = blue, label =$i_2$, position = left]{i2}
        \Vertex[x = 3, y = 3, color = green, label =$j_1$, position = above left]{j1}
        \Vertex[x = 3, y = 1, color = green, label =$j_2$, position = below left]{j2}
        \Vertex[x = 5, y = 3, color = red, label =$k_1$, position = above right]{k1}
        \Vertex[x = 5, y = 1, color = red, label =$k_2$, position = above right]{k2}

        \Edge[color=red, bend = 10](i1)(j1)
        \Edge[color=red, bend = 10](j1)(k2)
        
        \Edge[color=blue, bend = 0](i1)(j2)
        \Edge[color=blue, bend = 0](j2)(k1)
        
        \Edge[color=black, bend = 10](i2)(j1)
        \Edge[color=black, bend = 10](j1)(k1)

        \Edge[color=green, style={dashed}, bend = 0](i2)(j2)
        \Edge[color=green, style={dashed}, bend = 0](j2)(k2)
        
        \end{tikzpicture}
        \caption*{$S_{\sigma,2} = 1$}
    \end{subfigure}
    \caption{The two informative hexad wirings compatible with the degree sequence $(2, 2, 2, 2, 2, 2)$. These wirings, described by the indicator functions $S_{\sigma,1}$ and $S_{\sigma,2}$, provide identification of $\beta_0$ independent of the fixed effects.}
    \label{fig:subhypergraphs_222222}
\end{figure}
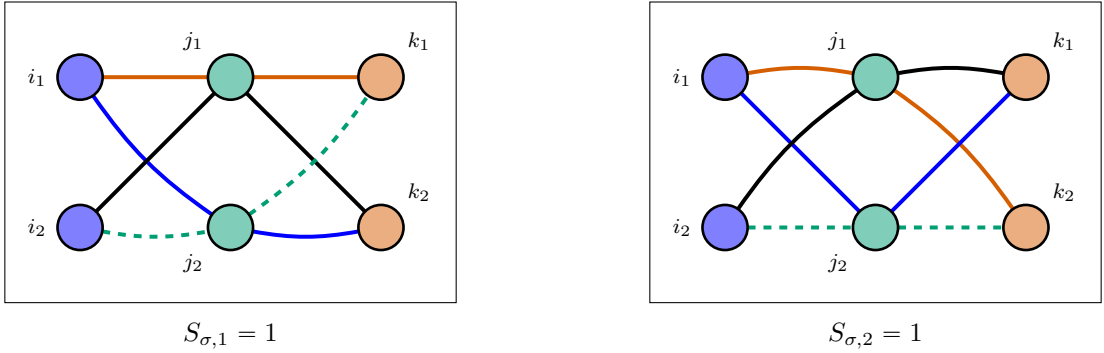

\begin{theorem}[Sufficiency]\label{thm:sufficiency_overlapping} 
Suppose that Assumption \ref{a:logit} holds. Then, 
\begin{align*}
    \mathrm{P}(S_{\sigma,1} = 1 | S_{\sigma} = 1, \mathrm{F}_{\sigma}, X_{\sigma}) 
    =
    \mathrm{P}(S_{\sigma,1} = 1 | S_{\sigma} = 1, X_{\sigma}) 
    = \Lambda(W_{\sigma}' \beta_0)
\end{align*}     
for any hexad $\sigma = (i_1,j_1,k_1,i_2,j_2,k_2)$,
where 
\begin{align*}
    W_{\sigma} = 
    \left(X_{i_1 j_1 k_1}+X_{i_1 j_2 k_2}+X_{i_2 j_1 k_2}+X_{i_2 j_2 k_1}\right) 
    - 
    \left(X_{i_1 j_2 k_1}+X_{i_1 j_1 k_2}+X_{i_2 j_1 k_1}+X_{i_2 j_2 k_2}\right)    .
\end{align*}
\end{theorem}
Theorem \ref{thm:sufficiency_overlapping} confirms that conditional on $\sigma$ admitting one of the two identifying wirings, the probability of observing the first informative wiring is independent of $\mathrm{F}_\sigma$ and has the logistic form. It directly follows that the conditional probability of observing the second informative wiring is also independent of $\mathrm{F}_\sigma$.

\section{The hexad logit estimator}
\label{sec:estimation}

In this section we propose a conditional likelihood estimator for $\beta_0$, the \textit{hexad logit estimator}. For a generic hexad $\sigma=(i_1,j_1,k_1,i_2,j_2,k_2)$ define
\begin{align*}
    W_{\sigma,1} &= X_{i_1 j_1 k_1}+X_{i_1 j_2 k_2}+X_{i_2 j_1 k_2}+X_{i_2 j_2 k_1}, \\
    W_{\sigma,2} &= X_{i_1 j_2 k_1}+X_{i_1 j_1 k_2}+X_{i_2 j_1 k_1}+X_{i_2 j_2 k_2},
\end{align*}
noting that $W_\sigma = W_{\sigma,1} - W_{\sigma,2}$ as defined in Theorem \ref{thm:sufficiency_overlapping}. Define, furthermore, the shorthand notation
$p_{\sigma, c}(\beta) = \mathrm{P}(S_{\sigma,c} = 1 | S_\sigma = 1, X_\sigma)$ where $c\in\{1,2\}$ indexes the two informative wirings defined in the previous section. Therefore, remembering Theorem \ref{thm:sufficiency_overlapping}, we have
\begin{align}
    p_{\sigma, c}(\beta)
    &= 
    \frac
    {\exp( W_{\sigma, c}' \beta)}
    {\sum_{c'=1}^2 \exp  ( W_{\sigma, c'}' \beta ) }.
    \label{eq:nol_pcb_overlap}
\end{align}
The conditional likelihood function that combines information across all informative hexads is given by
\begin{align}
    l_N(\beta) 
    &= 
    \frac{1}{m_N}
    \sum_{\sigma \in \Sigma} 
    l_\sigma(\beta),
    \qquad
    \text{where}
    \qquad
    l_\sigma(\beta)
    = 
    \sum_{c=1}^2 S_{\sigma, c} \log p_{\sigma, c}(\beta).
    \label{eq:log_likelihood}
\end{align}
We note that this is a composite log-likelihood function due to dependence across hexads that share common triads. 
Our proposed hexad logit estimator is given by 
\begin{align*}
    \widehat \beta = \arg\max_{\beta \in \mathcal B} l_N(\beta), 
\end{align*}
where $\mathcal B$ is a compact parameter space containing the true parameter $\beta_0$. Defining 
\begin{align*}
    \overline{W}_\sigma(\beta) 
    = 
    p_{\sigma,1}(\beta) W_{\sigma, 1} 
    + 
    p_{\sigma,2}(\beta) W_{\sigma, 2},
\end{align*}
the corresponding score function is given by
\begin{align*}
    Z_N(\beta)
    =
    \frac{1}{m_N}
    \sum_{\sigma \in \Sigma} 
    s_\sigma(\beta),
    \qquad
    \text{where}
    \qquad
    s_{\sigma}(\beta)
    = 
    \sum_{c = 1}^2  S _{\sigma, c} \left(W_{\sigma, c} - \overline{W}_\sigma(\beta) \right).
\end{align*} 
Also, the Hessian contribution for each hexad $\sigma$ is equal to
\begin{align*}
    H_{\sigma}(\beta) 
    &= 
    - S_\sigma \sum_{c = 1}^2 p_{\sigma, c}(\beta)\left( W_{\sigma, c} - \overline{W}_\sigma(\beta) \right) \left(W_{\sigma, c} - \overline{W}_\sigma(\beta) \right)'.
\end{align*}
Derivations for the score and Hessian are provided in Appendix \ref{app:score_and_Hessian}.

\section{Large sample theory}

This section develops the large sample theory for the hexad logit estimator, deriving conditions on network sparsity under which it is consistent and asymptotically normal.

A quantity that plays an important role in the asymptotic properties of the hexad logit estimator is the unconditional probability of link formation, defined as
\begin{align}
\rho_N  =  \mathrm{P}\left(Y_{ijk} = 1 \right). 
\label{eq:rhodefn}
\end{align}
Asymptotic sparsity of the network is conceptualised by $\rho_N$ tending to zero as $N$ increases. Importantly, $\rho_N$ cannot go to zero faster than $O(1/N^2)$ or the average node degree $N^2 \rho_N$ will also converge to zero, rendering the network asymptotically empty.\footnote{See, for example, Assumption 4.(ii) of \citet{Graham17} for a corresponding requirement in dyadic link formation.}

A second key quantity is
\begin{equation*}
p_N = \mathrm{P}(S_\sigma=1),
\end{equation*}
the unconditional probability that a hexad $\sigma$ is informative. Recalling that informative wirings consist of four hyperedges, we have $p_N = O(\rho_N^4)$. This quantity directly controls the rate at which informative wirings accumulate and is therefore essential to the asymptotic behaviour of the hexad logit estimator. 

We next introduce the assumptions that will be used in obtaining consistency and asymptotic normality.

\begin{assumption}[Compactness]\label{a:compact_and_interior}
    The parameter space $\mathcal B \subset \mathbb R^P$ is compact, and $\beta_0$ is in the interior of $\mathcal B$. 
    The support of the regressors $X_{ijk}$, $\mathcal X \subset \mathbb R^P$, is also compact.
\end{assumption}

\begin{assumption}[Random sampling]\label{a:random_sampling_nodes}
    Let $i = 1,\ldots,N$, $j=1,\ldots,N$, and $k=1,\ldots,N$ index random samples of nodes from the first, second, and third parts respectively, drawn from a population satisfying Assumptions \ref{a:logit} and \ref{a:compact_and_interior}.
\end{assumption}

\begin{assumption}[Identification]\label{a:identification}
     The matrix
     \begin{align*}
         \Gamma_0
         =
         - \lim_{N\to\infty}
         \frac{1}{p_N m_N} \sum_{\sigma \in \Sigma} \mathbb E \left( S_\sigma \sum_{c = 1}^2 p_{\sigma, c}(\beta_0)\left(W_{\sigma, c} - \overline W_\sigma(\beta_0) \right) \left(W_{\sigma, c} - \overline W_\sigma(\beta_0) \right)'\right),
     \end{align*}
     exists and is negative definite.
\end{assumption}
Assumptions \ref{a:compact_and_interior},  \ref{a:random_sampling_nodes} and \ref{a:identification} are analogous to those made by \citet{Graham17} and \citet{Jochmans18} in dyadic link formation models. 
Compactness of regressors in Assumption \ref{a:compact_and_interior} is not essential and can be exchanged with appropriate assumptions on the moments of regressors.
Assumption \ref{a:random_sampling_nodes} is standard.
Assumption \ref{a:identification} is a full-rank condition on the Hessian and guarantees that the objective function has a unique optimiser in the limit.

The dyad-specific fixed effects structure admits  more complex forms of heterogeneity. 
This leads to fundamental differences with both dyadic and triadic models that involve node-specific fixed effects only. In particular, under the dyad-specific heterogeneity structure, even when the network is asymptotically non-empty, it can still be \textit{too sparse} in the sense that in the limit one can accumulate infinitely many links but asymptotically zero informative wirings. This is a novel phenomenon, absent from models with dyad-level fixed effects. We formally establish and discuss these points in detail below, and provide a precise measure of when sparse becomes too sparse in terms of link formation probability $\rho_N$.

\subsection{Consistency}\label{sec:consistency}

This subsection presents the first key result of our large-sample analysis: the consistency of the hexad logit estimator under specific conditions on network sparsity.
\begin{theorem}[Consistency]\label{thm:consistency} If Assumptions \ref{a:logit}-\ref{a:identification} hold, and if $\rho_N=O(1/N^{\delta})$ for some $\delta$ with $0\leq \delta < 3/2$ then, $\widehat{\beta} - \beta_0 \to_p 0$ as $N\to \infty$.
\end{theorem}

This theorem formally establishes that for consistency to hold, link formation probability must tend to zero slower than $1/N^{3/2}$. In other words, in the given framework one needs something more than asymptotic non-emptiness of the network, which holds under the weaker condition that $N^2\rho_N$ remains above zero for $N$ sufficiently large. This is a fundamental deviation from the dyadic (and also triadic) model with node-level heterogeneity, as we discuss below. 

The stronger condition on $\delta$ in Theorem \ref{thm:consistency} is essentially caused by the more complex heterogeneity structure admitted by the dyad-specific effects. This translates into more complex informative wirings, involving a greater number of hyperedges and therefore forming with lower probability. Consequently, under the dyad-specific heterogeneity structure one requires more than asymptotic non-emptiness to collect sufficiently many informative wirings, leading to the requirement $0\leq \delta < 3/2$.

To formalise this intuition, we first investigate  the structure of informative wirings for the dyadic and triadic models with node-specific heterogeneity. The dyadic link formation model is given by 
\begin{align}
    Y_{ij}=1\{X_{ij}'\beta_0 + A_i + B_j - \varepsilon_{ij} \geq 0\},
    \label{eq:dyadicmodel}
\end{align}
which has already been considered in the literature. In particular, we know that in the bipartite case informative wirings consist of two edges only (\citealt{Charbonneau17} and \citealt{Jochmans18}). The triadic link formation model with node-specific effects is given by
\begin{align}
    Y_{ijk}=1\{X_{ijk}'\beta_0 + A_i + B_j + C_k - \varepsilon_{ijk} \geq 0\}.
    \label{eq:triadicmodel2}
\end{align}
We develop the theory for this variant in Appendix \ref{sec:triadicno} and
find that there are four informative wirings for this model, each of which consists of two hyperedges (see Section \ref{sec:suffnofe} and Figure \ref{fig:alternative_subhypergraphs_111111}).

As opposed to both these models, we have already seen that in the triadic model with dyad-effects, informative wirings consist of four hyperedges. Wirings with a greater number of hyperedges form with (much) lower probability as the probability of wiring formation decreases exponentially with the number of hyperedges involved in the wiring. This is why the current model requires more than an asymptotically non-empty network: informative wirings accumulate relatively slowly so something denser than a barely non-empty network is required. 

The reason for the specific threshold value of $3/2$ for $\delta$ is a novel pathological scenario under the triadic model with dyad-level heterogeneity: unless $\delta < 3/2$, as $N\to\infty$ the network will not produce infinitely many informative wirings in the limit even though there will be asymptotically infinitely many hyperedges. In fact, under $\delta >3/2$ there will be asymptotically zero informative wirings. This clearly is a problematic scenario, as collection of informative wirings is essential to convergence. Furthermore, this is a peculiarity of the model in \eqref{eq:ofe}: as it turns out,  the dyadic and triadic models with node-level heterogeneity automatically avoid such a scenario.
 
\begin{table}[t!]
\centering
\renewcommand{\arraystretch}{1.3}
\begin{tabular}{|l|c|c|c|}
\hline
 & \textbf{Dyadic} & \textbf{Triadic} & \textbf{Triadic} \\
\hline
Heterogeneity      & $A_{i} + B_{j}$ & $A_{i} + B_{j} + C_{k}$ & $A_{ij} + B_{jk} + C_{ik}$ \\
Link probability   & $\mathrm{P}(Y_{ij}=1)=\widecheck{\rho}_N$ & $\mathrm{P}(Y_{ijk}=1)=\widetilde{\rho}_N$ & $\mathrm{P}(Y_{ijk}=1)=\rho_N$ \\
Average node degree & $N\widecheck{\rho}_N$ & $N^2\widetilde{\rho}_N$ & $N^2\rho_N$ \\
\# links & $N^2\widecheck{\rho}_N$ & $N^3\widetilde{\rho}_N$ & $N^3\rho_N$ \\
\# informative wirings & $O(N^4\widecheck{\rho}_N^2)$ & $O(N^6\widetilde{\rho}_N^2)$ & $O(N^6\rho_N^4)$ \\
\hline
\end{tabular}
\caption{Key quantities for the dyadic link formation model (second column), and the triadic link formation model with node-specific effects (third column) or dyad-specific effects (fourth column). The fixed effects structure for each model is stated in the second row. A \emph{link} refers to an edge in dyadic models, and a hyperedge in triadic models. `Link probability' is the unconditional probability of link formation. `Average node degree' is the expected number of links that a generic node will belong to. `\# links' is the expected total number of links in the network.
`\# informative wirings' reports the asymptotic magnitude 
of the expected number of informative wirings in the network.}
\label{tbl:taxonomy}
\end{table}   

To establish these points formally, 
let $\widecheck{\rho}_N$ be the edge formation probability under the dyadic model \eqref{eq:dyadicmodel} and $\widetilde{\rho}_N$ be the hyperedge formation probability under the triadic model \eqref{eq:triadicmodel2}. As defined earlier in \eqref{eq:rhodefn}, $\rho_N$ is the corresponding probability for our main model in \eqref{eq:ofe}. In what follows, a \emph{link} refers to an edge in the dyadic model and a hyperedge in the triadic models. The expected number of links in the network is given by the product of the total number of dyads/triads and the relevant probability of link formation. This is $N^2\widecheck{\rho}_N$ for the dyadic model, $N^3\widetilde{\rho}_N$ for the triadic model with node-specific effects, and $N^3{\rho}_N$ for the triadic model with dyad-specific effects. The expected number of informative wirings, on the other hand, is proportional to the product of the total number of tetrads/hexads in the network and the probability of a generic tetrad/hexad being informative. As explained previously, under node-level heterogeneity an informative wiring consists of two edges/hyperedges. Then, the probability of observing an informative wiring is $O(\widecheck \rho _N^2)$ and $O(\widetilde \rho _N^2)$ under dyadic and triadic models with node-specific effects, respectively. In the triadic model with dyad-level effects, on the other hand, informative wirings form with probability $O(\rho_N^4)$. Noting that a dyadic model has $O(N^4)$ tetrads whereas a triadic model contains $O(N^6)$ hexads finally yields the asymptotic rates $O(N^4\widecheck \rho_N^2)$, $O(N^6\widetilde \rho_N^2)$, and $O(N^6 \rho_N^4)$. These calculations are summarised in Table \ref{tbl:taxonomy}. We also state the average node degree for each setting.

Our calculations now clearly reveal that in the models with node-specific effects, the expected number of informative wirings is proportional to the square of the expected number of links. Therefore, informative wirings accumulate automatically with links: by design, one simply cannot have infinitely many links without also accumulating infinitely many informative wirings in the limit. Under the triadic link formation model with dyad-level effects, however, it is now clear that unless $\delta<3/2$, one will end up in the limit with infinitely many links but zero (or finite) informative wirings, and fail to achieve consistency.\footnote{The importance of having both infinitely many links and infinitely many informative wirings is directly reflected in the proof of Theorem \ref{thm:consistency}. See in particular equation \eqref{eq:disc1} and the surrounding discussion. For consistency to hold, the right hand side of equation \eqref{eq:disc1} must go to zero, and this happens only if the number of both the links and the informative wirings tend to infinity. The corresponding theoretical arguments for the dyadic model and the triadic model with node-level effects can be found in equations \eqref{eq:disc3} and \eqref{eq:disc2}, respectively.} We also confirm that in the models with node-level effects, asymptotic non-emptiness (i.e. average node degree being above 0 for sufficiently large $N$) automatically guarantees accumulation of infinitely many informative wirings as $N\to\infty$. However, for the triadic model with dyad-level effects non-emptiness does not guarantee this.

\subsection{Asymptotic normality}\label{sec:an}

We now turn to the asymptotic distribution of the hexad logit estimator, establishing the conditions for its normality.

\begin{theorem}[Asymptotic normality]\label{thm:an} 
If Assumptions \ref{a:logit}-\ref{a:identification} hold, and if $\rho_N=O(1/N^{\delta})$ for some $\delta$ with $0\leq \delta < 1$, then for any non-random $P\times 1$ vector $c$ we have
\begin{align*}
    \sqrt{N^3 \rho _N}
    \frac
    { c'(\widehat{\beta} - \beta_0)}
    {\sqrt{c'\Gamma_0^{-1} \overline{M}_H \Gamma_0^{-1} c}}
    \to_d \mathcal{N}(0,2^6),
\end{align*}
where $\overline M_H$ is as defined in equation \eqref{eq:MH} in the Appendix.
\end{theorem}

Theorem \ref{thm:an} shows that the hexad logit estimator has the $\sqrt{N^3\rho_N}$ convergence rate. This is in line with the extant results in the literature, where the convergence rate is the root of the expected number of links across the network.

What is different is the stronger requirement $0\leq \delta <1$ on link formation probability. To understand why, first recall that the score is given by
\begin{align*}
    Z_N
    =
    \frac{1}{N^3(N-1)^3}
    \sum_{\sigma \in \Sigma} 
    s_\sigma(\beta_0).
\end{align*} 
This score function is analogous to a U-statistic, as previously noted in the literature for the score functions of dyadic link formation models. This analogy is helpful in obtaining the asymptotic normality of the score function by using results from the U-statistic literature. Importantly, however, under $\delta\geq 1$ the score $Z_N$ becomes akin to a highly degenerate U-statistic---and the degeneracy of a U-statistic typically leads to asymptotic non-normality. 

To establish these points more formally, first consider the variance decomposition of the score, which is central to its asymptotic behaviour:
\begin{align}
    \mathrm{Var}(Z_{N}) &=  
    \sum_{q=0}^6 
    \overline{\mathcal{C}}_{q,N},
    \label{eq:decompprose}
\end{align}
where 
\begin{align*}
    \overline{\mathcal{C}}_{q,N} 
    =
    \frac{1}{N^6(N-1)^6}
    \underset{[\sigma,\sigma']_q}{\sum \sum}
    \mathbb{E} [s_\sigma s_{\sigma'}]
    \qquad
    \text{for }
    q=0,\ldots,6,
\end{align*}
and ${\sum \sum}_{[\sigma,\sigma']_q}$ denotes summation over all hexad pairs $(\sigma,\sigma')$ with $q$ common nodes.
This decomposition is derived and analysed in detail in Section \ref{sect:scorvar} in the Appendix.

That $\overline {\mathcal{C}}_{0,N}=\overline{\mathcal{C}}_{1,N}=\overline{\mathcal{C}}_{2,N}=0$ follows from the conditional independence of triad-specific shocks: hexad-pairs that share less than three common nodes cannot share a common triad and are therefore conditionally independent. 
More importantly, when $\overline{\mathcal{C}}_{3,N}$ is the leading term in \eqref{eq:decompprose}, a valid (triad-level) H\'{a}jek projection exists that is asymptotically equivalent to the score, and has a limiting Normal distribution. This proof approach is typical for U-statistics and is key to obtaining asymptotic normality of $\widehat \beta$.

However, our analysis also reveals that $\overline{\mathcal{C}}_{3,N}$ is not necessarily the leading term of the decomposition. Indeed, by equation \eqref{eq:altdecomp} in the Appendix we have,
\begin{align}
    \mathrm{Var}(Z_{N})
    &=
    \underset{\overline{\mathcal{C}}_{3,N}}
    {\underbrace{O\left( \frac{\rho_N^7}{N^3} \right)}}
    +
    \underset{\overline{\mathcal{C}}_{4,N}}
    {\underbrace{O\left( \frac{\rho_N^7}{N^4} \right)}}
    +
    \underset{\overline{\mathcal{C}}_{5,N}}
    {\underbrace{O\left( \frac{\rho_N^6}{N^5} \right)}}
    +
    \underset{\overline{\mathcal{C}}_{6,N}}
    {\underbrace{O\left( \frac{\rho_N^4}{N^6} \right)}}.
    \label{eq:decompprose2}
\end{align}
It is straightforward to confirm that, as long as $\delta <2$, $\overline{\mathcal{C}}_{4,N}$ and $\overline{\mathcal{C}}_{5,N}$ are always dominated by $\overline{\mathcal{C}}_{3,N}$. However, $\overline{\mathcal{C}}_{6,N}$ dominates $\overline{\mathcal{C}}_{3,N}$ if $\delta > 1$. This would render $Z_N$ analogous to a degenerate U-statistic.\footnote{Typically, a (textbook) U-statistic would be called degenerate if the decomposition component for the covariance between terms with one common index ($\overline{\mathcal{C}}_{1,N}=0$ in our case) is zero. That we are able to obtain asymptotic normality under $\overline{\mathcal{C}}_{1,N}=\overline{\mathcal{C}}_{2,N}=0$ is a consequence of the random shock $\varepsilon_{ijk}$ being conditionally independent across triads. This triadic independence is what obtains asymptotic normality and therefore triadic information corresponds to what \citet{Serfling} calls the ``basic information''. As discussed in detail in Appendix \ref{app:th1}, the appropriate H\'{a}jek projection is therefore on the triadic information and the term $\overline{\mathcal{C}}_{3,N}$ captures the variance of this projection. Therefore, degeneracy in our case corresponds to $\overline{\mathcal{C}}_{3,N}=0$.} This usually leads to non-normal limiting distributions known as \textit{Gaussian chaos}. While it is possible to pin down the resulting distribution, it is practically difficult to utilise, as a simple formula usually does not exist; see, e.g., Chapter 5.5.2 of \citet{Serfling} and Chapter 12.3 of \citet{vdV}. Consequently, we use $0\leq \delta<1$ as a sufficient condition for asymptotic normality.

The \textit{tension} between the number of links and the number of identifying wirings that we discussed in the previous section is also at work here, and rather subtly. To see this, notice one can also write
\begin{align*}
    \overline{\mathcal{C}}_{3,N} 
    = 
    O\left(
        \frac{\rho_N^8}{N^3 \rho_N}
    \right)
    \qquad
    \text{and}
    \qquad
    \overline{\mathcal{C}}_{6,N} 
    = 
    O\left(
        \frac{\rho_N^8}{N^6 \rho_N^4}
    \right),
\end{align*}
which reveals the number of links and the number of informative wirings in the denominators of $\overline{\mathcal{C}}_{3,N}$ and $\overline{\mathcal{C}}_{6,N}$.
For consistency, it was enough that both the number of links, $O(N^3 \rho_N)$, and number of informative wirings $O(N^6\rho_N^4)$, go to infinity. The above display reveals that for non-degeneracy we additionally need the informative wirings to accumulate faster than the number of links; otherwise, $\overline{\mathcal{C}}_{3,N}$ will not be the leading term. Intuitively, if the network is not \textit{too} sparse (that is if $0\leq \delta <1$), the hyperedges we accumulate will be enough to construct an eventually even larger number of informative wirings. However, under $1\leq \delta <3/2$, the network is too sparse for sufficiently fast accumulation of informative hexads.

The dyadic and triadic models with node-level effects have no such degeneracy issue; see, in particular, Sections \ref{sec:dyadicdecomp} and \ref{sec:tsimplean} for derivations of variance decompositions analogous to \eqref{eq:decompprose2} for these two models. From our discussion of Table \ref{tbl:taxonomy} we already know that these models do not admit a scenario where one can have many links but slow accumulation of informative wirings.

\section{Simulation study}
\label{sec:simulation}

In this section, we study the finite sample performance of the hexad logit estimator.
The data generating process follows the model in equation \eqref{eq:ofe}, specified for a scalar covariate $X_{ijk}$ and parameter $\beta_0$:
\begin{equation*}
Y_{ijk} = 1\{X_{ijk}\beta_0 + A_{ij} + B_{jk} + C_{ik} - \varepsilon_{ijk} \geq 0\}.
\end{equation*}
The errors $\varepsilon_{ijk}$ follow a standard logistic distribution, the covariate $X_{ijk} \sim \mathcal{N}(0,1)$ is drawn independently across triads and independently of the errors and fixed effects, and the true parameter is set to $\beta_0 = 1$.

We mimic the simulation design in \citet{Jochmans18} and adapt it to the triadic setting.
To induce degree heterogeneity and dependence driven by dyad-level fixed effects, each node $r \in \{1,\ldots,N\}$ is assigned a tendency parameter
\begin{equation*}
\phi_r = \frac{N-r}{N-1} \in [0,1];
\end{equation*}
so, lower indices imply higher tendency values.
Sparsity is governed by a common scale parameter $c_N$ that depends on the sample size.
We consider three regimes:
\begin{equation*}
\text{dense: } c_N = 0, \quad
\text{log-log: } c_N = \ln\ln N, \quad
\text{log-sqrt: } c_N = \ln\sqrt{N}.
\end{equation*}

\begin{table}[t!]
\centering
\begin{tabular}{lcccccccc}

\multicolumn{9}{c}{Dense Regime ($c_N = 0$)} \\
\midrule
$N$ & $\overline{\widehat{\beta}}$ & SD($\widehat{\beta}$) & $\overline{\text{se}}(\widehat{\beta})$ & RMSE & $\overline{\text{se}}/\text{SD}$ & $C_{90}$ & $C_{95}$ & Power \\
\midrule
20 & 1.00 & 0.04 & 0.06 & 0.04 & 1.43 & 0.977 & 0.994 & 1.00 \\
25 & 1.00 & 0.03 & 0.04 & 0.03 & 1.39 & 0.976 & 0.992 & 1.00 \\
30 & 1.00 & 0.02 & 0.02 & 0.02 & 1.26 & 0.965 & 0.988 & 1.00 \\
40 & 1.00 & 0.01 & 0.01 & 0.01 & 1.18 & 0.952 & 0.986 & 1.00 \\
50 & 1.00 & 0.01 & 0.01 & 0.01 & 1.11 & 0.935 & 0.972 & 1.00 \\
\addlinespace[1.4em]
\multicolumn{9}{c}{Log-Log Regime ($c_N = \ln\ln N$)} \\
\midrule
$N$ & $\overline{\widehat{\beta}}$ & SD($\widehat{\beta}$) & $\overline{\text{se}}(\widehat{\beta})$ & RMSE & $\overline{\text{se}}/\text{SD}$ & $C_{90}$ & $C_{95}$ & Power \\
\midrule
20 & 1.00 & 0.06 & 0.11 & 0.06 & 1.79 & 0.996 & 1.000 & 1.00 \\
25 & 1.00 & 0.04 & 0.06 & 0.04 & 1.61 & 0.988 & 0.999 & 1.00 \\
30 & 1.00 & 0.03 & 0.04 & 0.03 & 1.45 & 0.980 & 0.997 & 1.00 \\
40 & 1.00 & 0.02 & 0.02 & 0.02 & 1.29 & 0.963 & 0.993 & 1.00 \\
50 & 1.00 & 0.01 & 0.01 & 0.01 & 1.16 & 0.942 & 0.972 & 1.00 \\
\addlinespace[1.4em]
\multicolumn{9}{c}{Log-Sqrt Regime ($c_N = \ln\sqrt{N}$)} \\
\midrule
$N$ & $\overline{\widehat{\beta}}$ & SD($\widehat{\beta}$) & $\overline{\text{se}}(\widehat{\beta})$ & RMSE & $\overline{\text{se}}/\text{SD}$ & $C_{90}$ & $C_{95}$ & Power \\
\midrule
20 & 1.02 & 0.10 & 0.23 & 0.10 & 2.23 & 1.000 & 1.000 & 1.00 \\
25 & 1.01 & 0.06 & 0.12 & 0.06 & 1.95 & 0.996 & 1.000 & 1.00 \\
30 & 1.01 & 0.04 & 0.07 & 0.04 & 1.80 & 0.998 & 1.000 & 1.00 \\
40 & 1.00 & 0.02 & 0.04 & 0.02 & 1.58 & 0.994 & 0.999 & 1.00 \\
50 & 1.00 & 0.02 & 0.02 & 0.02 & 1.40 & 0.984 & 0.995 & 1.00 \\

\end{tabular}
\caption{Simulation results across regimes. The true parameter value is $\beta_0=1$ in all cases. The average and standard deviation of $\widehat{\beta}$ across replications are given by $\overline{\widehat{\beta}}$ and  $\mathrm{SD}(\widehat{\beta})$. $\overline{\text{se}}(\widehat{\beta})$ is the Monte Carlo average of estimated standard errors. $\overline{\text{se}}/\text{SD}$ yields the ratio of the average estimated standard error to the empirical standard deviation, whereas RMSE is the root mean square error. The 90\% and 95\% simulation coverage rates are based on the normal approximation and are given by $C_{90}$ and $C_{95}$, respectively. Finally,  Power is the rejection probability of $H_0:\beta_0=0$ at 5\% against the (true) alternative $\beta_0=1$. All results are based on 1000 replications.}
\label{tab:sim_stacked}
\label{tbl:sim}
\end{table}

Dyad-specific fixed effects are then constructed by taking the negative of the average of node tendencies and applying the sparsity parameter:
\begin{equation*}
A_{ij} = -\frac{c_N}{2}(\phi_i + \phi_j), \quad
B_{jk} = -\frac{c_N}{2}(\phi_j + \phi_k), \quad
C_{ik} = -\frac{c_N}{2}(\phi_i + \phi_k).
\end{equation*}
This specification has three key features.
First, the tendency profile induces degree heterogeneity: lower-indexed nodes have higher $\phi_r$ and therefore receive more negative fixed effects. 
Consequently, the lower a node is indexed, the less likely it is to form links.
Second, larger values of the sparsity parameter $c_N$ yield sparser networks by reducing the overall probability of link formation.
Third, setting $c_N = 0$ removes the fixed effects entirely, so link probabilities depend only on $X_{ijk}\beta_0$ and the logistic shock, producing the densest networks in our study.

\begin{figure}[t!]
    \centering
    \begin{subfigure}[t]{0.45\textwidth}
        \centering
        \includegraphics[width=\textwidth]{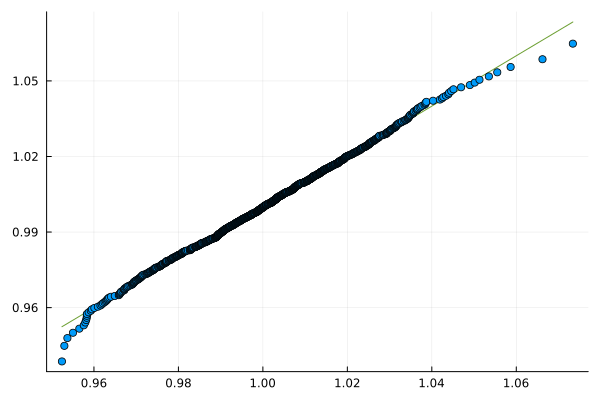}
        \caption*{$c_N = 0$ (Dense)}
    \end{subfigure}
    \hfill
    \begin{subfigure}[t]{0.45\textwidth}
        \centering
        \includegraphics[width=\textwidth]{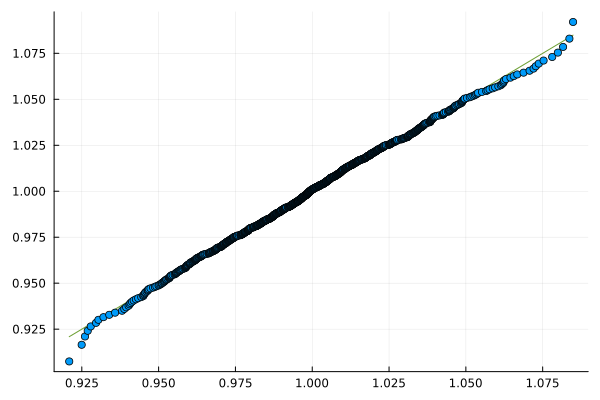}
        \caption*{$c_N = \ln\ln N$}
    \end{subfigure}
    
    \vspace{1em}
    
    \begin{subfigure}[t]{0.45\textwidth}
        \centering
        \includegraphics[width=\textwidth]{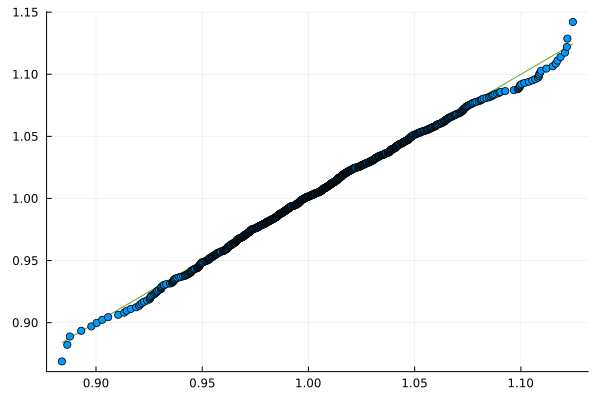}
        \caption*{$c_N = \ln\sqrt{N}$}
    \end{subfigure}
    \hfill
    \begin{subfigure}[t]{0.45\textwidth}
        \centering
        \includegraphics[width=\textwidth]{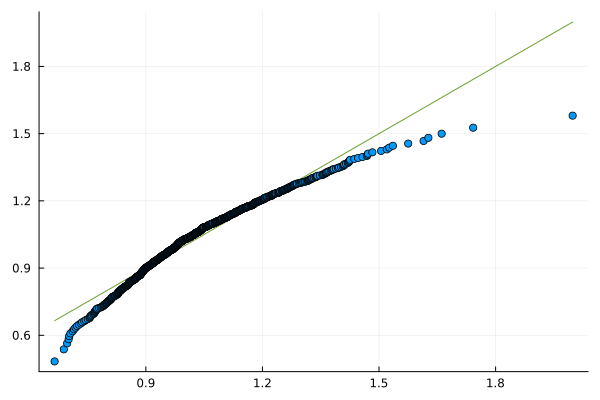}
        \caption*{$c_N = \ln N$}
    \end{subfigure}
    \caption{Q-Q Plots of the Hexad Logit Estimator for $N=30$.
    The plots compare the empirical distribution of $\widehat{\beta}$ from 1000 replications against a Normal distribution with mean and variance given by the Monte Carlo average and variance of $\widehat \beta$.}
    \label{fig:qqplots30}
\end{figure}

We examine $N \in \{20, 25, 30, 40, 50\}$ and the three sparsity regimes defined above.
Each design cell uses $K = 1000$ replications.
Inference is based on a triad-clustered sandwich variance estimator that sums hexad scores over shared triads $(i,j,k)$.
Convergence tolerance is set at $10^{-8}$ with a maximum of 1000 Newton-Raphson iterations.

Table \ref{tbl:sim} reports, for each $(N, c_N)$ design cell, the following: the Monte Carlo mean and standard deviation of $\widehat{\beta}$; the Monte Carlo average of the estimated standard error $\text{se}(\widehat{\beta})$; the root mean squared error relative to $\beta_0 = 1$; the calibration ratio $\overline{\text{se}}/\text{SD}$ comparing the average estimated standard error to the empirical standard deviation; coverage at 90\% and 95\% based on normal approximations; and the rejection probability of $H_0: \beta = 0$ at the 5\% level (power under $\beta_0 = 1$). We also provide Q-Q plots to compare the sample distribution of $\widehat \beta$ to the Normal distribution with mean and variance given by the Monte Carlo average and variance of $\widehat \beta$ (Figures \ref{fig:qqplots30} and \ref{fig:qqplots50}). In these Q-Q plots we also include the results for $c_N = \log N$, which is sparser than the settings introduced already.

\begin{figure}[t!]
    \centering
    \begin{subfigure}[t]{0.45\textwidth}
        \centering
        \includegraphics[width=\textwidth]{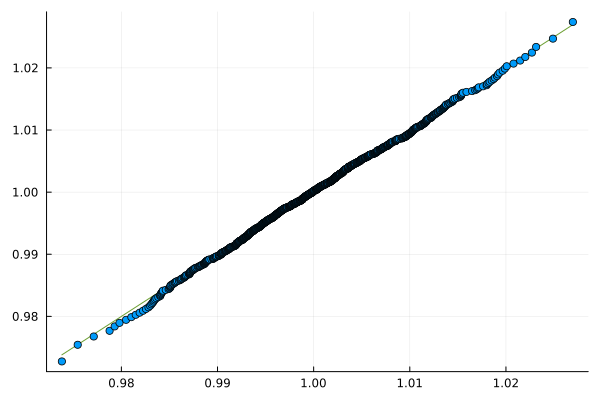}
        \caption*{$c_N = 0$ (Dense)}
    \end{subfigure}
    \hfill
    \begin{subfigure}[t]{0.45\textwidth}
        \centering
        \includegraphics[width=\textwidth]{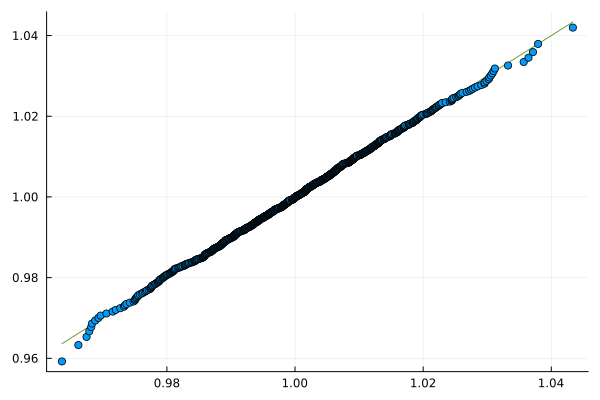}
        \caption*{$c_N = \ln\ln N$}
    \end{subfigure}
    
    \vspace{1em}
    
    \begin{subfigure}[t]{0.45\textwidth}
        \centering
        \includegraphics[width=\textwidth]{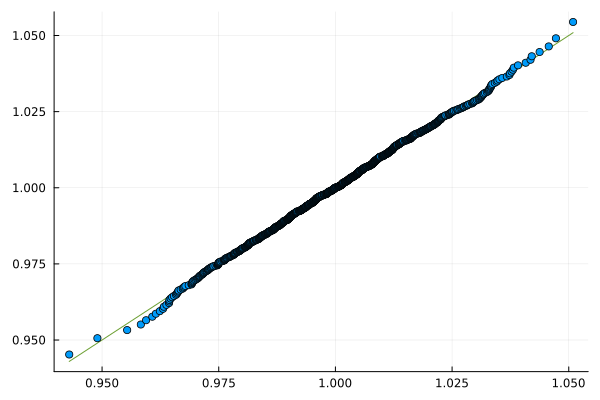}
        \caption*{$c_N = \ln\sqrt{N}$}
    \end{subfigure}
    \hfill
    \begin{subfigure}[t]{0.45\textwidth}
        \centering
        \includegraphics[width=\textwidth]{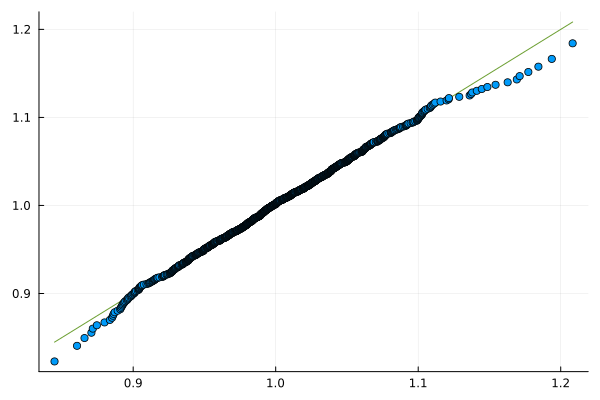}
        \caption*{$c_N = \ln N$}
    \end{subfigure}
    \caption{Q-Q Plots of the Hexad Logit Estimator for $N=50$.
    The plots compare the empirical distribution of $\widehat{\beta}$ from 1000 replications against a Normal distribution with mean and variance given by the Monte Carlo average and variance of $\widehat \beta$.}
    \label{fig:qqplots50}
\end{figure}

Several observations stand out:
First, especially across dense and log-log regimes, sampling distributions are tightly centered at $\beta_0$ with small standard deviations.
As would be expected, dispersion increases as informativeness declines; that is, as $N$ decreases and/or the network becomes sparser. This is most pronounced at smaller sample sizes but attenuates as $N$ reaches 50.
Second, the Q-Q plots confirm that the estimator is generally approximately normally distributed. One exception is the sparsest setting of $c_N=\ln N$, especially at the low sample size of $N=30$. However, we observe that the approximation improves quickly as $N$ increases to 50. This again highlights that information accumulation depends on both the sample size and the level of sparsity. 
Third, estimated standard errors are conservative in finite samples, especially when $N$ is small and the network is more sparse, yielding coverage at or above nominal levels.
However, even under the more sparse settings, $\overline{\text{se}}/\text{SD}$ and the coverage rates improve quickly with $N$.
Fourth, the test of $H_0: \beta = 0$ at the 5\% level has practically unit rejection probability across all reported designs, reflecting strong power against the fixed alternative $\beta_0 = 1$.

All in all, the hexad logit estimator performs well even at moderately large $N$. Although confidence bands are relatively conservative, we note that they become tighter as $N$ increases. When $N$ is in the hundreds, the network will generate a large number of informative wirings (even under more sparse settings) which will lead to tighter confidence bands.

\section{Conclusion}
In this paper we introduced a triadic link formation model with dyad-specific effects and developed a likelihood estimator based on hexad subnetworks---the hexad logit estimator. Our theoretical analysis shows that the estimator is consistent when the unconditional link probability satisfies $\rho_N = O(N^{-\delta})$ with $0 \le \delta < 3/2$. Furthermore, asymptotic normality is obtained with rate $\sqrt{N^3 \rho_N}$ when $0 \le \delta < 1$. 
These requirements on $\rho_N$ go beyond asymptotic non-emptiness because, under dyad-level heterogeneity, accumulation of links need not lead to asymptotic accumulation of informative wirings.
The stated sparsity thresholds provide precise conditions under which regular estimation and inference are attainable. These are novel features specific to the model at hand, which we do not observe in link formation models with node-level heterogeneity. 
We conjecture that models for higher dimensional graphs---which allow for more complicated forms of heterogeneity---will be subject to analogous and more severe sparsity constraints. A more general investigation of this is the subject of ongoing work. 

Our findings also carry implications for applied research. While dyad-level effects offer a richer and often more realistic way to capture heterogeneity, we see that this flexibility imposes a limit on the level of sparsity. This can be an important consideration in granular or disaggregate datasets which are becoming increasingly common in economic analysis. Aggregating data along an appropriate dimension could be a practical first step in such situations.

\newpage

\appendix

\section{Proofs and derivations} \label{app:proofs}

\subsection{Derivations for likelihood, score, and Hessian} \label{app:score_and_Hessian}

We derive the conditional choice probabilities, the per-hexad score, and the Hessian for the composite conditional likelihood.
These expressions are used throughout the paper and proofs.

Recall from \eqref{eq:nol_pcb_overlap} that 
\begin{align*}
    \mathbb E [S_{\sigma,c} | S_\sigma=1, \mathrm{F}_\sigma, X_\sigma]
    =
    p_{\sigma, c}(\beta) 
    &= 
    \frac
    {\exp(W_{\sigma, c}' \beta)}
    {\sum_{c'=1}^2 \exp  \left( W_{\sigma, c'}' \beta \right) }.
\end{align*}
To simplify notation (from time to time) we drop the argument of a function whenever it is evaluated at $\beta_0$ and write $p_{\sigma, c} = p_{\sigma, c}(\beta_0)$, $\overline{W}_\sigma = \overline{W}_{\sigma}(\beta_0)$, and similarly for  $l_\sigma$, $s_{\sigma}$, and $H_{\sigma}$.

We start with the following derivative which will be used in obtaining the score and the Hessian:
\begin{align}
    \frac{\partial p_{\sigma, c}(\beta)}{\partial \beta} 
    &=
    \frac{
      W_{\sigma, c} \exp(W_{\sigma, c}'\beta) 
      \left( \sum_{c'=1}^2 \exp\left(W_{\sigma, c'}' \beta\right) \right) 
    }{
        \left( \sum_{c'=1}^2 \exp\left(W_{\sigma, c'}' \beta\right) \right)^2
    } 
    \notag 
    \\
    & \qquad -
    \frac{
      \exp(W_{\sigma, c}'\beta)
      \left( \sum_{c'=1}^2 W_{\sigma, c'}\exp\left(W_{\sigma, c'}' \beta\right) \right) 
    }{
        \left( \sum_{c'=1}^2 \exp\left(W_{\sigma, c'}' \beta\right) \right)^2
    }
    \notag
    \\
    &= 
    p_{\sigma, c}(\beta) W_{\sigma, c} - p_{\sigma, c}(\beta) \sum_{c' = 1}^2 p_{\sigma, c'}(\beta) W_{\sigma, c'} \notag \\
    &= 
    p_{\sigma, c}(\beta) \left(W_{\sigma, c} - \sum_{c' = 1}^2 p_{\sigma, c'}(\beta) W_{\sigma, c'} \right) \notag \\
    &=
    p_{\sigma, c}(\beta) \left(W_{\sigma, c} - \overline W_\sigma(\beta) \right), \label{eq:dp_db}
\end{align}
which is a standard derivation for conditional logit models (cf. \citealt{CameronTrivedi05}, Section 15.4.2).

Next, recall from equation \eqref{eq:log_likelihood} that hexad $\sigma$'s contribution to the composite log-likelihood is
\begin{align*}
l_\sigma(\beta) 
  &= \sum_{c = 1}^2 S_{\sigma, c} \log p_{\sigma, c}(\beta).
\end{align*}
Then, using \eqref{eq:dp_db} we obtain
\begin{align}
    s_{\sigma}(\beta)
    &= 
    \sum_{c = 1}^2 \frac{S_{\sigma, c}}{p_{\sigma, c}(\beta)} \frac{\partial p_{\sigma, c}(\beta)}{\partial \beta} 
    = 
    \sum_{c = 1}^2 S_{\sigma, c} \left(W_{\sigma, c} - \overline W_\sigma(\beta) \right). 
    \label{eq:score_sigma}
\end{align}
Note that
\begin{align*}
    \mathbb E [s_\sigma | S_\sigma=1, \mathrm{F}_\sigma, X_\sigma]
    &= 
    \sum_{c=1}^2 
    \mathbb E[S_{\sigma, c} | S_\sigma =1, \mathrm{F}_\sigma, X_\sigma]
    \left(W_{\sigma, c} - \overline W_\sigma  \right)
    \\
    &= \sum_{c=1}^2 
    p_{\sigma,c} \left(W_{\sigma, c} - \overline W_\sigma \right) \\
    &= 
    \sum_{c=1}^2 p_{\sigma,c} W_{\sigma, c} - 
    \sum_{c=1}^2 p_{\sigma,c} \overline W_\sigma \\
    &= 0,
\end{align*}
where the first equality follows from \eqref{eq:score_sigma}, the second from the definition of $p_{\sigma,c}$, and the final equality follows because the first term equals $\overline W_\sigma$, $\overline W_\sigma$ in the second term does not depend on $c$, and $\sum_{c=1}^2 p_{\sigma,c}=1$.
By the Law of Iterated Expectations, it also follows that
\begin{align}
    \mathbb E [s_\sigma | \mathrm{F}_\sigma, X_\sigma] = 0.
    \label{eq:zeroscore}
\end{align}
    
To obtain the Hessian contribution, first observe that
\begin{align}
    H_\sigma(\beta)
    =
    \frac{\partial \sum_{c = 1}^2 S_{\sigma, c} \left(W_{\sigma, c} - \overline W_\sigma(\beta) \right)}{\partial \beta'}
    &= 
    - \left( \sum_{c = 1}^2 S_{\sigma, c} \right) \frac{\partial \overline W_\sigma(\beta)}{\partial \beta'}
    = 
    - S_\sigma \frac{\partial \overline W_\sigma(\beta)}{\partial \beta'}
    \label{eq:H}
\end{align}
so the outer summation over $c$ disappears, leaving only the summation implicit in $\overline W_\sigma(\beta)$.
Second,
\begin{align}
    \frac{\partial \overline W_\sigma(\beta)}{\partial \beta'} 
    &= \frac{\partial \sum_{c = 1}^2 p_{\sigma, c}(\beta) W_{\sigma, c}}{\partial \beta'} \notag \\
    &= \sum_{c = 1}^2 p_{\sigma, c}(\beta) W_{\sigma, c} \left(W_{\sigma, c} - \overline W_\sigma(\beta) \right) '
 \notag \\
    &= \sum_{c = 1}^2 p_{\sigma, c}(\beta) \left(W_{\sigma, c} - \overline W_\sigma(\beta) \right) \left(W_{\sigma, c} - \overline W_\sigma(\beta) \right)', \label{eq:H_no_outer}
\end{align}
where we have used that
\begin{align*}
\sum_{c = 1}^2 p_{\sigma, c}(\beta) \overline W_{\sigma}(\beta) \left(W_{\sigma, c} - \overline W_\sigma(\beta) \right)' 
&= 
\sum_{c = 1}^2 p_{\sigma, c}(\beta) \overline{W}_{\sigma}(\beta) W_{\sigma,c}' 
- 
\sum_{c = 1}^2 p_{\sigma, c}(\beta) \overline W_\sigma(\beta) \overline W_{\sigma}'(\beta) \\
&=
\overline W_{\sigma}(\beta) \left(\sum_{c = 1}^2 p_{\sigma, c}(\beta) W_{\sigma, c}'\right)  - \left(\sum_{c = 1}^2 p_{\sigma, c}(\beta)\right) \overline W_\sigma(\beta) \overline W_{\sigma}(\beta)' \\
&=
\overline W_\sigma(\beta) \overline W_{\sigma}(\beta)' - \overline W_\sigma(\beta) \overline W_{\sigma}(\beta)' = 0.
\end{align*}
Combining \eqref{eq:H} and \eqref{eq:H_no_outer} we therefore obtain
\begin{equation}
    H_{\sigma}(\beta) = - S_\sigma \sum_{c = 1}^2 p_{\sigma, c}(\beta)\left(W_{\sigma, c} - \overline W_\sigma(\beta) \right) \left(W_{\sigma, c} - \overline W_\sigma(\beta) \right)'. \label{eq:Hessian_sigma}
\end{equation}

Finally, by inspection of the Hessian in equation \eqref{eq:Hessian_sigma} we can deduce that for any $p=1,\ldots,P$, the third-order derivative component has the form
\begin{align}
    \frac{\partial H_{\sigma}(\beta)}{\partial \beta_p}
    =
    S_{\sigma}
    \sum_{c=1}^2  
    h_{\sigma,c,p}(\beta,W_\sigma)
    \label{eq:third_derivative}
\end{align}
where $\beta_p$ is entry $p$ of the parameter vector $\beta$, whereas $h_{\sigma,c,p}(\cdot,\cdot)$ does not depend on $(S_{\sigma,1},S_{\sigma,2})$. 

\subsection{Useful lemmas}\label{sect:lemma}

In the following, we denote a pair of hexads by $(\sigma,\sigma') \in \Sigma \times \Sigma = \Sigma^2$, 
and let the function $\overline{\mathrm{q}}_i: \Sigma^2 \to \{0,1,2\}$ return the number of nodes that hexads $\sigma$ and $\sigma'$ have in common in part $i$ where $i \in \{1,2,3\}$. To keep the notation concise, we also use the shorthand
\begin{align*}
    \overline {\mathrm q} (\sigma,\sigma') = (\overline {\mathrm q}_1 (\sigma,\sigma'),\overline {\mathrm q}_2 (\sigma,\sigma'),\overline {\mathrm q}_3 (\sigma,\sigma')). 
\end{align*}
Then, $\overline {\mathrm q} (\sigma,\sigma')=(q_1,q_2,q_3)$ denotes that the hexad pair $(\sigma,\sigma')$ has $q_1$ common nodes in part 1, $q_2$ common nodes in part 2, and $q_3$ common nodes in part 3.

Let the number of pairs of hexads with $(q_1, q_2, q_3)$ common nodes in parts 1, 2, and 3 be given by
\begin{align}
    m_{(q_1, q_2, q_3), N} = \sum_{(\sigma,\sigma') \in \Sigma^2} 1 
    \{
        \overline {\mathrm q} (\sigma,\sigma')=(q_1,q_2,q_3)
    \}.
    \label{eq:mbar}
\end{align}
The following lemma establishes the order of $m_{(q_1, q_2, q_3), N}$ as a function of $N$ as the number of nodes diverges.
    
\begin{lemma}\label{lemma:count} For $m_{(q_1, q_2, q_3), N}$ as defined in \eqref{eq:mbar} and for $q_i \in \{0,1,2\}$, we have 
    \begin{align*}
        m_{(q_1, q_2, q_3), N} = O(N^{12-(q_1+q_2+q_3)})
        \qquad
        \text{as }
        N\to\infty.
    \end{align*}
\end{lemma}
\begin{proof} We count the number of hexad pairs $(\sigma,\sigma')$ that satisfy the condition $\overline {\mathrm q} (\sigma,\sigma')=(q_1,q_2,q_3)$. 
Fix the first hexad $\sigma = (i_1, j_1, k_1, i_2, j_2, k_2)$. 
Since the total number of hexads is $|\Sigma| = N^3(N-1)^3 = O(N^6)$, there are $O(N^6)$ ways of doing this.
Next, we consider forming the second hexad $\sigma' = (i'_1, j'_1, k'_1, i'_2, j'_2, k'_2)$ given $\sigma$ such that $\overline {\mathrm q} (\sigma,\sigma')=(q_1,q_2,q_3)$.

For each part $i$, we must select the two nodes $\{i'_1, i'_2\}$:
First, choose which $q_i$ nodes from the set $\{i_1, i_2\}$ used by $\sigma$ will also be used by $\sigma'$. There are $\binom{2}{q_i} = O(1)$ ways to choose these common nodes.
Second, choose the remaining $2-q_i$ nodes for $\sigma'$ from the $N-2$ nodes available in part $i$ that were not used by $\sigma$. There are $\binom{N-2}{2-q_i} = O(N^{2-q_i})$ ways to choose these non-common nodes.
Thus, the number of ways to choose $\{i_1',i_2'\}$ in $\sigma'$ with the specified overlap is $O(1) \times O(N^{2-q_i}) = O(N^{2-q_i})$. Once the specific nodes for part $i$ are chosen, there is a fixed number of ways to assign them to the ordered roles $i'_1, i'_2$ in the definition of $\sigma'$.

The preciding discussion holds for all $i\in \{1,2,3\}.$ It follows that for a given hexad $\sigma$ there are $O(N^{2-q_1})O(N^{2-q_2})O(N^{2-q_3})$ ways of finding another hexad $\sigma'$ such that $\overline {\mathrm q} (\sigma,\sigma')=(q_1,q_2,q_3)$. Since the number of ways to choose the initial $\sigma$ is $O(N^6)$ and the number of choices for $\sigma'$ given $\sigma$ is asymptotically independent of the specific nodes of $\sigma$, the total number of pairs $m_{(q_1, q_2, q_3), N}$ has the following asymptotic order:
\begin{align*}
    m_{(q_1, q_2, q_3), N} &= (\text{Number of ways to choose } \sigma) \times (\text{Number of ways to choose } \sigma' | \sigma) \\
    &= O(N^6) \times O(N^{6-(q_1+q_2+q_3)}) \\
    &= O(N^{12-(q_1+q_2+q_3)}).
\end{align*}
This completes the proof.
\end{proof}

\begin{lemma}\label{lemma:main2}
Let $\ell_\sigma = \ell(X_\sigma, \mathrm{F}_\sigma, \varepsilon_\sigma)$ be a random vector that is bounded for all $\sigma$. Suppose the following properties hold for each $\sigma\in \Sigma$:
(i) $\mathbb{E} [\ell_\sigma] = 0 $; 
(ii) $\ell_\sigma \indep \ell_{\sigma'} | \mathbf{X}, \mathbf{F}$ for all $\sigma'$ such that $\overline{\mathrm{q}}_i(\sigma,\sigma') = 0$ for at least one $i$; 
(iii) $\ell_\sigma \indep \ell_{\sigma'}$ for all $\sigma'$ with $\overline{\mathrm{q}}_1(\sigma,\sigma')=\overline{\mathrm{q}}_2(\sigma,\sigma')=\overline{\mathrm{q}}_3(\sigma,\sigma')=0$; 
(iv) for $S_\sigma$ as defined in \eqref{eq:Stildedefn}, $\ell_\sigma=0$ if $S_{\sigma}=0$. 
Let $\rho_N$ be as defined in equation \eqref{eq:rhodefn}. 

Then, 
\begin{align*}
    \mathbb{E}\left[
    \left(
    \sum_{\sigma \in \Sigma}
    \ell_{\sigma}
    \right)^2
    \right]
    =
    \sum_{q_1=0}^{2}\sum_{q_2=0}^{2}\sum_{q_3=0}^{2}
    \mathrm{C}_{(q_1,q_2,q_3),N},
\end{align*}
where 
\begin{enumerate}
    \item $\mathrm{C}_{(0,0,0),N} = 0 $;
    \item For $(q_1,q_2,q_3)$ with at least one 0 and at least one non-zero $q_i$, $\mathrm{C}_{(q_1,q_2,q_3),N}=O(N^{12-(q_1+q_2+q_3)} \rho_N^8)$;
    \item $\mathrm{C}_{(1,1,1),N} = O(N^9\rho_N^7)$;
    \item 
    $\mathrm{C}_{(2,1,1),N} = \mathrm{C}_{(1,2,1),N} =
    \mathrm{C}_{(1,1,2),N} = O(N^8\rho_N^7)$;
    \item $\mathrm{C}_{(2,2,1),N} = \mathrm{C}_{(1,2,2),N} =
    \mathrm{C}_{(2,1,2),N} = O(N^7\rho_N^6)$;
    \item $\mathrm{C}_{(2,2,2),N} = O(N^6\rho_N^4)$.
\end{enumerate}
If, in addition, $\mathbb{E}[\ell_\sigma | \mathbf{X}, \mathbf{F}] = 0$, then $\mathrm{C}_{(q_1,q_2,q_3),N} = 0$ for all cases where $q_i = 0$ for any $i$. 

We also have
\begin{align*}
    \mathrm{C}_{(1,1,1),N}
    =
    2^6(N(N-1)(N-2))^3
    \mathrm{Cov}(\ell_{((i,j,k),(a,b,c))} , \ell_{((i,j,k),(d,e,f))} ),
\end{align*}
for arbitrary triads $(i,j,k)$, $(a,b,c)$ and $(d,e,f)$ that share no nodes with each other.
\end{lemma}
\begin{proof}
The proof proceeds in two parts. Part I establishes the decomposition and asymptotic rates, while Part II addresses the special case under the additional assumption $\mathbb{E}[\ell_\sigma | \mathbf{X}, \mathbf{F}] = 0$ and derives the specific coefficient for $\mathrm{C}_{(1,1,1),N}$.

\textbf{Part I.} 
We decompose the variance $\mathbb{E}[(\sum_{\sigma \in \Sigma} \ell_{\sigma})^2]$ into components according to the pattern $(q_1,q_2,q_3)$ of common nodes between pairs of hexads $(\sigma, \sigma')$:
\begin{align}
    \mathbb{E}\left[\left(\sum_{\sigma \in \Sigma} \ell_{\sigma}\right)^2\right] 
    &= 
    \sum_{\sigma \in \Sigma} \sum_{\sigma' \in \Sigma} \mathbb{E}[\ell_{\sigma} \ell_{\sigma'}] \notag \\
    &=
    \sum_{q_1=0}^{2}\sum_{q_2=0}^{2}\sum_{q_3=0}^{2} \mathrm{C}_{(q_1,q_2,q_3),N},
    \label{eq:varmdecomp_corrected}
\end{align}
where
\begin{align}
    \mathrm{C}_{(q_1,q_2,q_3),N}
    = 
    \sum_{\sigma \in \Sigma} \sum_{\sigma' \in \Sigma} \mathbb{E}[\ell_{\sigma} \ell_{\sigma'}] \cdot 1\{\overline{\mathrm{q}}_1(\sigma,\sigma')=q_1, \overline{\mathrm{q}}_2(\sigma,\sigma')=q_2, \overline{\mathrm{q}}_3(\sigma,\sigma')=q_3\}.
    \label{eq:defCqN_corrected} 
\end{align}
Importantly, 
recalling that by assumption $\ell_\sigma$ is bounded and $\ell_\sigma = 0$ if $S_\sigma = 0$, we obtain
\begin{align}
    \mathbb{E}[\ell_{\sigma} \ell_{\sigma'}]
    &=
    \mathbb{E}[\mathbb{E}[\ell_{\sigma} \ell_{\sigma'} | S_{\sigma} S_{\sigma'}]] \notag \\
    &=
    \mathbb{E}[\ell_{\sigma} \ell_{\sigma'} | S_{\sigma} S_{\sigma'} = 1] \cdot \mathrm{P}(S_{\sigma} S_{\sigma'} = 1) + \mathbb{E}[\ell_{\sigma} \ell_{\sigma'} | S_{\sigma} S_{\sigma'} = 0] \cdot \mathrm{P}(S_{\sigma} S_{\sigma'} = 0) \notag \\
    &=
    \mathbb{E}[\ell_{\sigma} \ell_{\sigma'} | S_{\sigma} S_{\sigma'} = 1] \cdot \mathrm{P}(S_{\sigma} S_{\sigma'} = 1) \quad (\text{since } \ell_\sigma \ell_{\sigma'}=0 \text{ if } S_\sigma S_{\sigma'}=0) \notag \\
    &=
    O(\mathrm{P}(S_{\sigma} S_{\sigma'} = 1)).
    \label{eq:pdelta1b_corrected} 
\end{align}
This shows that the magnitude of $\mathbb{E}[\ell_{\sigma} \ell_{\sigma'}]$ (and thereby that of $\mathrm{C}_{(q_1,q_2,q_3),N}$) depends on the magnitude of $\mathrm{P}(S_{\sigma} S_{\sigma'} = 1)$, the probability of both $\sigma$ and $\sigma'$ being informative when they exhibit the overlap pattern $(q_1,q_2,q_3)$. This, in turn, depends on the minimum number of distinct hyperedges one needs under this scenario; defining this number as $\Delta$ yields $O(\mathrm{P}(S_{\sigma} S_{\sigma'} = 1)) = O(\rho_N^\Delta)$.\footnote{
We avoid using the more precise notation $\mathrm P(S_\sigma S_{\sigma'}=1; (q_1,q_2,q_3))$ and $\Delta(q_1,q_2,q_3)$ as it would make the presentation cumbersome.}

We next analyse the required number of hyperedges $\Delta$ for each possible overlap pattern $(q_1,q_2,q_3)$ between two hexads $\sigma$ and $\sigma'$:

First, consider the case $q_1=q_2=q_3=0$. In this scenario, the hexads $\sigma$ and $\sigma'$ share no common nodes. It follows from condition (iii) and  $\mathbb{E}[\ell_\sigma]=0$ that $\mathbb{E}[\ell_\sigma \ell_{\sigma'}] = \mathbb{E}[\ell_\sigma]\mathbb{E}[\ell_{\sigma'}] = 0$.

Next, we focus on the scenario where $(q_1,q_2,q_3)$ includes at least one 0 and at least one non-zero $q_i$. In that case, the hexads $\sigma$ and $\sigma'$ cannot have any common hyperedges. Recalling that informative hexads require the presence of four specific hyperedges, we conclude that one requires 4+4=8 hyperedges in order for both $\sigma$ and $\sigma'$ to be informative. Consequently, for this scenario we have $\Delta=8$ leading to $\mathrm{P}(S_{\sigma} S_{\sigma'} = 1) = O(\rho_N^8)$.

The next setting $(q_1,q_2,q_3) = (1,1,1)$ is the first one under which $\sigma$ and $\sigma'$ can share (at most) one hyperedge connecting the common nodes. Once there is a common hyperedge between the two hexads, we will need a minimum of three hyperedges per hexad in order for both hexads to be informative (as informativeness requires four specific hyperedges per hexad). In other words, a minimum of $1+3+3=7$ distinct hyperedges is required for both hexads to be informative. Consequently, $\Delta=7$ and $\mathrm{P}(S_{\sigma} S_{\sigma'} = 1) = O(\rho_N^7)$ in this scenario. 

Consider now the case where $(q_1,q_2,q_3)$ consist of one 2 and two 1s. This makes it possible for the hexads $\sigma$ and $\sigma'$ to share at most two hyperedges. However, due to the structure of the informative wirings, $\sigma$ and $\sigma'$ can share \textit{at most one informative hyperedge} in this setting. To see why, consider, without loss of generality, the example where the first hexad is given by
\begin{align*}
    \sigma=(1,1,1,2,2,2).
\end{align*}
Then, this hexad admits an informative wiring if it contains either the hyperedges $(1,1,1),(1,2,2),(2,1,2),(2,2,1)$ or the hyperedges $(1,1,2),(2,1,1),(1,2,1),(2,2,2)$. The question then is: in the given setting can we think of another hexad $\sigma'$ which contains two hyperedges that appear either in the first wiring or the second? The answer is no: Notice that any pair of hyperedges within either informative wiring shares only one node in common. For such hyperedges to belong to both $\sigma$ and $\sigma'$ we would require, e.g., $q_1=2$, $q_2=2$ and $q_3=1$. Consequently, this setting yields the same $\Delta$ as in the previous case and we again have $\mathrm{P}(S_{\sigma} S_{\sigma'} = 1) = O(\rho_N^7)$.

The next setting where $(q_1,q_2,q_3)$ consist of two 2s and one 1 provides the flexibility to have pairs of hexads $\sigma$ and $\sigma'$ that contain multiple informative hyperedges. Continuing on the example mentioned in the previous part, we could for example have a $\sigma'$ which contains, e.g., the hyperedges $(1,1,1)$ and $(1,2,2)$ of the first informative wiring, or $(1,1,2)$ and $(2,2,2)$ of the second wiring (other selections are possible under this scenario). Therefore, it is possible to have two common informative hyperedges between $\sigma$ and $\sigma'$. Informativeness of both hexads then requires two more (non-common) hyperedges from each of the two hexads. This implies that $\Delta=2+2+2$ and $\mathrm{P}(S_{\sigma} S_{\sigma'} = 1) = O(\rho_N^6)$ in this case.

Finally, $(q_1,q_2,q_3)=(2,2,2)$ simply means that the hexads $\sigma$ and $\sigma'$ are identical and therefore there is only one unique hexad. Since a single hexad requires four specific hyperedges to be informative, in this scenario we have $\Delta=4$ yielding the rate $\mathrm P (S_{\sigma} S_{\sigma'} = 1) = O(\rho_N^4)$.

Finally, combining the above derived probability rates $O(\rho_N^\Delta)$ with the appropriate count coefficient $m_{(q_1,q_2,q_3),N} = O(N^{12-(q_1+q_2+q_3)})$ from Lemma \ref{lemma:count} proves the claim of the first part of the Lemma.

\textbf{Part II.}
Under the additional assumption that $\mathbb{E}[\ell_{\sigma} | \mathbf{X}, \mathbf{F}] = 0$, consider any pair $(\sigma, \sigma')$ where $q_i = 0$ for at least one $i$. In that case, by assumption (ii) of the Lemma, $\ell_\sigma$ and $\ell_{\sigma'}$ are conditionally independent given $\mathbf{X}, \mathbf{F}$. It follows that
\begin{align*}
    \mathbb{E}[\ell_{\sigma}\ell_{\sigma'}]
    &=
    \mathbb{E}[\mathbb{E}[\ell_{\sigma}\ell_{\sigma'} | \mathbf{X}, \mathbf{F}]] \notag \\
    &=
    \mathbb{E}[\mathbb{E}[\ell_{\sigma} | \mathbf{X}, \mathbf{F}] \cdot \mathbb{E}[\ell_{\sigma'} | \mathbf{X}, \mathbf{F}]] \quad (\text{by conditional independence}) \notag \\
    &=
    \mathbb{E}[0 \cdot 0] = 0.
\end{align*}
This implies that $\mathrm{C}_{(q_1,q_2,q_3),N} = 0$ whenever $q_i=0$ for any $i$.

Next, we derive the explicit formula for $\mathrm{C}_{(1,1,1),N}$. This term sums $\mathbb{E}[\ell_{\sigma}\ell_{\sigma'}]$ over pairs $(\sigma, \sigma')$ that share exactly one node in each part. Let such a pair be represented by $\sigma = (i,j,k,a,b,c)$ and $\sigma' = (i,j,k,d,e,f)$, where $(i,j,k)$ are the common nodes, and $a,d \in \{1 \dots N\}\setminus\{i\}$; $b,e \in \{1 \ldots N\}\setminus\{j\}$; $c,f \in \{1 \ldots N\}\setminus\{k\}$; and $a \neq d$, $b \neq e$, $c \neq f$.

Next we establish that there are $2^6[N(N-1)(N-2)]^3 $ ways in which such a set of nodes can be obtained. First, there are $N^3$ choices for the set of common nodes $(i,j,k)$. For each part, this leaves $(N-1)$ ways in which the first non-common node can be chosen, and then $(N-2)$ ways in which the second non-common node can be chosen. Since there are three parts, this yields $(N-1)^3(N-2)^3$ ways in which the non-common nodes can be chosen. Combining this with $N^3$ yields $[N(N-1)(N-2)]^3$ different node selections. Next, although all the nodes have been determined, their ordering within each part is left undetermined. The nodes for each part can appear in one of two orderings; e.g. if the nodes are $i$ and $a$, then the first node can be either $i$ or $a$, yielding two choices. Since there are six parts in total (three for each hexad), we have $2^6$ orderings for each selection of nodes. Consequently, the number of cases with the overlap pattern $(1,1,1)$ is given by $2^6[N(N-1)(N-2)]^3$.

We finally have
\begin{align*}
    \mathrm{C}_{(1,1,1),N}
    &= 
    \sum_{\sigma \in \Sigma} \sum_{\sigma' \in \Sigma} \mathbb{E}[\ell_{\sigma} \ell_{\sigma'}] \cdot 1\{\overline{\mathrm{q}}_1(\sigma,\sigma')=1, \overline{\mathrm{q}}_2(\sigma,\sigma')=1, \overline{\mathrm{q}}_3(\sigma,\sigma')=1\}
    \\
    &=
    2^6 [N(N-1)(N-2)]^3 \times \mathbb{E}[\ell_{((i,j,k),(a,b,c))} \ell_{((i,j,k),(d,e,f))}] \\
    &=
    2^6 [N(N-1)(N-2)]^3 \times \mathrm{Cov}(\ell_{((i,j,k),(a,b,c))}, \ell_{((i,j,k),(d,e,f))}),
\end{align*}
where the last equality uses $\mathbb{E}[\ell_\sigma]=0$. The indices $i,j,k,a,b,c,d,e,f$ represent arbitrary distinct nodes chosen according to the specified overlap structure. This completes the proof.
\end{proof}

\subsection{Sufficiency} 
\label{app:sufficiency_proof}

\begin{proof}[Proof of Theorem \ref{thm:sufficiency_overlapping}]
Let $\sigma = (i_1,j_1,k_1,i_2,j_2,k_2)$ be the generic hexad.
This proof demonstrates that, when conditioning on $S_\sigma = 1$, the fixed effects cancel out, leaving a conditional probability that depends only on the parameter $\beta_0$ and not on $\mathrm{F}_\sigma$.

To that end, define
\begin{align*}
    \kappa _1 
    &= 
    \exp (A_{i_1 j_1}+A_{i_1 j_2}+A_{i_2 j_1}+A_{i_2 j_2}
        + B_{j_1 k_1}+B_{j_1 k_2}+B_{j_2 k_1}+B_{j_2 k_2}
        \\
        & \qquad + C_{i_1 k_1}+C_{i_1 k_2}+C_{i_2 k_1}+C_{i_2 k_2}),
\end{align*}
and
\begin{align*}
    \kappa _2
    &= 
    \prod_{(i,j,k) \in \mathbb I_\sigma} 
    \frac{1}{1 + \exp(A_{ij} + B_{jk} + C_{ik} + X_{ijk}^\prime \beta_0)},
\end{align*}
where $\mathbb{I}_\sigma$ is the set of all triads that can be formed from the nodes constituting the hexad $\sigma$. 
The term $\kappa_2$ represents the product of the denominators from the logistic probability expressions for all possible triads that can be formed within hexad $\sigma$.

By the definition of $S_{\sigma,1}$ and $S_{\sigma,2}$, we have
\begin{align*}
    \mathrm{P} \left( S_{\sigma, 1} = 1 | X_\sigma, \mathrm F_{\sigma}\right) 
    &=
    \kappa _1 \kappa _2 
    \exp\left( \left(X_{i_1 j_1 k_1} + X_{i_1 j_2 k_2} + 
        X_{i_2 j_1 k_2} + X_{i_2 j_2 k_1}\right)^\prime \beta_0 \right)
    \\
    \mathrm{P} \left(  S_{\sigma, 2} = 1 | X_\sigma, \mathrm F_{\sigma}\right) 
    &=
     \kappa _1  \kappa _2 
    \exp\left( \left(X_{i_1 j_1 k_2} + X_{i_1 j_2 k_1} + 
        X_{i_2 j_1 k_1} + X_{i_2 j_2 k_2}\right)^\prime \beta_0 \right).
\end{align*}
Then, noticing that
\begin{align*}
    \frac{\mathrm{P} \left(  S_{\sigma, 2} = 1 | X_\sigma, \mathrm F_{\sigma}\right) }{\mathrm{P} \left(  S_{\sigma, 1} = 1 | X_\sigma, \mathrm F_{\sigma}\right) }
    =
    \exp\left(-  W_{\sigma}^\prime \beta_0\right),    
\end{align*}
we obtain
\begin{align*}
    \mathrm{P} \left(  S_{\sigma, 1} = 1 |  S_{\sigma,1} +  S_{\sigma,2} = 1, X_\sigma, \mathrm F_{\sigma}\right) &=
    \frac{
            \mathrm{P} \left(  S_{\sigma, 1} = 1 | X_\sigma, \mathrm F_{\sigma}\right)
        }
        {
            \mathrm{P} \left(  S_{\sigma, 1} = 1 | X_\sigma, \mathrm F_{\sigma}\right) 
            + 
            \mathrm{P} \left(  S_{\sigma, 2} = 1 | X_\sigma, \mathrm F_{\sigma}\right)
        } \\
    &= \frac{1}{1 + \exp\left(-  W_{\sigma}^\prime \beta_0\right)},
\end{align*}
and the result follows.
\end{proof}

\subsection{Consistency} \label{sec:consistency_proof}

\begin{proof}[Proof of Theorem \ref{thm:consistency}] 
Recalling that $m_N=|\Sigma|=N^3(N-1)^3$ and $p_N=\mathrm{P}(S_\sigma = 1) = O(\rho_N^4)$, we define the scaled sample and population objective functions
        \begin{align*}
        \widehat Q_N(\beta) 
        = 
        \frac{1}{m_N {p}_N} \sum_{\sigma \in \Sigma} l_\sigma(\beta)
        \qquad
        \text{and}
        \qquad
        Q_{N,0}(\beta)
        =
        \frac{1}{m_N p_N} \sum_{\sigma \in \Sigma} 
        \mathbb{E} [l_\sigma(\beta) ].
    \end{align*}
Notice that since 
    $\sum_{\sigma \in \Sigma} \mathbb{E} [l_\sigma(\beta) ] 
    = 
    |\Sigma| \times \mathbb{E} [l_\sigma(\beta) | S_\sigma=1] \times \mathrm P(S_\sigma =1)$,
    the scaling by $m_N p_N$ is the appropriate one, yielding $Q_{N,0}(\beta)=O(1)$. This is crucial for establishing convergence.
    
By equation \eqref{eq:Hessian_sigma}, the sample Hessian for $\widehat Q _N (\beta)$ is given by 
    \begin{align*}
        \frac{1}{m_N p_N} \sum_{\sigma \in \Sigma}
        H_\sigma(\beta)
        =
        -\frac{1}{m_N p_N} \sum_{\sigma \in \Sigma} S_\sigma \sum_{c = 1}^2 p_{\sigma, c}(\beta)\left(W_{\sigma, c} - \overline W_\sigma(\beta) \right) \left(W_{\sigma, c} - \overline W_\sigma(\beta) \right)'.
    \end{align*}
This matrix is negative semi-definite: this is because it is a negative sum of outer products (which are positive semi-definite matrices). This confirms that the sample objective function $\widehat{Q}_N(\beta)$ is concave in $\beta$. By Assumption \ref{a:identification}, the limiting population objective function $Q_0(\beta) = \lim_{N \to \infty} Q_{N,0}(\beta)$ has a negative definite Hessian at $\beta_0$, and is therefore strictly concave in a neighborhood of $\beta_0$. Therefore, $\beta_0$ is the unique maximiser of $Q_0(\beta)$. Establishing pointwise convergence in probability of $\widehat Q_N(\beta)$ to $Q_0(\beta)$ is therefore sufficient to prove consistency. This is what we turn to next.

First, we show that $\widehat{Q}_N(\beta) - Q_{N,0}(\beta) = o_p(1)$. Let
    \begin{align*}
        \widebar l_\sigma(\beta) = l_\sigma(\beta) - \mathbb{E}[l_\sigma(\beta)].
    \end{align*}
By Chebyshev's inequality, for any $\epsilon > 0$ we have
    \begin{align}
        \mathrm{P}\left( \left| \frac{\sum_{\sigma \in \Sigma} 
        \widebar{l}_\sigma(\beta)}
                 {m_N p_N} 
        \right| 
        > \epsilon \right) 
        \leq
        \frac{1}{\epsilon^2}
        \mathbb{E} 
        \left[
        \left(
            \frac{\sum_{\sigma \in \Sigma} \widebar l_\sigma(\beta)}
            {m_N p_N}
        \right)^2
        \right],
        \label{eq:chebyshev}
    \end{align}
where existence of the second-order moment in \eqref{eq:chebyshev} follows from Assumption \ref{a:compact_and_interior}. The random variable $\widebar l_\sigma(\beta)$ satisfies conditions (i), (ii), (iii), and (iv) of Lemma \ref{lemma:main2}.
By an application of Lemma \ref{lemma:main2}, and recalling that $m_N p_N=O(N^6\rho_N^4)$, we then obtain
\begin{align}
        \mathbb{E} 
        \left[
        \left(
            \frac{\sum_{\sigma \in \Sigma} \widebar l_\sigma(\beta)}
            {m_N p_N}
        \right)^2
        \right]
        =&
        O\left(\frac{N^{11}\rho_N^8}{N^{12}\rho_N^8}\right)
        +
        O\left(\frac{N^9\rho_N^7}{N^{12}\rho_N^8}\right)
        \notag
        \\
        &+
        O\left(\frac{N^8\rho_N^7}{N^{12}\rho_N^8}\right)
        +
        O\left(\frac{N^7\rho_N^6}{N^{12}\rho_N^8}\right)
        +
        O\left(\frac{N^6\rho_N^4}{N^{12}\rho_N^8}\right)
        \notag
        \\
        =&
        O\left(\frac{1}{N}\right)
        +
        O\left(\frac{1}{N^3\rho_N}\right)
        \notag
        \\
        &+
        O\left(\frac{1}{N^4\rho_N}\right)
        +
        O\left(\frac{1}{N^5\rho_N^2}\right)
        +
        O\left(\frac{1}{N^6\rho_N^4}\right).
        \label{eq:disc1}
    \end{align}
Under the maintained assumption that $\rho_N = O(1/N^\delta)$ where $0 \leq \delta < 3/2$, all the terms in the above display converge to 0 as $N\to\infty$. Consequently,
    \begin{align*}
    \lim_{N\to\infty}
        \mathrm{P}\left(\left|\frac{\sum_{\sigma \in \Sigma} \widetilde l_\sigma(\beta)}{m_N p_N}\right| > \epsilon\right) = 0,
    \end{align*}
which establishes that $\widehat{Q}_N(\beta) - Q_{N,0}(\beta) = o_p(1)$. Next, by the definition of the limit, $Q_{N,0}(\beta) - Q_0(\beta) = o(1)$. Therefore, we finally obtain
    \begin{align*}
        \widehat{Q}_N(\beta) - Q_0(\beta) = \widehat{Q}_N(\beta) - Q_{N,0}(\beta) + Q_{N,0}(\beta) - Q_0(\beta) = o_p(1),
    \end{align*}
which proves pointwise convergence in probability of $\widehat{Q}_N(\beta)$ to $Q_0(\beta)$. Given the strict concavity of $Q_0(\beta)$ and the fact that $\beta_0$ is its unique maximiser, standard arguments for extremum estimators imply that $\widehat{\beta} \to_p \beta_0$ as $N \to \infty$, completing the proof.
\end{proof}

\subsection{Asymptotic normality}\label{app:th1}

This section contains the proof of Theorem \ref{thm:an}.

\textbf{Notational convention:} In what follows, from time to time we will use an alternate notation for indexing hexad-specific quantities such as the hexad score. To illustrate, consider the score function for the generic hexad $\sigma=(i_1,j_1,k_1,i_2,j_2,k_2)$. In the standard notation this is given by $s_\sigma(\beta)$. In the alternate notation, we write
\begin{align*}
    s_{i_1,i_2|j_1,j_2|k_1,k_2}(\beta).
\end{align*}
This notation makes it easier to visually separate nodes belonging to different parts, which will be convenient in some of our later arguments. 

We use the following shorthand:
$\underset
    {i_1 \neq i_2}
    {\sum \sum}$
denotes 
$\sum_{i_1=1}^N \sum_{\{i_2=1,\ldots,N \,\, : \,\, i_2 \neq i_1\}}$;  $\underset {j_1 \neq j_2} {\sum \sum}$ and $\underset {k_1 \neq k_2} {\sum \sum}$ are defined analogously. Also, $\underset{1\leq i,j,k \leq N}{\sum \sum \sum}$ denotes summation across all values of $i,j,k$.

\subsubsection*{Main steps of the proof} 
The proof is analogous to the corresponding proof for the tetrad logit estimator in \citet{Graham17}, and consists of several steps. Below, we first discuss the main steps of the proof, and then prove these individual steps in Sections \ref{sect:scorvar}, \ref{sect:projection}, \ref{sect:projequiv}, \ref{sect:an1proj} and \ref{sect:hessian}. 
 
Remember that the score function for the conditional likelihood estimator is given by
\begin{align*}
    Z_{N}(\beta)
    &=
    \frac{1}{N^3 (N-1)^3}
    \underset
        {i_1 \neq i_2}
        {\sum \sum}
    \underset
        {j_1 \neq j_2}
        {\sum \sum}
    \underset
        {k_1 \neq k_2}
        {\sum \sum}
    s_{i_1,i_2|j_1,j_2|k_1,k_2}(\beta).
\end{align*}
Sections \ref{sect:scorvar}, \ref{sect:projection} and \ref{sect:projequiv} obtain a valid H\'{a}jek projection for $Z_N$, defined as $Z_N^*$, and prove that $Z_N$ and $Z_N^*$ (when properly scaled) are asymptotically equivalent in the sense that
\begin{align}
\mathbb{E}
\left[ 
\left(
    \sqrt{N^3/\rho_N^7} Z_{N} - \sqrt{N^3/\rho_N^7} Z_{N}^*
\right)^2
\right]    
\to 0,
\qquad
\text{as } N\to\infty.
\label{eq:m1}
\end{align}
It is then proved in Section \ref{sect:an1proj} that
\begin{align}
    \frac{c'\Gamma_0^{-1}}
    {\sqrt{c'\Gamma_0^{-1} \overline M _H \Gamma_0^{-1} c }}
    \sqrt{\frac{N^3}{\rho_N^7}}
    Z_N^*
    &\to_d
    \mathcal{N}(0,2^6),
    \label{eq:m2}
\end{align}
where
\begin{align*}
    \overline M _H 
    = 
    \frac{1}{\rho_N^7 N^3} \underset{1\leq i,j,k \leq N}{\sum \sum \sum} 
    \mathbb{E} [\overline{s}_{ijk} \overline{s}_{ijk}' | \mathbf{X}, \mathbf{F}],
\end{align*}
and $\overline{s}_{ijk}$ is defined in equation \eqref{eq:sijk}. Section \ref{sect:hessian} proves convergence of the (scaled) sample Hessian to $\Gamma_0$; that is,
\begin{align}
    \frac{1}{\rho_N^4}\nabla_{\beta \beta}l_N(\widehat \beta)    
    =
    \frac{1}{\rho_N^4 m_N} \sum_{\sigma \in \Sigma} H_\sigma (\widehat \beta) 
    \to_p 
    \Gamma_0.
    \label{eq:m3}
\end{align}
To bring all these results together, consider now the mean-value expansion
\begin{align*}
    (\widehat \beta - \beta_0)
    &=
    -[\nabla_{\beta \beta} l_N (\overline \beta)]^{-1}Z_N,
\end{align*}
where $\overline \beta$ is some mean value between $\widehat \beta$ and $\beta_0$.
Rescaling and using \eqref{eq:m3} yields
\begin{align*}
    \sqrt{N^3 \rho_N}
    (\widehat \beta - \beta_0)
    &=
    -\left[\frac{1}{\rho_N^4}\nabla_{\beta \beta} l_N (\overline \beta)\right]^{-1} 
    \sqrt{\frac{N^3}{\rho_N^7}}
    Z_N
    =
    -\Gamma_0^{-1} 
    \sqrt{\frac{N^3}{\rho_N^7}}
    Z_N
    +
    o_p(1).
\end{align*}
Combining this with \eqref{eq:m1} and \eqref{eq:m2} finally yields
\begin{align*}
    \sqrt{N^3\rho_N}
    \frac
    { c'(\widehat{\beta} - \beta_0)}
    {\sqrt{c'\Gamma_0^{-1} \overline M _H \Gamma_0^{-1} c}}
    &=
    \frac{-c'\Gamma_0^{-1} }{\sqrt{c'\Gamma_0^{-1} \overline M _H \Gamma_0^{-1} c}}
    \sqrt{\frac{N^3}{\rho_N^7}}
    Z_N^*
    +
    o_p(1)
    \to_d \mathcal{N}(0,2^6).
\end{align*}

\subsubsection{Rate of the variance of the score}\label{sect:scorvar}
First, we note that the score function $s_{\sigma} = s_{\sigma}(\beta_0)$ satisfies assumptions (i)-(iv) of Lemma \ref{lemma:main2}. Furthermore, by \eqref{eq:zeroscore} we also have $\mathbb E [s_\sigma | \mathbf{X},\mathbf{F}]=0$. We can therefore use both parts of Lemma \ref{lemma:main2} to analyse the variance decomposition for the score function
\begin{align*}
    Z_N 
    =
    \frac{1}{N^3(N-1)^3} \sum_{\sigma \in \Sigma} s_\sigma.
\end{align*}
In particular, analogous to the generic definition in \eqref{eq:defCqN_corrected}, let
\begin{align*}
    \mathcal{C}_{(q_1,q_2,q_3),N}
    = 
    \sum_{\sigma \in \Sigma} \sum_{\sigma' \in \Sigma} \mathbb{E}[s_{\sigma} s_{\sigma'}] \cdot 1\{\overline{\mathrm{q}}_1(\sigma,\sigma')=q_1, \overline{\mathrm{q}}_2(\sigma,\sigma')=q_2, \overline{\mathrm{q}}_3(\sigma,\sigma')=q_3\}.
\end{align*}
We then have
\begin{align}
    \mathrm{Var}(Z_{N}) &= \frac{1}{N^6(N-1)^6} \sum_{q_1=0}^{2}\sum_{q_2=0}^{2}\sum_{q_3=0}^{2} \mathcal{C}_{(q_1,q_2,q_3),N},
    \label{eq:vardecomp1}
\end{align}
where, by Lemma \ref{lemma:main2}, 
(i) $\mathcal{C}_{(q_1,q_2,q_3),N}=0$ whenever $q_i=0$ for any $i$; 
(ii)  $\mathcal{C}_{(1,1,1),N}=O(N^9 \rho_N^7)$;
(iii) $\mathcal{C}_{(2,1,1),N}=\mathcal{C}_{(1,2,1),N}=\mathcal{C}_{(1,1,2),N}=O(N^8 \rho_N^7)$; 
(iv) $\mathcal{C}_{(2,2,1),N}=\mathcal{C}_{(1,2,2),N}=\mathcal{C}_{(2,1,2),N}=O(N^7 \rho_N^6)$;
(v) $\mathcal{C}_{(2,2,2),N}=O(N^6\rho_N^4)$. It follows that,
\begin{align}
 \mathrm{Var}(Z_{N}) &= \frac{1}{O(N^{12})} 
 \left[ 
     O(N^9 \rho_N^7) 
     + {O(N^8 \rho_N^7)} 
     + {O(N^7 \rho_N^6)} 
     + O(N^6 \rho_N^4)
 \right] 
 \notag
 \\
 &= O\left(\frac{\rho_N^7}{N^3}\right) + O\left(\frac{\rho_N^7}{N^4}\right) + O\left(\frac{\rho_N^6}{N^5}\right) + O\left(\frac{\rho_N^4}{N^6}\right).
 \label{eq:keydecomp}
\end{align}
Under the assumption $\rho_N=O(1/N^\delta)$ with $0\leq \delta <1$, the first term (which is associated with $\mathcal C _{(1,1,1),N}$) provides the leading order contribution, $O(\rho_N^7/N^3)$. Consequently,
\begin{align*}
 \mathrm{Var}\left( \sqrt{\frac{N^3}{\rho_N^7}} Z_{N}\right) 
 &= O(1) + O\left( \frac{1}{N} \right) + O\left( \frac{1}{N^2 \rho_N} \right) + O\left( \frac{1}{N^3 \rho_N^3 } \right).
\end{align*}
As $N\to\infty$, all terms except the leading $O(1)$ term converge to zero under the condition $0\leq \delta <1$.

Using the second part of Lemma \ref{lemma:main2}, and also recalling that $\mathrm{Cov}(s_{(i,j,k,a,b,c)} , s_{(i,j,k,d,e,f)} )$ is $O(\rho_N^7)$, we also have 
\begin{align}
 \mathrm{Var}\left(\sqrt{\frac{N^3}{\rho_N^7}} Z_{N}\right)
 &= 
 \frac{2^6}{\rho_N^7}\mathrm{Cov}(s_{(i,j,k,a,b,c)} , s_{(i,j,k,d,e,f)} )
 +o(1),
 \label{eq:varzn}
\end{align}
where $(i,j,k)$, $(a,b,c)$ and $(d,e,f)$ share no nodes with each other.

Before moving on, we note that we can represent the decomposition for $\mathrm {Var} (Z_N)$ also in terms of components that depend on the number of common nodes between hexad-pairs, irrelevant of which parts these nodes belong to. More specifically,
\begin{align}
    \mathrm{Var}(Z_{N}) &= \frac{1}{N^6(N-1)^6} 
    \sum_{q=0}^6 
    \mathcal{C}_{q,N},
    \label{eq:vardecomp3}
\end{align}
where 
\begin{align*}
    \mathcal{C}_{q,N} 
    =
    \underset{[\sigma,\sigma']_q}{\sum \sum}
    \mathbb{E} [s_\sigma s_{\sigma'}]
    \qquad
    \text{for }
    q=0,\ldots,6,
\end{align*}
and ${\sum \sum}_{[\sigma,\sigma']_q}$ denotes summation over all hexad pairs $(\sigma,\sigma')$ with $q$ common nodes. In other words, $\mathcal C _{q,N}$ is the contribution of the covariances between scores of hexad-pairs that share $q$ nodes, irrelevant of which parts these nodes $q$ belong to. To see why \eqref{eq:vardecomp3} is equivalent to \eqref{eq:vardecomp1}, note the following: first,
\begin{align*}
    \mathcal C _{0,N} = \mathcal C _{1,N} = \mathcal C _{2,N} = 0
\end{align*} 
since these all correspond to variants of $\mathcal C _{(q_1,q_2,q_3),N}$ where $q_i=0$ for at least one of $i=1,2,3$. The correspondences between the remaining terms (excluding $\mathcal C _{(q_1,q_2,q_3),N}$ that are directly equal to zero) are as follows:
\begin{align*}
    \mathcal C _{3,N} &= \mathcal C_{(1,1,1),N} = O(N^9 \rho_N^7) ,
    \\
    \mathcal C _{4,N} &= \mathcal C_{(2,1,1),N} + \mathcal C_{(1,2,1),N} + \mathcal C_{(1,1,2),N} = O(N^8 \rho_N^7) ,
    \\
    \mathcal C _{5,N} &= \mathcal C_{(2,2,1),N} + \mathcal C_{(1,2,2),N} + \mathcal C_{(2,1,2),N} = O(N^7 \rho_N^6) ,
    \\
    \mathcal C _{6,N} &= \mathcal C_{(2,2,2),N} = O(N^6 \rho_N^4).
\end{align*}
Therefore, we have
\begin{align}
    \mathrm{Var}(Z_{N})
    &=
    \frac{1}{N^6(N-1)^6}
    \left(
        \mathcal{C}_{3,N}
        +
        \mathcal{C}_{4,N}
        +
        \mathcal{C}_{5,N}
        +
        \mathcal{C}_{6,N}
    \right)
    \notag
    \\
    &=
    \underset{\frac{\mathcal{C}_{3,N}}{N^6(N-1)^6}}{\underbrace{O\left( \frac{\rho_N^7}{N^3} \right)}}
    +
    \underset{\frac{\mathcal{C}_{4,N}}{N^6(N-1)^6}}
    {\underbrace{O\left( \frac{\rho_N^7}{N^4} \right)}}
    +
    \underset{\frac{\mathcal{C}_{5,N}}{N^6(N-1)^6}}
    {\underbrace{O\left( \frac{\rho_N^6}{N^5} \right)}}
    +
    \underset{\frac{\mathcal{C}_{6,N}}{N^6(N-1)^6}}
    {\underbrace{O\left( \frac{\rho_N^4}{N^6} \right)}},
    \label{eq:altdecomp}
\end{align}
as in \eqref{eq:keydecomp}.

\subsubsection{Projection}\label{sect:projection}

Let $\mathcal I_{ijk} = (X_{ijk},\mathrm{F}_{ijk},\varepsilon_{ijk})$ and define
\begin{align}
    \overline{s}_{ijk} = \mathbb E [s_{a,a'|b,b'|c,c'} | \mathcal{I}_{ijk}]
    \label{eq:sijk}
\end{align}
where, $(a,b,c)$ and $(a',b',c')$ are some triads such that $a\neq a'$, $b\neq b'$ and $c \neq c'$. Notice that $\overline{s}_{ijk}=0$ unless $i$ equals one of $(a,a')$ \textbf{and} $j$ equals one of $(b,b')$ \textbf{and} $k$ equals one of $(c,c')$. To see why, first note that unless this condition holds, none of the random shocks belonging to the hexad $a,a'|b,b'|c,c'$ can be the same as $\varepsilon_{ijk}$.\footnote{These are given by $\varepsilon_{abc}$, $\varepsilon_{a'bc}$, $\varepsilon_{ab'c}$, $\varepsilon_{abc'}$, $\varepsilon_{a'b'c}$, $\varepsilon_{a'bc'}$, $\varepsilon_{ab'c'}$ and $\varepsilon_{a'b'c'}$.} 
Consequently, all these random shocks will be independent of $\varepsilon_{ijk}$
conditional on $(\mathbf{X},\mathbf{F})$. This is a consequence of conditional independence of triad-specific shocks, as implied by Assumption \ref{a:logit}. Then,
\begin{align}
    \mathbb E [s_{a,a'|b,b'|c,c'} | \mathcal{I}_{ijk}]
    &=
    \mathbb E 
    \left[
        \mathbb E [s_{a,a'|b,b'|c,c'} | \mathbf{X}, \mathbf{F}, \varepsilon_{ijk}]
    \,
    |
    \,
    \mathcal I _{ijk}
    \right]
    \notag
    \\
    &=
    \mathbb E 
    \left[
        \mathbb E [s_{a,a'|b,b'|c,c'} | \mathbf{X}, \mathbf{F}]
    \,
    |
    \,
    \mathcal I _{ijk}
    \right]
    \notag
    \\
    &= 0,
    \label{eq:proj0}
\end{align}
where the first equality is an application of the Law of Iterated Expectations, whereas the final result follows from \eqref{eq:zeroscore}.

Next, we obtain the H\'{a}jek projection for $Z_N$.
Consider the class of functions 
\begin{align}
    \underset{1\leq i, j, k \leq N}{ \sum\sum \sum }
    g_{ijk}(\mathcal{I}_{ijk}),
    \label{eq:pclass}
\end{align}
where $g_{ijk}$ are arbitrary measurable functions with finite second moments. We will show that 
\begin{align}
    Z_N^*
    &= 
    \underset{1 \leq i,j,k \leq N}{\sum \sum \sum }
    \mathbb E 
    [Z_N
    \,|\, 
    \mathcal{I}_{ijk}],
    \label{eq:hajek0}
\end{align}
is a valid (H\'{a}jek) projection of $Z_N$ onto the class of functions defined in \eqref{eq:pclass}. For this we have to verify two conditions: (i) that \eqref{eq:hajek0} belongs to the class given in \eqref{eq:pclass}, and (ii) that the projection error is orthogonal to the class of functions given in \eqref{eq:pclass}. The first condition obviously holds. As for the second condition, we first note that
\begin{align*}
    \mathbb E 
    [Z_N
    \,|\, 
    \mathcal{I}_{ijk}]
    &=
    \frac{1}{N^3(N-1)^3}
    \underset
        {i_1 \neq i_2}
        {\sum \sum}
    \underset
        {j_1 \neq j_2}
        {\sum \sum}
    \underset
        {k_1 \neq k_2}
        {\sum \sum}
    \mathbb E [s_{i_1,i_2|j_1,j_2|k_1,k_2}  | \mathcal I _{ijk}]
    =
    \frac{2^3}{N^3}
    \overline{s}_{ijk},
\end{align*}
where the final result follows from the fact that for any given $(i,j,k)$, the event 
$\{i\in \{i_1,i_2\},
\, 
j\in \{j_1,j_2\},
\,
k\in \{k_1,k_2\}\}$
with the condition $i_1\neq i_2, j_1\neq j_2, k_1\neq k_2 $
will occur in $(2(N-1))^3$ ways. It follows that
\begin{align*}
    Z_N^*
    &= 
    \frac{2^3}{N^3}
    \underset{1 \leq i,j,k \leq N}{\sum \sum \sum }
    \overline{s}_{ijk}.
\end{align*}
To prove orthogonality, we will use the following two results: First, for any $(i,j,k)$ we have
\begin{align}
    \mathbb E
    \left[
    Z_N
    g_{ijk}(\mathcal I_{ijk}) 
    \right]  
    &=
    \mathbb E
    \left\{
    \left.
    \mathbb E
    \left[
        Z_N
        g_{ijk}(\mathcal I_{ijk}) 
    \, \right| \,
    \mathcal I_{ijk}
    \right]
    \right\}
    \notag
    \\
    &=
    \mathbb E
    \left\{
    \left.
    \mathbb E
    \left[
        Z_N
    \, \right| \,
    \mathcal I_{ijk}
    \right]
    g_{ijk}(\mathcal I_{ijk}) 
    \right\}
    \notag
    \\
    &=
    \frac{2^3}{N^3}
    \mathbb E
    \left[
    \overline{s}_{ijk}
    g_{ijk}(\mathcal I_{ijk}) 
    \right].
    \label{eq:hajek1}
\end{align}
Second, for any $(i,j,k)$ it also holds that
\begin{align}
    \mathbb E
    \left[
        Z_N^* g_{ijk}(\mathcal I_{ijk}) 
    \right]
    &=
    \frac{2^3}{N^3}
    \underset{1 \leq a,b,c \leq N}{\sum \sum \sum }
        \mathbb E[
        \overline{s}_{abc}
        g_{ijk}(\mathcal I_{ijk}) 
        ]
    \notag
    \\
    &=
    \frac{2^3}{N^3}
    \underset{1 \leq a,b,c \leq N}{\sum \sum \sum }
        \mathbb E\left\{
        \mathbb E\left[
        \left.
        \overline{s}_{abc}
        g_{ijk}(\mathcal I_{ijk}) 
        \, \right| \,
        \mathcal{I}_{ijk}
        \right]
        \right\}
    \notag
    \\
    &=
    \frac{2^3}{N^3}
    \underset{1 \leq a,b,c \leq N}{\sum \sum \sum }
        \mathbb E\left\{
        \mathbb E\left[
        \left.
        \overline{s}_{abc}
        \, \right| \,
        \mathcal{I}_{ijk}
        \right]
        g_{ijk}(\mathcal I_{ijk})
        \right\}
    \notag
    \\
    &=
    \frac{2^3}{N^3}
    \mathbb E
    \left[
    \overline{s}_{ijk}
    g_{ijk}(\mathcal I_{ijk})
    \right],
    \label{eq:hajek2}
\end{align}
since 
$\mathbb E\left[ \left. \overline{s}_{abc} \, \right| \, \mathcal{I}_{ijk} \right]$ 
is equal to
$\overline{s}_{ijk}$ if $(a,b,c)=(i,j,k)$, and to $0$ otherwise as implied by equation \eqref{eq:proj0}.
Using \eqref{eq:hajek1} and \eqref{eq:hajek2} it follows that
\begin{align*}
    \mathbb E
    \left[
    (Z_N - Z_N^*)
    \underset{1\leq i,j,k \leq N}{\sum \sum \sum }
    g_{ijk}(\mathcal I_{ijk}) 
    \right]
    &=
    \underset{1\leq i,j,k \leq N}{\sum \sum \sum }
    \mathbb E
    \left[
    (Z_N - Z_N^*)
    g_{ijk}(\mathcal I_{ijk}) 
    \right]
    =0,
\end{align*}
which proves orthogonality of the projection error. Hence, 
\begin{align}
    Z_N^*
    &= 
    \frac{2^3}{N^3}
    \underset{1 \leq i,j,k \leq N}{\sum \sum \sum }
    \overline{s}_{ijk}
    \label{eq:hajek3}
\end{align}
is a valid H\'{a}jek projection.

\subsubsection{Asymptotic equivalence of the score and its projection}\label{sect:projequiv}
We next show that the score and its H\'{a}jek projection are asymptotically equivalent, in the sense that
\begin{align*}
\mathbb{E}
\left[ 
\left(
    \sqrt{N^3/\rho_N^7} Z_{N} - \sqrt{N^3/\rho_N^7} Z_{N}^*
\right)^2
\right]    
\to 0,
\qquad
\text{as } N\to\infty.
\end{align*}
Define 
\begin{align*}
    \mathrm{Cov}_{(1,1,1),N}
    =
    \mathrm{Cov}(s_{(i,j,k,a,b,c)} , s_{(i,j,k,d,e,f)} ),
\end{align*}
and let $\sigma$ and $\sigma'$ be two hexads that share the triad $(i,j,k)$. 
Then,
for any two such hexads $(\sigma,\sigma')$ we have
\begin{align}
    \mathrm{Cov} _{(1,1,1),N}
    &=
    \mathbb{E}\{ \mathbb{E} [s_\sigma \, s_{\sigma'} | \mathcal{I}_{ijk}] \} 
    =
    \mathbb{E}
    \{ 
    \mathbb{E} [s_{\sigma} | \mathcal{I}_{ijk}]
    \mathbb{E} [s_{\sigma'} | \mathcal{I}_{ijk}]
    \}
    =
    \mathbb{E} [ \overline{s}_{ijk} \, \overline{s}_{ijk} ].
    \label{eq:proj1}
\end{align}
Consequently,
\begin{align}
    \frac{N^3}{\rho_N^7}
    \mathrm{Var}\left( Z_N^* \right)
    &=
    \frac{N^3}{\rho_N^7}
    \mathrm{Var}\left(
        \frac{2^3}{N^3}
        \underset{1\leq i, j, k\leq N}{\sum \sum \sum }
        \overline{s}_{ijk}
    \right)
    \notag
    \\
    &=
    \frac{2^6}{N^3 \rho_N^7}
    \underset{1\leq i, j, k\leq N}{\sum \sum \sum }
    \,
    \mathrm{Var}(\overline{s}_{ijk})
    \notag
    \\
    &=
    \frac{2^6}{N^3 \rho_N^7}
    \underset{1\leq i, j, k\leq N}{\sum \sum \sum }
    \mathbb{E}[ \overline{s}_{ijk}  \, \overline{s}_{ijk} ]
    \notag
    \\
    &=
    \frac{2^6}{N^3 \rho_N^7}
    \underset{1\leq i, j, k\leq N}{\sum \sum \sum }
    \mathrm{Cov} _{(1,1,1),N}
    \notag
    \\
    &=
    \frac{2^6}{\rho_N^7}\mathrm{Cov} _{(1,1,1),N},
    \label{eq:varproject}
\end{align}
where the third equality follows from the fact that for any $\sigma$, $\mathbb{E} [\overline{s}_{ijk} ] = \mathbb{E}\{\mathbb{E} [s_{\sigma} | \mathcal{I}_{ijk} ]\} = \mathbb{E}[s_{\sigma}]=0 $, and the fourth equality follows from \eqref{eq:proj1}.
Next, since $(Z_N-Z_N^*)$ is the projection error, it must be orthogonal to the projection, yielding $\mathbb{E}[(Z_N-Z_N^*)Z_N^*]=0$. It follows that,
\begin{align}
    \mathrm{Cov}(Z_N,Z_N^*)
    &=
    \mathbb{E}[(Z_N-Z_N^*)Z_N^*] + \mathbb{E}[Z_N^* Z_N^*]
    =
    \mathrm{Var}(Z_N^*).
    \label{eq:covproject}
\end{align}
Combining \eqref{eq:varzn}, \eqref{eq:varproject} and \eqref{eq:covproject}, we obtain that as $N\to \infty$
\begin{align*}
    \mathbb{E}
    \left[ 
    \left(
        \sqrt{\frac{N^3}{\rho_N^7}} Z_{N} - \sqrt{\frac{N^3}{\rho_N^7}} Z_{N}^*
    \right)^2
    \right]
    &=
    \frac{N^3}{\rho_N^7}\mathrm{Var}(Z_N) 
    +
    \frac{N^3}{\rho_N^7}\mathrm{Var}(Z_N^*)
    -
    \frac{2N^3}{\rho_N^7}\mathrm{Cov}(Z_N,Z_N^*)
    \notag
    \\
    &=
    2^6 \frac{\mathrm{Cov}_{(1,1,1),N}}{\rho_N^7} + o(1)
    -
    2^6 \frac{\mathrm{Cov}_{(1,1,1),N}}{\rho_N^7}
    \notag
    \\
    &=
    o(1).
\end{align*}
This proves that the normalised score and its H\'{a}jek projection are asymptotically equivalent.

\subsubsection{CLT for the projection}\label{sect:an1proj}

Our argument here is analogous to the CLT argument of \citet{Graham17}, which in turn is based on \citet{Chatterjee06}. We first switch to `triad indexing'. That is, we map the triad indices $(i,j,k)$ to $h=1,\ldots,H$ where $H=N^3$. This relabelling is convenient as $\overline{s}_{ijk}$ is conditionally independent across triads $(i,j,k)$. Switching now to $\overline{s}_h$ we define
\begin{align}
    M_h = \rho_N^{-7} \mathbb{E} [\overline{s}_h \overline{s}_h' | \mathbf{X}, \mathbf{F}],
    \qquad
    \overline M _H = \frac{1}{H} \sum_{h=1}^H M_h,
    \label{eq:MH}
\end{align}
and
\begin{align}
    R_h
    =
    \frac{c'\Gamma_0^{-1} \left(\overline{s}_h / \sqrt{\rho_N^7} \right) }
    {\sqrt{c'\Gamma_0^{-1} \overline M _H \Gamma_0^{-1} c }}.
    \label{eq:defrh}
\end{align}
Let $Y_1,\ldots,Y_H$ be some (conditional on $\mathbf X$ and $\mathbf F$) independently distributed random variables where
\begin{align*}
    Y_h|{\mathbf X , \mathbf F} \sim \mathcal{N} 
    \left( 
    0,
    \frac{c'\Gamma_0^{-1} M _h \Gamma_0^{-1} c}
    {c'\Gamma_0^{-1} \overline M _H \Gamma_0^{-1} c }
    \right).
\end{align*}
This implies that $\frac{1}{\sqrt{H}}\sum_{h=1}^H Y_h \sim \mathcal{N}(0,1)$.
Define 
\begin{align*}
    Z_h 
    = 
    (R_1,\ldots,R_h,Y_{h+1},\ldots,Y_H)
    \qquad
    \text{and}
    \qquad
    Z_h^0 
    = 
    (R_1,\ldots,R_{h-1},0,Y_{h+1},\ldots,Y_H).
\end{align*}
$Z_h^0$ is the same as $Z_h$, except that its $h^{\mathrm{th}}$ entry is replaced by zero. Let $f(x):\mathbb{R}\to\mathbb{R}$ be an arbitrary three-times differentiable function with $\sup_x | \partial^p f(x)/ \partial x^p | \leq B < \infty$ for $p=1,2,3.$ Notice now that for any such function
\begin{align}
    & f\left(\frac{1}{\sqrt{H}} \sum_{h=1}^H R_h \right)
    -
    f\left(\frac{1}{\sqrt{H}} \sum_{h=1}^H Y_h \right)
    \notag
    \\
    &=
    f\left(\frac{1}{\sqrt{H}} (R_1 + Y_2 + Y_3 + \ldots + Y_H) \right)
    -
    f\left(\frac{1}{\sqrt{H}} (Y_1 + Y_2 + Y_3 + \ldots + Y_H) \right)
    \notag
    \\
    & +
    f\left(\frac{1}{\sqrt{H}} (R_1 + R_2 + Y_3 + \ldots + Y_H) \right)
    -
    f\left(\frac{1}{\sqrt{H}} (R_1 + Y_2 + Y_3 + \ldots + Y_H) \right)
    \notag
    \\
    & +
    f\left(\frac{1}{\sqrt{H}} (R_1 + R_2 + R_3 + \ldots + Y_H) \right)
    -
    f\left(\frac{1}{\sqrt{H}} (R_1 + R_2 + Y_3 + \ldots + Y_H) \right)
    \notag
    \\
    & \vdots
    \notag
    \\
    & +
    f\left(\frac{1}{\sqrt{H}} (R_1 + \ldots + R_{H-1} + R_H) \right)
    -
    f\left(\frac{1}{\sqrt{H}} (R_1 + \ldots + R_{H-1} + Y_H) \right)
    \notag
    \\
    &=
    \sum_{h=1}^H 
    \left[
    f\left(\frac{1}{\sqrt{H}} Z_h' \iota \right)
    -
    f\left(\frac{1}{\sqrt{H}} Z_{h-1}' \iota \right)
    \right],
    \label{eq:telescope}
\end{align}
where $\iota$ is an $H\times 1$ vector of ones.
Expanding $f\left( \frac{1}{\sqrt{H}} Z_h'\iota  \right) $ about  $R_h=0$ yields
\begin{align}
    f\left( \frac{1}{\sqrt{H}} Z_h'\iota  \right) 
    =&
    f\left( \frac{1}{\sqrt{H}} (Z_h^0)'\iota  \right)
    +
    \frac{1}{\sqrt{H}} f' \left( \frac{1}{\sqrt{H}} (Z_h^0)'\iota  \right) R_h
    +
    \frac{1}{2H} 
    f'' \left( \frac{1}{\sqrt{H}} (Z_h^0)'\iota  \right) R_h^2
    \notag
    \\
    &+
    \frac{1}{6H^{3/2}} 
    f''' \left( \frac{1}{\sqrt{H}} (\overline{Z}_h^0)'\iota  \right) R_h^3,
    \label{eq:r1}
\end{align}
where $\overline{Z}_h^0$ is some mean value between $Z_h^0$ and $Z_h$. Next, an expansion of $f\left( \frac{1}{\sqrt{H}} Z_{h-1}'\iota  \right) $ about $Y_h=0$ yields
\begin{align}
    f\left( \frac{1}{\sqrt{H}} Z_{h-1}'\iota  \right) 
    =&
    f\left( \frac{1}{\sqrt{H}} (Z_h^0)'\iota  \right)
    +
    \frac{1}{\sqrt{H}} f' \left( \frac{1}{\sqrt{H}} (Z_h^0)'\iota  \right) Y_h
    +
    \frac{1}{2H} 
    f'' \left( \frac{1}{\sqrt{H}} (Z_h^0)'\iota  \right) Y_h^2
    \notag
    \\
    &+
    \frac{1}{6H^{3/2}} 
    f''' \left( \frac{1}{\sqrt{H}} (\overline{Z}_h^0)'\iota  \right) Y_h^3,
    \label{eq:r2}
\end{align}
where $\overline{Z}_h^0$ is again a mean value (not necessarily the same as the one appearing in \eqref{eq:r1}).
By \eqref{eq:r1} and \eqref{eq:r2} we obtain
\begin{align}
    f\left(\frac{1}{\sqrt{H}} Z_h' \iota \right)
    -
    f\left(\frac{1}{\sqrt{H}} Z_{h-1}' \iota \right)
    \leq &
    \frac{1}{\sqrt{H}} f' \left( \frac{1}{\sqrt{H}} (Z_h^0)'\iota  \right) (R_h - Y_h)
    \notag
    \\
    &+
    \frac{1}{2H} 
    f'' \left( \frac{1}{\sqrt{H}} (Z_h^0)'\iota  \right) (R_h^2 - Y_h^2)
    \notag
    \\
    &+
    \frac{ \left( |R_h^3| + |Y_h^3| \right) B }{6H^{3/2}} ,
    \label{eq:r3}
\end{align}
where we have used
$\sup_{x} | f''' \left( x \right) | \leq B$.
Next, we note the following: (i) $Z_h^0$ does not contain $R_h$; (ii) $Y_h$ is, by definition, (conditionally) independent of $R_h$ and
$Z_h^0$; and (iii)
\begin{align*}
    \mathbb{E}[R_h | \mathbf{X}, \mathbf{F}]
    =
    \frac{c'\Gamma_0^{-1} \mathbb{E}[\overline{s}_h | \mathbf{X}, \mathbf{F}] / \sqrt{\rho_N^7}  }
    {\sqrt{c'\Gamma_0^{-1} \overline{M}_H \Gamma_0^{-1} c }}
    =
    0.
\end{align*}
Then,
\begin{align}
    \mathbb{E}
    \left[
        f' \left( \frac{1}{\sqrt{H}} (Z_h^0)'\iota  \right) (R_h - Y_h)
    \right]
    =
    0,
    \label{eq:r4}
\end{align}
where we have also used $\mathbb{E}[Y_h|\mathbf X , \mathbf F]=\mathbb{E}[R_h|\mathbf X , \mathbf F]=0$.
Similarly, and also using the fact that {$\mathbb{E}[Y_h^2|\mathbf X , \mathbf F] = \mathbb{E}[R_h^2|\mathbf X , \mathbf F]$} by design, we obtain
\begin{align}
    \mathbb{E}
    \left[
        f'' \left( \frac{1}{\sqrt{H}} (Z_h^0)'\iota  \right) (R_h^2-Y_h^2)
    \right]
    =
    0.
    \label{eq:r5}
\end{align}
Combining \eqref{eq:r3}, \eqref{eq:r4} and \eqref{eq:r5}, we obtain
\begin{align}
\mathbb{E}
\left[
    f\left(\frac{1}{\sqrt{H}} Z_h' \iota \right)
    -
    f\left(\frac{1}{\sqrt{H}} Z_{h-1}' \iota \right)
\right]
\leq
\frac{  (\sup_h \mathbb{E}[|R_h^3|] + \sup_h \mathbb{E}[|Y_h^3|])  B }{6H^{3/2}} .
\label{eq:r6}
\end{align}
We note that by general properties of normally distributed random variables, $\sup_h \mathbb{E}[|Y_h^3|] < \infty$. Moreover, by Assumption \ref{a:compact_and_interior} we also have $\sup_h \mathbb{E}[|R_h^3|]<\infty$. Then, there exists some $\overline B < \infty$ such that 
$(\sup_h \mathbb{E}[|R_h^3|] + \sup_h \mathbb{E}[|Y_h^3|])  B \leq \overline B$.
From \eqref{eq:telescope} and \eqref{eq:r6} it follows that
\begin{align*}
    \left|
    \mathbb{E}\left[f\left(\frac{1}{\sqrt{H}} \sum_{h=1}^H R_h \right)\right]
    -
    \mathbb{E}\left[f\left(\frac{1}{\sqrt{H}} \sum_{h=1}^H Y_h \right)\right]
    \right|
    &\leq
    \sum_{h=1}^H
    \left|
    \mathbb{E}
    \left[
        f\left(\frac{1}{\sqrt{H}} (Z_h' \iota) \right)
        -
        f\left(\frac{1}{\sqrt{H}} (Z_{h-1}' \iota) \right)
    \right]
    \right|
    \\
    & \leq
    O\left(\frac{1}{\sqrt H }\right) ,
\end{align*}
and consequently,
\begin{align}
    \lim_{H\to\infty}
    \left|
    \mathbb{E}\left[f\left(\frac{1}{\sqrt{H}} \sum_{h=1}^H R_h \right)\right]
    -
    \mathbb{E}\left[f\left(\frac{1}{\sqrt{H}} \sum_{h=1}^H Y_h \right)\right]
    \right|
    =0,
    \label{eq:clt}
\end{align}
for all functions $f(\cdot)$ for which $\sup_{x} | f''' \left( x \right)| <\infty$. Recalling that $\frac{1}{\sqrt{H}} \sum_{h=1}^H Y_h \to_d \mathcal{N}(0,1)$ by design, equation \eqref{eq:clt} finally implies that
\begin{align}
    \frac{1}{\sqrt{H}} \sum_{h=1}^H R_h \to_d \mathcal{N}(0,1).
    \label{eq:cltfin}
\end{align}
Now, for the H\'{a}jek projection in equation \eqref{eq:hajek3} we have 
\begin{align*}
    \frac{1}{\sqrt H}
    \sum_{h=1}^H
    \overline{s}_{h}/\sqrt{\rho_N^7}
    &=
    \frac{1}{2^3}\sqrt{\frac{H}{\rho_N^7}} Z_N^* .
\end{align*}
Then,
\begin{align*}
    \frac{c'\Gamma_0^{-1}}
    {\sqrt{c'\Gamma_0^{-1} \overline M _H \Gamma_0^{-1} c }}
    \frac{1}{2^3}\sqrt{\frac{H}{\rho_N^7}}
    Z_N^*
    &=
    \frac{c'\Gamma_0^{-1}}
    {\sqrt{c'\Gamma_0^{-1} \overline{M}_H \Gamma_0^{-1} c }}
    \left( \frac{1}{\sqrt{H}} \sum_{h=1}^H \overline{s}_h /\sqrt{\rho_N^7}  \right)
    \\
    &=
    \frac{1}{\sqrt{H}}\sum_{h=1}^H R_h
    \\
    &\to_d
    \mathcal{N}(0,1),
\end{align*}
where we have used \eqref{eq:defrh} and \eqref{eq:cltfin}.
It follows that
\begin{align*}
    \frac{c'\Gamma_0^{-1}}
    {\sqrt{c'\Gamma_0^{-1} \overline M _H \Gamma_0^{-1} c }}
    \sqrt{\frac{H}{\rho_N^7}}
    Z_N^*
    &\to_d
    \mathcal{N}(0,2^6).
\end{align*}

\subsubsection{Proof of Hessian convergence}\label{sect:hessian}
In this part we prove that
\begin{align}
    \frac{1}{m_N p_N} \sum_{\sigma \in \Sigma} H_{\sigma}(\widehat \beta) \to_p \Gamma_0,
    \label{eq:hessian1}
\end{align}
where
\begin{align*}
    \Gamma_0 = \lim_{N\to\infty} \frac{1}{m_N p_N} \sum_{\sigma \in \Sigma} \mathbb E [H_\sigma(\beta_0)],
\end{align*}
as defined in Assumption \ref{a:identification}. From \eqref{eq:hessian1} we also obtain
\begin{align*}
    \frac{1}{\rho_N^4 m_N} \sum_{\sigma \in \Sigma} H_\sigma (\widehat \beta) 
    \to_p
    \Gamma_0.
\end{align*}
To see why, notice that 
\begin{align*}
    \frac{1}{\rho_N^4 m_N} \sum_{\sigma \in \Sigma} H_\sigma (\widehat \beta) 
    -
    \frac{1}{p_N m_N} \sum_{\sigma \in \Sigma} H_\sigma (\widehat \beta) 
    &=
    \frac{(1-\rho_N)^4 - 1}{\rho_N^4(1-\rho_N)^4 m_N} \sum_{\sigma \in \Sigma} H_\sigma (\widehat \beta) 
    \\
    &=
    ((1-\rho_N)^4 - 1)
    (\Gamma_0 + o_p(1))
    \\
    &=
    o_p(1),
\end{align*}
where the first equality follows from $p_N=\rho_N^4 (1-\rho_N)^4$ whereas the second equality uses \eqref{eq:hessian1}.

We now proceed to the proof of \eqref{eq:hessian1}.

\textbf{Part I: Pointwise convergence.}
For any fixed $\beta \in \mathcal{B}$, by arguments analogous to those made in the proof of Theorem \ref{thm:consistency}, it can be shown that
\begin{align*}
    \frac{1}{m_N p_N} \sum_{\sigma \in \Sigma} H_\sigma (\beta) 
    -
    \frac{1}{m_N p_N} \sum_{\sigma\in \Sigma} \mathbb E [H_\sigma (\beta)]
    \to_p 0,
\end{align*}
as this involves convergence of a sample average (of $H_\sigma(\beta)$ over $\sigma \in \Sigma$).
Defining,
\begin{align*}
    \Gamma_0(\beta) 
    =
    \lim_{N\to\infty}\frac{1}{m_N p_N} \sum_{\sigma\in \Sigma} \mathbb E [H_\sigma (\beta)],
\end{align*}
we then obtain the pointwise convergence
\begin{align}
    \frac{1}{m_N p_N} \sum_{\sigma \in \Sigma} H_\sigma (\beta) 
    \to_p \Gamma_0(\beta)
    \qquad
    \text{for each }
    \beta \in \mathcal B.
    \label{eq:Hpwise}
\end{align}
Note that $\Gamma_0(\beta_0)=\Gamma_0$.

\textbf{Part II: Stochastic equicontinuity.}
To upgrade the pointwise convergence in \eqref{eq:Hpwise} to uniform convergence over the compact set $\mathcal{B}$, we next establish stochastic equicontinuity. Define,
\begin{align*}
    H_N(\beta) = \frac{1}{m_N}\sum_{\sigma\in \Sigma} H_\sigma(\beta).
\end{align*}
A sufficient condition for stochastic equicontinuity is that the gradient of $p_N^{-1}H_N(\beta)$ with respect to $\beta$ is uniformly bounded in probability, i.e., $\sup_{\beta \in \mathcal{B}} ||p_N^{-1} \nabla_{\beta_p} H_N(\beta) || = O_p(1)$ for each $p$, for a suitable matrix norm.
For any component $p \in \{1, \dots, P\}$ we have
\begin{align*}
    \nabla_{\beta_p} H_N(\beta) = \frac{1}{m_N} \sum_{\sigma \in \Sigma} \nabla_{\beta\beta\beta_p} l_\sigma(\beta);
\end{align*}
in other words, 
$\nabla_{\beta_p} H_N(\beta)$ is a $P\times P$ matrix with its row $p_r$ and column $p_c$ entry given by 
\begin{align*}
    \frac{1}{m_N} \sum_{\sigma \in \Sigma} \frac{\partial^3 l_\sigma(\beta)}{\partial \beta_{p_r} \, \partial \beta_{p_c} \, \partial \beta_{p}}    .
\end{align*}
From equation \eqref{eq:third_derivative}, we know that $\nabla_{\beta\beta\beta_p}l_\sigma(\beta) = S_\sigma \sum_{c=1}^2 h_{\sigma,c,p}(\beta, W_\sigma)$. Under Assumption \ref{a:compact_and_interior}, the function $h_{\sigma,c,p}(\cdot,\cdot)$ is uniformly bounded in Frobenius norm; that is, there exists some $B<\infty$ such that $||h_{\sigma,c,p}(\beta,W_\sigma)||_F \leq B < \infty$. Using the triangle inequality, we obtain
\begin{align}
    \left\| \nabla_{\beta\beta\beta_p}l_N(\beta) \right\|_F 
    &= \left\| \frac{1}{m_N} \sum_{\sigma\in \Sigma} S_\sigma \sum_{c=1}^2 h_{\sigma,c,p}(\beta, W_\sigma) \right\|_F \notag 
    \\
    &\leq \frac{1}{m_N} \sum_{\sigma\in \Sigma} S_\sigma \sum_{c=1}^2 \| h_{\sigma,c,p}(\beta, W_\sigma) \|_F 
    \notag
    \\
    &\leq \frac{2B}{m_N} \sum_{\sigma\in \Sigma} S_\sigma.
    \label{eq:h0}
\end{align}
This bound holds uniformly over $\beta \in \mathcal{B}$. 

Now, notice that the centred object $S_\sigma - \mathbb E [S_\sigma]$ satisfies conditions (i)-(iv) of Lemma \ref{lemma:main2}. 
Then, invoking Lemma \ref{lemma:main2},
\begin{align}
    \frac{1}{(N(N-1))^3} 
    \sum_{\sigma\in \Sigma} 
    (S _\sigma 
    -
    \mathbb E [ S _\sigma ]
    )
    &=
    O_p\left(
        \sqrt{\max\left\{
        \frac{ \rho _N^8}{N},
        \frac{ \rho _N^7}{N^3},
        \frac{ \rho _N^7}{N^4},
        \frac{ \rho _N^6}{N^5},
        \frac{ \rho _N^4}{N^6}
        \right\}}
    \notag
    \right)
    \\
    &=
    O_p\left(
        \sqrt{
        \frac{ \rho _N^8}{N}
        }
    \right),
    \label{eq:h1}
\end{align}
where we have also used the maintained assumption that $ \rho _N = O(1/N^{ \delta})$ with $0\leq \delta <1$.
Notice also that, since $\mathbb E [S_\sigma]= \mathrm{P}(S_\sigma=1)$ and since $\mathrm P (S_\sigma=1)=O(\rho_N^4)$, we have
\begin{align}
    \frac{1}{(N(N-1))^3} 
        \sum_{\sigma\in \Sigma}
            \mathbb E [S_{\sigma}]
    &= O(\rho_N^4).
    \label{eq:h2}
\end{align}
Recalling that $p_N=O(\rho_N^4)$ and $m_N=(N(N-1))^3$, it follows from \eqref{eq:h0}, \eqref{eq:h1} and \eqref{eq:h2} that
\begin{align}
    \sup_{\beta \in \mathcal B}
    \left|\left|
    \frac{1}{ p_N}\nabla_{\beta\beta\beta_p}l_N( \beta)
    \right| \right|_{F}
    &=
    O\left(\frac{ \rho_N^4}{ \rho_N^4}\right)
    +
    O_p\left(\frac{ \rho_N^4}{ \rho_N^4 \sqrt{N}}\right)
    =
    O_p(1).
    \label{eq:h3}
\end{align}
This establishes the necessary condition for stochastic equicontinuity of $p_N^{-1} H_N(\beta)$ on $\mathcal{B}$.
   
\textbf{Part III: Uniform Convergence.}
With pointwise convergence in probability established in Step 1, stochastic equicontinuity established in Step 2, and the compactness of the parameter space $\mathcal{B}$ under Assumption \ref{a:compact_and_interior}, standard uniform convergence theorems (e.g., Lemma 2.8 of \citet{NeweyMcFadden94}) apply. 
This gives the desired uniform convergence in probability,
\begin{align}
    \sup_{\beta \in \mathcal{B}} \left\| 
    \frac{1}{m_N p_N} \sum_{\sigma \in \Sigma} H_\sigma (\beta) 
    - 
    \Gamma_0(\beta) \right\| 
    \to_p 
    0,
    \label{eq:Huniform}
\end{align}    
because $\Gamma_0(\beta)$ is continuous on $\mathcal{B}$.

\textbf{Part IV: Convergence at $\widehat{\beta}$.}
Finally, by $\widehat{\beta} \to_p \beta_0$ (Theorem \ref{thm:consistency}), $\beta_0$ being in the interior of $\mathcal{B}$ (Assumption \ref{a:compact_and_interior}), and the uniform convergence result obtained in \eqref{eq:Huniform}, it follows that
\begin{align*}
    \frac{1}{m_N p_N} \sum_{\sigma \in \Sigma} H_{\sigma}(\widehat \beta) \to_p \Gamma_0(\beta_0) = \Gamma_0,
\end{align*}
as desired.

\newpage

\section{Minimal degree sequence}
\label{sec:simpleds}

Consider the link formation model in \eqref{eq:ofe} with edge-specific fixed effects. Let the hexad degree sequence be as defined in Section \ref{sec:hexads}. In this part we establish the following:
\begin{enumerate}
    \item $(2,2,2,2,2,2)$ is the minimal degree sequence under which we can find informative wirings. By `minimal' we mean that there is no other degree sequence with all node degrees at most two, which admits informative wirings.
    \item The degree sequence $(2,2,2,2,2,2)$ admits a total of eight wirings and only two of them are informative.
\end{enumerate}

We first note that a set of informative wirings should contain exactly the same fixed effects and exactly the same number of times: this ensures that the same fixed effects  appear in all the individual probabilities used in constructing the conditional likelihood function. This, in turn, leads to cancelling-out of all fixed effects in the resulting conditional likelihood function, as desired.

We next remember that in the link formation model \eqref{eq:ofe}, each hyperedge consists of three edges and the fixed effects involved in a hyperedge directly correspond to these edges. For example, the generic hyperedge $[i,j,k]$ involves edges between the node pairs $(i,j)$, $(j,k)$ and $(i,k)$, and therefore contains the edge-specific fixed effects $A_{ij}$, $B_{jk}$ and $C_{ik}$.

Rather than considering different degree sequences and checking whether they admit informative wirings that satisfy the required conditions, we instead start with a generic first hyperedge (which is the minimal requirement for a wiring to exist), and then add the minimal number of  hyperedges and wirings to maintain the condition that
``we should have at least two wirings that contain the same number of fixed effects the same number of times''. This will directly leads us to informative wirings with the lowest degree sequences.

Specifically, we start, without loss of generality, with the hyperedge $[1,1,1]$. This also means that the first triad of the hexad is $(1,1,1)$; that is, a triad consisting of the first node of each part. With only this hyperedge, we obviously have only one wiring yet:
\begin{align*}
    \text{Wiring 1} : \qquad [1,1,1].
\end{align*}
This hyperedge contains the fixed effects $A_{11}$, $B_{11}$ and $C_{11}$. Then any other informative wiring should also contain the same fixed effects. The following set of hyperedges satisfies these requirements:
	\begin{align}
		[1,1,2], \quad [2,1,1] \quad 
		\text{and} \quad [1,2,1].
		\label{eq:hyper2}
	\end{align}
    Obviously, one can also use the hyperedge $[1,1,1]$ but that would simply duplicate the first hyperedge and the first wiring.
	We observe the following:
	\begin{enumerate}
		\item As required, the fixed effects $A_{11}$, $B_{11}$ and $C_{11}$ are contained in the hyperedges in equation \eqref{eq:hyper2}, and only once.
		\item The hyperedges in \eqref{eq:hyper2} introduce the following fixed effects as a by-product: $A_{21}$, $A_{12}$, $B_{12}$, $B_{21}$, $C_{12}$ and $C_{21}$. These fixed effects appear only once.
		\item The hyperedges in \eqref{eq:hyper2} imply the presence of the second triad $(2,2,2)$. This choice of nodes is without loss of generality, and what matters is that we need one more triad that has no common nodes with the original triad $(1,1,1)$. Therefore, the generic hexad under consideration is given by the nodes of $(1,1,1)$ and $(2,2,2)$.
	\end{enumerate}
    At this point we have
    \begin{align*}
        \text{Wiring 1} &: \qquad [1,1,1].
        \\
        \text{Wiring 2} &: \qquad [1,1,2], \quad [2,1,1], \quad [1,2,1].
    \end{align*}
    We note that for the hexad generated by the triads $(1,1,1)$ and $(2,2,2)$ we cannot find any other hyperedge or sets of hyperedges that contain $A_{11}$, $B_{11}$ and $C_{11}$. Therefore, there can be only two informative wirings.

    The hyperedges in \eqref{eq:hyper2} include the fixed effects $A_{21}$, $A_{12}$, $B_{12}$, $B_{21}$, $C_{12}$ and $C_{21}$. These have to be matched in Wiring 1. Therefore, we have to find a new set of hyperedges which contain these fixed effects, but are different from the hyperedges already included in Wiring 2. These are given by
    \begin{align}
        [2,1,2], \quad [2,2,1] \quad 
		\text{and} \quad [1,2,2].
		\label{eq:hyper3}
    \end{align}
    We note the following:
    \begin{enumerate}
        \item All the fixed effects $A_{21}$, $A_{12}$, $B_{12}$, $B_{21}$, $C_{12}$ and $C_{21}$ are contained in the collection of hyperedges in equation \eqref{eq:hyper3} and only once.
        \item\label{e2:p2} The hyperedges in \eqref{eq:hyper3} introduce the following fixed effects as a by-product, which now have to be matched in Wiring 2: $A_{22}$, $B_{22}$, $C_{22}$. This can be achieved by adding the hyperedge $[2,2,2]$ to Wiring 2, which contains all the stated fixed effects but is not one of the hyperedges contained in Wiring 1.
    \end{enumerate}
	Therefore, we finally have
    \begin{align*}
        \text{Wiring 1} &: \qquad [1,1,1], \quad [2,1,2], \quad [2,2,1],  \quad [1,2,2],
        \\
        \text{Wiring 2} &: \qquad [2,2,2], \quad  [1,1,2], \quad [2,1,1], \quad [1,2,1].
    \end{align*}

    Both these wirings contain all the fixed effects involved in the hexad generated by the triads $(1,1,1)$ and $(2,2,2)$. Moreover, both wirings contain these fixed effects only once. Since both wirings contain exactly the same fixed effects, we do not have to generate any further hyperedges. We also note that, both wirings have the degree sequence $(2,2,2,2,2,2)$. We therefore finally conclude that we cannot obtain informative wirings with a lesser degree sequence, and these are the only informative wirings we can generate under this degree sequence. We reached these two wirings by starting from a generic hyperedge $[1,1,1]$. Therefore, this result applies to any hexad: one can simply relabel the two nodes from each part as 1 and 2, and then use the wirings found above. 

    \begin{table}[t!]
    \centering
    \begin{tabular}{cccc}
    \hline
    \textbf{Wiring 1} & \textbf{Wiring 2} & \textbf{Wiring 3} & \textbf{Wiring 4} \\
    \hline
    $[1,1,1]$ & $[2,2,2]$ & $[1,1,1]$ & $[1,1,1]$ \\
    $[2,1,2]$ & $[1,1,2]$ & $[1,1,2]$ & $[1,2,1]$ \\
    $[2,2,1]$ & $[2,1,1]$ & $[2,2,1]$ & $[2,1,2]$ \\
    $[1,2,2]$ & $[1,2,1]$ & $[2,2,2]$ & $[2,2,2]$ \\
    \hline
    \\
    \hline
    \textbf{Wiring 5} & \textbf{Wiring 6} & \textbf{Wiring 7} & \textbf{Wiring 8} \\
    \hline
    $[1,1,1]$ & $[1,1,2]$ & $[1,1,2]$ & $[1,2,1]$ \\
    $[2,1,1]$ & $[2,2,1]$ & $[1,2,1]$ & $[2,1,1]$ \\
    $[1,2,2]$ & $[1,2,2]$ & $[2,1,2]$ & $[1,2,2]$ \\
    $[2,2,2]$ & $[2,1,1]$ & $[2,2,1]$ & $[2,1,2]$ \\
    \end{tabular}
    \caption{Wirings admitted by the degree sequence \((2,2,2,2,2,2)\).}
    \label{tbl:wirings}
    \end{table}
    
    We also note that Wirings 1 and 2 are not the only wirings that one can have under the degree sequence $(2,2,2,2,2,2)$. There are eight possible wirings under this degree sequence, as presented in Table \ref{tbl:wirings}.
    Importantly, it can be confirmed that, apart from Wirings 1 and 2, there are no two wirings that contain the same fixed effects (even if they are allowed to contain the same hyperedges). In other words, the degree sequence $(2,2,2,2,2,2)$ is by itself not sufficient for constructing a conditional likelihood function.

\newpage

\section{Triadic model with node-level heterogeneity}\label{sec:triadicno}

In this part we analyse the triadic model with node-level effects. A large portion of the theory is analogous to that for the dyad-level effects model. Therefore---except when a new argument has to be developed for the model at hand---our treatment here will mostly be a sketch of the underlying steps.

The notation here is also analogous to that for the dyad-level fixed effects model. To avoid confusion, we distinguish key objects that belong to the node-level effects model by the ``$\sim$'' sign, e.g. $\widetilde{\mathrm{F}} _{ijk}$ instead of $\mathrm F _{ijk}$ etc.

All the proofs for this part can be found in Section \ref{sec:nofe_proofs}

\subsection{The model}
Link formation is determined by 
\begin{align}
    Y_{ijk} = 1\{A_{i} + B_{j} + C_{k} + X_{ijk}' \beta_0 - \varepsilon_{ijk} \geq 0\},
    \label{eq:nw_formation_nonoverlapping}
\end{align}
where $\varepsilon_{ijk}$ is logistic and independent across triads conditional on $(X_{ijk},A_i,B_j,C_k)$. We use $\widetilde{\mathrm{F}}_{ijk}=(A_i,B_j,C_k) $ to denote the fixed effects for the ordered triple $(i,j,k)$. We also let $\widetilde{\mathbf{F}}=(\widetilde{\mathrm F} _{ijk})_ {\mathbb T_N}$ whereas $ \mathbb T_N$ is defined as before. The next assumption  formalises the DGP and is analogous to Assumption \ref{a:logit}.
\begin{sectionassumption}\label{a:logit_nofe}
The conditional likelihood of the network $\mathbf Y = \mathbf y$ is
\begin{equation*}
    \mathrm{P} \left(\mathbf Y = \mathbf y | \mathbf X, \widetilde{\mathbf{F}}\right) 
    = 
    \prod_{(i,j,k) \in \mathbb T_N} 
    \mathrm{P} \left(\left. Y_{ijk} = y_{ijk} \right| X_{ijk}, \widetilde{\mathrm{F}}_{ijk} \right),
\end{equation*}
where
\begin{align*}
    \mathrm{P} \left(\left. Y_{ijk} = y \right| X_{ijk}, \widetilde{\mathrm{F}}_{ijk} \right) 
    = 
    \begin{cases}
    \Lambda\left( X_{ijk}' \beta_0 + A_i + B_j + C_k \right) & \text{ if }y = 1, \\
    1-\Lambda\left( X_{ijk}' \beta_0 + A_i + B_j + C_k \right) & \text{ if }y=0.
    \end{cases}
\end{align*}
\end{sectionassumption}

\subsection{Sufficiency}\label{sec:suffnofe}
Let $\overline Y_{abc} = 1 - Y_{abc} $. For a generic hexad $\sigma=(i_1,j_1,k_1,i_2,j_2,k_2)$ we define the indicator functions
\begin{align*}
    \widetilde{S}_{\sigma, 1}
    = 
    Y_{i_1 j_1 k_1}
    Y_{i_2 j_2 k_2}
    \overline Y_{i_1 j_2 k_1}
    \overline Y_{i_2 j_1 k_2}
    \overline Y_{i_2 j_1 k_1}
    \overline Y_{i_1 j_2 k_2}
    \overline Y_{i_1 j_1 k_2}
    \overline Y_{i_2 j_2 k_1},
    \\
    \widetilde{S}_{\sigma, 2}
    = 
    \overline Y_{i_1 j_1 k_1}
    \overline Y_{i_2 j_2 k_2}
    Y_{i_1 j_2 k_1}
    Y_{i_2 j_1 k_2}
    \overline Y_{i_2 j_1 k_1}
    \overline Y_{i_1 j_2 k_2}
    \overline Y_{i_1 j_1 k_2}
    \overline Y_{i_2 j_2 k_1},
    \\
    \widetilde{S}_{\sigma, 3}
    = 
    \overline Y_{i_1 j_1 k_1}
    \overline Y_{i_2 j_2 k_2}
    \overline Y_{i_1 j_2 k_1}
    \overline Y_{i_2 j_1 k_2}
    Y_{i_2 j_1 k_1}
    Y_{i_1 j_2 k_2}
    \overline Y_{i_1 j_1 k_2}
    \overline Y_{i_2 j_2 k_1},
    \\
    \widetilde{S}_{\sigma, 4}
    = 
    \overline Y_{i_1 j_1 k_1}
    \overline Y_{i_2 j_2 k_2}
    \overline Y_{i_1 j_2 k_1}
    \overline Y_{i_2 j_1 k_2}
    \overline Y_{i_2 j_1 k_1}
    \overline Y_{i_1 j_2 k_2}
    Y_{i_1 j_1 k_2}
    Y_{i_2 j_2 k_1}.
\end{align*}
Each of these indicator functions checks whether the hexad $\sigma$ admits one of the wirings that we will use for constructing the conditional likelihood function (see also Figure \ref{fig:alternative_subhypergraphs_111111}). Notice that these are the four possible wirings admitted by the degree sequence $(1,1,1,1,1,1)$.

For a given $\sigma$, at most one of $\widetilde{S}_{\sigma,1},\ldots,\widetilde{S}_{\sigma,4}$ will be equal to one.
Then, whether a given hexad $\sigma$ admits any one of these wirings is indicated by the binary variable
\begin{align}
    \widetilde{S}_{\sigma} 
    = 
    \widetilde{S}_{\sigma, 1} 
    + 
    \widetilde{S}_{\sigma, 2} 
    + 
    \widetilde{S}_{\sigma, 3} 
    + 
    \widetilde{S}_{\sigma, 4}.
    \label{eq:suff1}
\end{align}

\begin{figure}[t]
    \centering
    
    \begin{subfigure}[b]{0.45\linewidth}
        \centering
        \begin{tikzpicture}
        \SetVertexStyle[FillOpacity = 0.5]
        \draw (0.0,0.0) rectangle (6,4);
        \Vertex[x = 1, y = 3, color = blue, label = $i_1$, position = left]{i1}
        \Vertex[x = 1, y = 1, color = blue, label = $i_2$, position = left]{i2}
        \Vertex[x = 3, y = 3, color = green, label = $j_1$, position = above left]{j1}
        \Vertex[x = 3, y = 1, color = green, label = $j_2$, position = above left]{j2}
        \Vertex[x = 5, y = 3, color = red, label = $k_1$, position = above left]{k1}
        \Vertex[x = 5, y = 1, color = red, label = $k_2$, position = above left]{k2}
        
        \Edge(i1)(j1)
        \Edge(j1)(k1)
        \Edge(i2)(j2)
        \Edge(j2)(k2)
        \end{tikzpicture}
        \caption*{$S_{\sigma, 1}=1$}
    \end{subfigure}
    \hfill
    \begin{subfigure}[b]{0.45\linewidth}
        \centering
        \begin{tikzpicture}
        \SetVertexStyle[FillOpacity = 0.5]
        \draw (0.0,0.0) rectangle (6,4);
        \Vertex[x = 1, y = 3, color = blue, label = $i_1$, position = left]{i1}
        \Vertex[x = 1, y = 1, color = blue, label = $i_2$, position = left]{i2}
        \Vertex[x = 3, y = 3, color = green, label = $j_1$, position = left]{j1}
        \Vertex[x = 3, y = 1, color = green, label = $j_2$, position = left]{j2}
        \Vertex[x = 5, y = 3, color = red, label = $k_1$, position = left]{k1}
        \Vertex[x = 5, y = 1, color = red, label = $k_2$, position = left]{k2}
        \Edge(i1)(j2)
        \Edge(j2)(k1)
        \Edge(i2)(j1)
        \Edge(j1)(k2)
        \end{tikzpicture}
        \caption*{$S_{\sigma, 2}=1$}
    \end{subfigure}

    \begin{subfigure}[b]{0.45\linewidth}
        \centering
        \vspace{22pt}

        \begin{tikzpicture}
        \SetVertexStyle[FillOpacity = 0.5]
        \draw (0.0,0.0) rectangle (6,4);
        \Vertex[x = 1, y = 3, color = blue, label = $i_1$, position = left]{i1}
        \Vertex[x = 1, y = 1, color = blue, label = $i_2$, position = left]{i2}
        \Vertex[x = 3, y = 3, color = green, label = $j_1$, position = left]{j1}
        \Vertex[x = 3, y = 1, color = green, label = $j_2$, position = left]{j2}
        \Vertex[x = 5, y = 3, color = red, label = $k_1$, position = above left]{k1}
        \Vertex[x = 5, y = 1, color = red, label = $k_2$, position = above left]{k2}
        \Edge(i1)(j2)
        \Edge(j2)(k2)
        \Edge(i2)(j1)
        \Edge(j1)(k1)
        \end{tikzpicture}
        
        \caption*{$S_{\sigma, 3} = 1$}
    \end{subfigure}
    \hfill
    \begin{subfigure}[b]{0.45\linewidth}
        \centering
        \vspace{22pt}
        
        \begin{tikzpicture}
        \SetVertexStyle[FillOpacity = 0.5]
        
        \draw (0.0,0.0) rectangle (6,4);
        \Vertex[x = 1, y = 3, color = blue, label = $i_1$, position = left]{i1}
        \Vertex[x = 1, y = 1, color = blue, label = $i_2$, position = left]{i2}
        \Vertex[x = 3, y = 3, color = green, label = $j_1$, position = above left]{j1}
        \Vertex[x = 3, y = 1, color = green, label = $j_2$, position = above left]{j2}
        \Vertex[x = 5, y = 3, color = red, label = $k_1$, position = left]{k1}
        \Vertex[x = 5, y = 1, color = red, label = $k_2$, position = left]{k2}
        
        \Edge(i1)(j1)
        \Edge(j1)(k2)
        \Edge(i2)(j2)
        \Edge(j2)(k1)
        \end{tikzpicture}
        \caption*{$S_{\sigma, 4}=1$}
    \end{subfigure}
    \caption{All hexad wirings compatible with degree sequence $(1, 1, 1, 1, 1, 1)$.}
    \label{fig:alternative_subhypergraphs_111111}
\end{figure}
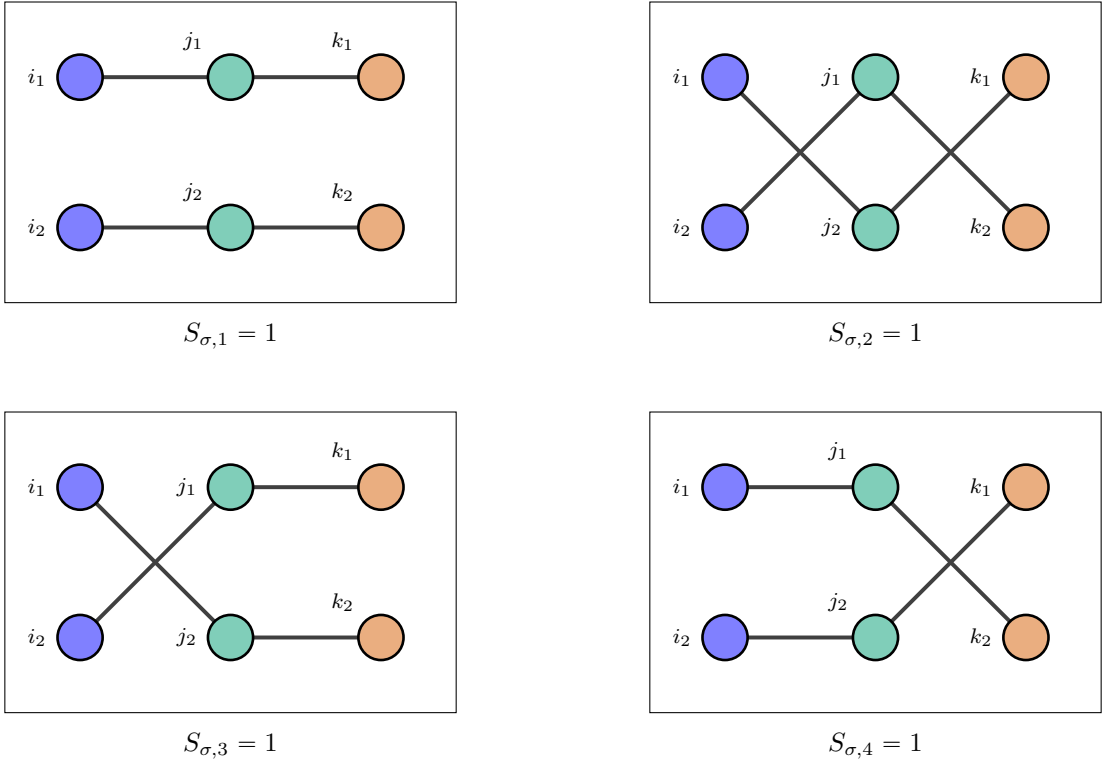

Let $\widetilde{\mathrm F} _\sigma$ and $\mathrm X _\sigma$ be the collections of fixed effects and covariates for the hexad $\sigma$. The next result formally establishes that the degree sequence $(1,1,1,1,1,1)$  allows for the construction of a conditional likelihood function.

\begin{sectiontheorem}\label{thm:sufficiency_nofe} 
Suppose that network formation is determined by Assumption \ref{a:logit_nofe}. Then, for any hexad $\sigma = (i_1,j_1,k_1,i_2,j_2,k_2)$,
\begin{align*}
    \mathrm{P}(\widetilde S _{\sigma, c} = 1 | \widetilde S_{\sigma} = 1, \widetilde{\mathrm{F}}_{\sigma}, X_{\sigma})
    &= 
    \frac
    {\exp(\widetilde{W}_{\sigma, c}' \beta_0)}
    {\sum_{c'=1}^4 \exp  \left( \widetilde{W}_{\sigma, c'}' \beta_0 \right) },
    \qquad
    \text{for } 
    c=1,2,3,4,
\end{align*}
where $\widetilde{W}_{\sigma, 1} = X_{i_1 j_1 k_1}+X_{i_2 j_2 k_2}$,
    $\widetilde{W}_{\sigma, 2} = X_{i_1 j_2 k_1} + X_{i_2 j_1 k_2}$,
    $\widetilde{W}_{\sigma, 3} = X_{i_2 j_1 k_1} + X_{i_1 j_2 k_2}$ and 
    $\widetilde{W}_{\sigma, 4} = X_{i_1 j_1 k_2} + X_{i_2 j_2 k_1}$.
\end{sectiontheorem}

\subsection{Estimation and large sample results}\label{sec:estimation_nofe}
The notation used here mimics the notation for the dyad-specific effects model. In particular we now have the conditional likelihood estimator of $\beta_0$ given by
\begin{align}
    \widehat \beta = \text{argmax}_{\beta \in \mathcal B} \widetilde{l}_N(\beta), \label{eq:CCMLE_nofe}
\end{align}
where
\begin{align*}
    \widetilde{l}_N(\beta) 
    &= 
    \frac{1}{m_N}\sum_{\sigma \in \Sigma} \sum_{c=1}^4 \widetilde{S}_{\sigma, c} \log \widetilde{p}_{\sigma, c}(\beta), 
\end{align*}
and
\begin{align*}
    \widetilde{p}_{\sigma, c}(\beta) 
    &= 
    \frac
    {\exp(\widetilde{W}_{\sigma, c}' \beta)}
    {\sum_{c'=1}^4 \exp  \left( \widetilde{W}_{\sigma, c'}' \beta \right) }.
\end{align*}
The quantity $m_N$ is the same as before as the cardinality of the hexad set is independent of the fixed effect structure.
The hexad-specific likelihood, score and Hessian functions are defined analogously:
\begin{align*}
    \widetilde{l}_{\sigma}(\beta)
    &=
    \sum_{c=1}^4 \widetilde{S}_{\sigma, c} \log \widetilde{p}_{\sigma, c}(\beta), 
    \\
    \widetilde{s}_{\sigma}(\beta) 
    &= 
    \sum_{c = 1}^4 \widetilde{S}_{\sigma, c} \left(\widetilde{W}_{\sigma, c} - \overline{\widetilde{W}}_\sigma(\beta) \right), 
    \\
    \widetilde{H}_{\sigma}(\beta) 
    &= 
    - \widetilde{S}_\sigma \sum_{c = 1}^4 \widetilde{p}_{\sigma, c}(\beta)\left(\widetilde{W}_{\sigma, c} - \overline{\widetilde{W}}_\sigma(\beta) \right) \left(\widetilde{W}_{\sigma, c} - \overline{\widetilde{W}}_\sigma(\beta) \right)',
\end{align*}
where
\begin{align*}
    \overline{\widetilde{W}}_\sigma(\beta) = \sum_{c = 1}^4 \widetilde{p}_{\sigma, c}(\beta) \widetilde{W}_{\sigma, c}.
\end{align*}
We define
\begin{align}
    \widetilde \rho _N = \mathrm P (Y_{ijk}=1),
    \label{eq:rhotilde}
\end{align}
the unconditional link formation probability under the model \eqref{eq:nw_formation_nonoverlapping}. 
Notice that this is not necessarily the same as $\rho_N$. We also define $\widetilde p _N = \mathrm P (\widetilde S _{\sigma}=1)$, the unconditional probability of a hexad being informative in the given setting. Importantly, we now have
\begin{align*}
    \widetilde p _N=O(\widetilde \rho _N^2),
\end{align*}
since informative wirings require formation of two specific hyperedges.

The original Assumptions \ref{a:compact_and_interior} and \ref{a:random_sampling_nodes} are still relevant in the current setting. The next Assumption is a straightforward modification of Assumption \ref{a:identification}.
\setcounter{sectionassumption}{3}
\begin{sectionassumption}\label{a:identification_nofe}
    The matrix
     \begin{align*}
         \widetilde{\Gamma}_0
         =
         - \lim_{N\to\infty}
         \frac{1}{\widetilde{p}_N m_N} 
         \sum_{\sigma \in \Sigma} 
         \mathbb E \left( \widetilde{S}_\sigma \sum_{c = 1}^4 \widetilde{p}_{\sigma, c}(\beta_0)\left(\widetilde{W}_{\sigma, c} - \overline{\widetilde{W}}_\sigma \right) \left(\widetilde{W}_{\sigma, c} - \overline{\widetilde{W}}_\sigma \right)'\right),
     \end{align*}
     exists and is negative definite.
\end{sectionassumption}

\begin{sectionassumption}\label{a:nonempty}
    $N^2 \widetilde{\rho}_N>0$ for large enough $N$.
\end{sectionassumption}
Assumption \ref{a:nonempty} is analogous to Assumption 4.(ii) of \citet{Graham17}, and ensures that the average node degree remains positive; consequently the network remains asymptotically non-empty.

The consistency and asymptotic normality results are presented next.
\begin{sectiontheorem}[Consistency]\label{thm:consistency_nofe} Let Assumptions \ref{a:logit_nofe}, \ref{a:compact_and_interior}, \ref{a:random_sampling_nodes}, \ref{a:identification_nofe} and \ref{a:nonempty} hold. Then, for $\widehat \beta$ as defined in \eqref{eq:CCMLE_nofe} we have
$\widehat{\beta} - \beta_0 \to_p 0$ as $N\to \infty$.
\end{sectiontheorem}
\begin{sectiontheorem}[Asymptotic normality]\label{thm:an_nofe} 
Let Assumptions \ref{a:logit_nofe}, \ref{a:compact_and_interior}, \ref{a:random_sampling_nodes}, \ref{a:identification_nofe} and \ref{a:nonempty} hold. Then, for $\widehat \beta$ as defined in \eqref{eq:CCMLE_nofe} and for any non-random $K\times 1$ vector $c$ we have
\begin{align*}
    \sqrt{N^3 \widetilde{\rho} _N}
    \frac
    { c'(\widehat{\beta} - \beta_0)}
    {\sqrt{c'\widetilde{\Gamma}_0^{-1} \overline{\widetilde{M}}_H \widetilde{\Gamma}_0^{-1} c}}
    \to_d \mathcal{N}(0,2^6),
\end{align*}
where $\overline{\widetilde{M}}_H$ is as defined in equation \eqref{eq:Mnofe}.
\end{sectiontheorem}

\subsection{Proofs}\label{sec:nofe_proofs}

\subsubsection{A preliminary lemma}
We first note that Lemma \ref{lemma:count} is still valid in the current setting since it is independent of the underlying heterogeneity structure.
Below we provide the counterpart of Lemma \ref{lemma:main2} for the present setting with node effects.

\begin{sectionlemma}\label{lemma:main2no}
Let $\ell_\sigma = \ell(X_\sigma, \mathrm{F}_\sigma, \varepsilon_\sigma)$ be a random vector that is bounded for all $\sigma$. Suppose the following properties hold for each $\sigma\in \Sigma$:
(i) $\mathbb{E} [\ell_\sigma] = 0 $; 
(ii) $\ell_\sigma \indep \ell_{\sigma'} | \mathbf{X}, \mathbf{F}$ for all $\sigma'$ such that $\overline{\mathrm{q}}_i(\sigma,\sigma') = 0$ for at least one $i$; 
(iii) $\ell_\sigma \indep \ell_{\sigma'}$ for all $\sigma'$ with $\overline{\mathrm{q}}_1(\sigma,\sigma')=\overline{\mathrm{q}}_2(\sigma,\sigma')=\overline{\mathrm{q}}_3(\sigma,\sigma')=0$; 
(iv) for $\widetilde{S}_\sigma$ as defined in \eqref{eq:suff1}, $\ell_\sigma=0$ if $\widetilde S _{\sigma}=0$. 
Let $\widetilde \rho _N$ be as defined in \eqref{eq:rhotilde}. 

Then, 
\begin{align*}
    \mathbb{E}\left[
    \left(
    \sum_{\sigma \in \Sigma}
    \ell_{\sigma}
    \right)^2
    \right]
    =
    \sum_{q_1=0}^{2}\sum_{q_2=0}^{2}\sum_{q_3=0}^{2}
    \mathrm{C}_{(q_1,q_2,q_3),N},
\end{align*}
where
\begin{enumerate}
    \item $\mathrm{C}_{(0,0,0),N} = 0 $;
    \item For $(q_1,q_2,q_3)$ with at least one 0 and at least one non-zero $q_i$, $\mathrm{C}_{(q_1,q_2,q_3),N}=O(N^{12-(q_1+q_2+q_3)} \widetilde \rho _N^4)$;
    \item $\mathrm{C}_{(1,1,1),N} = O(N^9 \widetilde \rho _N^3)$;
    \item 
    $\mathrm{C}_{(2,1,1),N} = \mathrm{C}_{(1,2,1),N} =
    \mathrm{C}_{(1,1,2),N} = O(N^8 \widetilde \rho _N^3)$;
    \item $\mathrm{C}_{(2,2,1),N} = \mathrm{C}_{(1,2,2),N} =
    \mathrm{C}_{(2,1,2),N} = O(N^7 \widetilde \rho _N^3)$;
    \item $\mathrm{C}_{(2,2,2),N} = O(N^6 \widetilde \rho _N^2)$.
\end{enumerate}
If, in addition, $\mathbb{E}[\ell_\sigma | \mathbf{X}, \mathbf{F}] = 0$, then $\mathrm{C}_{(q_1,q_2,q_3),N} = 0$ whenever $q_i = 0$ for any $i$. 

We also have
\begin{align*}
    \mathrm{C}_{(1,1,1),N}
    =
    2^6(N(N-1)(N-2))^3
    \mathrm{Cov}(\ell_{(i,j,k,a,b,c))} , \ell_{(i,j,k,d,e,f)} ),
\end{align*}
for arbitrary triads $(i,j,k)$, $(a,b,c)$ and $(d,e,f)$ that share no nodes with each other.
\end{sectionlemma}
\begin{proof}
The proof is identical to the proof of Lemma \ref{lemma:main2}, except for the parts that involve the unconditional probability of hyperedge formation, $\widetilde \rho _N$, since in the current setting informative tetrads require formation of two specific hyperedges (as opposed to four in the model with dyad effects). Therefore, we only focus on the parts involving $\widetilde \rho _N$. 

Recall from the proof of Lemma \ref{lemma:main2} that 
\begin{align*}
    \mathbb{E}\left[\left(\sum_{\sigma \in \Sigma} \ell_{\sigma}\right)^2\right] 
    &=
    \sum_{q_1=0}^{2}\sum_{q_2=0}^{2}\sum_{q_3=0}^{2} \mathrm{C}_{(q_1,q_2,q_3),N},
\end{align*}
with
\begin{align*}
    \mathrm{C}_{(q_1,q_2,q_3),N}
    = 
    \sum_{\sigma \in \Sigma} \sum_{\sigma' \in \Sigma} \mathbb{E}[\ell_{\sigma} \ell_{\sigma'}] \cdot 1\{\overline{\mathrm{q}}_1(\sigma,\sigma')=q_1, \overline{\mathrm{q}}_2(\sigma,\sigma')=q_2, \overline{\mathrm{q}}_3(\sigma,\sigma')=q_3\}.
\end{align*}
Analogous to \eqref{eq:pdelta1b_corrected} we also obtain
\begin{align*}
    \mathbb{E}[\ell_{\sigma} \ell_{\sigma'}]
    &=
    O(\mathrm{P}(\widetilde S _{\sigma} \widetilde S _{\sigma'} = 1)).    
\end{align*}
As before, we will determine the minimum number of distinct hyperedges required for both $\sigma$ and $\sigma'$ to be informative. As mentioned at the beginning, in the current setting an informative hexad requires the formation of two specific hyperedges.

As before, under $q_1=q_2=q_3=0$ the two hexads share no common nodes and therefore by condition (iii) we have $\mathbb{E}[\ell_\sigma \ell_{\sigma'}] = 0$.

When $(q_1,q_2,q_3)$ contains at least one 0 and at least one non-zero $q_i$, the two hexads $\sigma$ and $\sigma'$ will not share a common hyperedge. Therefore, for both hexads to be informative we will need a total of four hyperedges (two per hexad), yielding $\mathrm{P}(\widetilde S _{\sigma} \widetilde S _{\sigma'} = 1) = O(\widetilde \rho _N^4)$

In the case $(q_1,q_2,q_3)=(1,1,1)$ both hexads can have one common hyperedge. Therefore, a minimum of three hyperedges will be sufficient for making both hexads informative, yielding $\mathrm{P}(\widetilde S _{\sigma} \widetilde S _{\sigma'} = 1) = O(\widetilde \rho _N^3)$

The case where $(q_1,q_2,q_3)$ contains a mix of 1s and 2s is a subtle one: Under this pattern, two hexads can share two common hyperedges but these hyperedges will have at least one common node. Two such hyperedges cannot belong to an informative hexad because all informative hexads in the setting at hand are characterised by hyperedges that share no nodes. Consequently, under this scenario two informative hexads can share one hyperedge only --- otherwise they cannot be informative. Therefore, $\mathrm{P}(\widetilde S _{\sigma} \widetilde S _{\sigma'} = 1)= O(\widetilde \rho_N^3)$ for any $(q_1,q_2,q_3)$ containing both 1s and 2s.

Finally, the pattern $(q_1,q_2,q_3)=(2,2,2)$ corresponds to the case where the two hexads are simply identical. Therefore, $\mathrm{P}(\widetilde S _{\sigma} \widetilde S _{\sigma'} = 1) = O(\widetilde \rho _N^2)$ under this scenario.

Combining all the results obtained so far with Lemma \ref{lemma:count} confirms the first part of the lemma. The proof of the second part of the lemma is identical to the corresponding part in the proof of Lemma \ref{lemma:main2}.
\end{proof}

\subsubsection{Proof of consistency}

\begin{proof}[Proof of Theorem \ref{thm:consistency_nofe}] The proof is mostly analogous to the proof of Theorem \ref{thm:consistency}. The main difference occurs in proving
\begin{align}
    \frac{1}{m_N \widetilde{p}_N}
    \sum_{\sigma \in \Sigma}
    \widetilde{l}_\sigma(\beta) - \mathbb{E}[\widetilde{l}_\sigma(\beta)]
    \to_p 
    0,
    \label{eq:Cconst}
\end{align}
where $\widetilde{l}_\sigma(\beta)$ is as defined in Section \ref{sec:estimation_nofe}. Applying Chebyshev's inequality, we have that for any $\epsilon > 0$,
\begin{align*}
    \mathrm{P}\left( 
    \left| \frac{
    \sum_{\sigma \in \Sigma} 
    \widetilde{l}_\sigma(\beta) - \mathbb{E}[\widetilde{l}_\sigma(\beta)]
                }
                 {m_N \widetilde{p}_N} \right| > \epsilon 
            \right) 
    \leq
    \frac{1}{\epsilon^2}
    \mathbb{E} 
    \left[
    \left(
        \frac{\sum_{\sigma \in \Sigma} 
    \widetilde{l}_\sigma(\beta) - \mathbb{E}[\widetilde{l}_\sigma(\beta)]}
        {m_N \widetilde{p}_N}
    \right)^2
    \right].
\end{align*}
To bound the expectation on the right hand side of the above display, we will now invoke Lemma \ref{lemma:main2no} and obtain
\begin{align}
    \mathbb{E} 
    \left[
    \left(
        \frac{\sum_{\sigma \in \Sigma} 
    \widetilde{l}_\sigma(\beta) - \mathbb{E}[\widetilde{l}_\sigma(\beta)]}
        {m_N \widetilde{p}_N}
    \right)^2
    \right]
    &=
    O\left(\frac{1}{N^{12}\widetilde{\rho}_N^4}\right)
        [
            O(\widetilde{\rho}_N^4 N^{11})
            +
            O(\widetilde{\rho}_N^3 N^9)
            \notag
            \\
            & \qquad \qquad +
            O(\widetilde{\rho}_N^3 N^8)
            +
            O(\widetilde{\rho}_N^3 N^7)
            +
            O(\widetilde{\rho}_N^2 N^6)
        ]
    \notag
    \\
    &=
    O\left(
    \frac{1}{N}
    \right)
    +
    O\left(
    \frac{1}{N^3 \widetilde{\rho}_N}
    \right)
    \notag
    \\
    & \qquad \qquad +
    O\left(
    \frac{1}{N^4 \widetilde{\rho}_N}
    \right)
    +
    O\left(
    \frac{1}{N^5 \widetilde{\rho}_N}
    \right)
    +
    O\left(
    \frac{1}{N^6 \widetilde{\rho}_N^2}
    \right).
    \label{eq:disc2}
\end{align}
By Assumption \ref{a:nonempty}, $N^3 \widetilde \rho _N \to \infty$ as $N\to \infty$. Therefore, all the terms in the last expression converge to zero with $N$. Consequently, \eqref{eq:Cconst} holds. The rest of the proof proceeds analogously to the proof of Theorem \ref{thm:consistency}.
\end{proof}

\subsubsection{Proof of asymptotic normality}\label{sec:tsimplean}
\begin{proof}[Proof of Theorem \ref{thm:an_nofe}]
The main differences between the proofs of Theorems \ref{thm:an} and \ref{thm:an_nofe} again arise due to the differences in the identifying wirings and the related differences in Lemmas \ref{lemma:main2} and \ref{lemma:main2no}. Let the conditional likelihood score be given by
\begin{align*}
    \widetilde{Z}_N(\beta) = \frac{1}{N^3(N-1)^3}\sum_{\sigma \in \Sigma} \widetilde s_\sigma (\beta).
\end{align*}
By Lemma \ref{lemma:main2no}, it can be shown that 
\begin{align}
    \mathrm{Var}(\widetilde{Z}_N(\beta_0))
    &=
    O\left( \frac{1}{N^{12}} \right)
    \left[ 
        O(\widetilde{\rho}_N^3 N^9)
        +
        O(\widetilde{\rho}_N^3 N^8)
        +
        O(\widetilde{\rho}_N^3 N^7)
        +
        O(\widetilde{\rho}_N^2 N^6)
    \right]
    \notag
    \\
    &=
    O\left( \frac{\widetilde{\rho}_N^3 }{N^{3}} \right)
    +
    O\left( \frac{\widetilde{\rho}_N^3 }{N^{4}} \right)
    +
    O\left( \frac{\widetilde{\rho}_N^3 }{N^{5}} \right)
    +
    O\left( \frac{\widetilde{\rho}_N^2 }{N^{6}} \right),
    \label{eq:andisc1}
\end{align}
and consequently,
\begin{align*}
    \mathrm{Var}\left(\sqrt{\frac{N^3}{\widetilde{\rho}_N^3} }
    \widetilde{Z}_{N}(\beta_0)\right)
    &=
    O(1) 
    + 
    O\left( \frac{1}{N} \right)
    +
    O\left( \frac{1}{N^2} \right)
    +
    O\left( \frac{1}{N^3 \widetilde{\rho}_N } \right).
\end{align*}
Notice that the final term goes to zero as $N\to \infty$ since $N^3\widetilde{\rho}_N\to\infty$  under Assumption \ref{a:nonempty}.
This establishes that the rate of the variance of the score is $\sqrt{N^3/\widetilde{\rho}_N^3}$. 
As before, the explicit expression for the leading term can be obtained by using the final result of Lemma \ref{lemma:main2no}:
\begin{align*}
    \mathrm{Var}\left(\sqrt{\frac{N^3}{\widetilde{\rho}_N^3} }
    \widetilde{Z}_{N}(\beta_0)\right)
    &=
    2^6 
    \frac
    {\mathrm{Cov}(\widetilde{s}_{(i,j,k,a,b,c)} , \widetilde{s}_{(i,j,k,d,e,f)} )}
    {\widetilde \rho _N^3} + o(1).
\end{align*}
Apart from differences in definitions and scaling adjustments (due to informative wirings requiring two and not four hyperedges) the rest of the proof follows along exactly the same lines as the proof of Theorem \ref{thm:an}. In particular, 
defining the projection
\begin{align*}
    \overline{\widetilde{s}}_{ijk} = \mathbb E [\widetilde{s}_{a,a'|b,b'|c,c'} | X_{ijk}, \widetilde{\mathrm{F}}_{ijk},\varepsilon_{ijk}]
\end{align*}
it can be proved that the appropriate H\'{a}jek projection is given by
\begin{align*}
    \widetilde{Z}_N^*
    &= 
    \frac{2^3}{N^3}
    \underset{1 \leq i,j,k \leq N}{\sum \sum \sum }
    \overline{\widetilde s}_{ijk}
\end{align*}
and it is asymptotically equivalent to the score $\widetilde{Z}_N$ in the sense that 
\begin{align*}
    \mathbb{E}
    \left[ 
    \left(
        \sqrt{\frac{N^3}{\widetilde{\rho}_N^3}}
        \widetilde{Z}_{N} 
        - 
        \sqrt{\frac{N^3}{\widetilde{\rho}_N^3}} 
        \widetilde{Z}_{N}^*
    \right)^2
    \right]    
    \to 0,
    \qquad
    \text{as } N\to\infty.
\end{align*}
Switching to triad indexing with $h=1,\ldots,H$ where $H=N^3$, and letting $\widetilde{\mathbf F}$ denote the collection of $\widetilde{\mathrm{F}}_{ijk}$ across all triads, we next define
\begin{align}
    \widetilde{R} _h
    =
    \frac{c'\widetilde{\Gamma}_0^{-1} \left(\overline{\widetilde s}_h / \sqrt{\widetilde{\rho} _N^3} \right) }
    {\sqrt{c'\widetilde{\Gamma}_0^{-1} \overline{\widetilde{M}} _H \widetilde{\Gamma}_0^{-1} c }}
    \qquad
    \text{where}
    \qquad
    \overline{\widetilde{M}}_H 
    = 
    \frac{1}{H{\widetilde \rho}_N^{3}} \sum_{h=1}^H  
    \mathbb{E} [\overline{\widetilde s}_h \overline{\widetilde s}_h' | \mathbf{X}, \mathbf{\widetilde F}].
    \label{eq:Mnofe}
\end{align}
Then, by arguments analogous to those made in Section \ref{sect:an1proj}, one obtains
\begin{align*}
    \frac{1}{\sqrt{H}} \sum_{h=1}^H \widetilde{R}_h \to_d \mathcal N (0,1),
\end{align*}
and consequently,
\begin{align*}
    \frac{c'\widetilde{\Gamma}_0^{-1}}
    {\sqrt{c'\widetilde{\Gamma}_0^{-1} \overline{\widetilde{M}} _H \widetilde{\Gamma}_0^{-1} c }}
    \sqrt{\frac{H}{\widetilde{\rho}_N^3}}
    \widetilde{Z}_N^*
    &=
    2^3 \frac{1}{\sqrt{H}}\sum_{h=1}^H \widetilde{R}_h
    \to_d
    \mathcal N (0,2^6).
\end{align*}
The final step is to establish
\begin{align*}
    \frac{1}{m_N \widetilde{p}_N} \sum_{\sigma \in \Sigma} 
    \widetilde{H}_{\sigma}(\widehat \beta) 
    \to_p 
    \widetilde{\Gamma}_0,
\end{align*}
which implies that
\begin{align*}
    \frac{1}{m_N \widetilde{\rho}_N^2} \sum_{\sigma \in \Sigma} 
    \widetilde{H}_{\sigma}(\widehat \beta) 
    \to_p 
    \widetilde{\Gamma}_0.
\end{align*}
This can be obtained by following the same steps as in Section \ref{sect:hessian}, while keeping in mind that $\widetilde{p}_N = \widetilde{\rho}_N^2(1-\widetilde{\rho}_N)^6=O(\widetilde{\rho}_N^2)$ since in the current heterogeneity structure a hexad requires the formation of two particular hyperedges (and the absence of all the remaining ones) to be informative. Also, in obtaining stochastic equicontinuity, we will now have the following: first, by invoking Lemma \ref{lemma:main2no}
\begin{align*}
    \frac{1}{\widetilde \rho _N^2  (N(N-1))^3} 
    \sum_{\sigma\in \Sigma} 
    (\widetilde S _\sigma 
    -
    \mathbb E [\widetilde S _\sigma ]
    )
    &=
    O_p\left(
        \sqrt{\max\left\{
        \frac{1}{N},
        \frac{1}{ N^3 \widetilde \rho _N },
        \frac{1}{ N^4 \widetilde \rho _N },
        \frac{1}{ N^5 \widetilde \rho _N },
        \frac{1}{ N^6 \widetilde \rho _N^2 }
        \right\}}
    \right)
    \\
    &=
    o_p\left( 1 \right),
\end{align*}
where we have used $N^3 \widetilde \rho _N \to \infty$ as implied by Assumption \ref{a:nonempty}. Next, 
\begin{align*}
    \frac{1}{\widetilde \rho _N^2 (N(N-1))^3} 
    \sum_{\sigma\in \Sigma} 
    \mathbb E [\widetilde S _\sigma ]
    &=
    \frac{1}{\widetilde \rho _N^2 (N(N-1))^3} 
    \sum_{\sigma\in \Sigma} 
    \mathrm{P} (\widetilde S _\sigma =1 )
    =
    O(1).
\end{align*}
Consequently, the counterpart of \eqref{eq:h3} of Section \ref{sect:hessian} becomes,
\begin{align*}
    \sup _{\beta \in \mathcal B}
    \left|\left|
    \frac{1}{\widetilde \rho _N^2}\nabla_{\beta\beta\beta_p}
    \widetilde l _N( \beta)
    \right| \right|_{F}
    &=
    O\left(1\right)
    +
    o_p\left(1\right)
    =
    O_p(1).
\end{align*}
The rest of the proof is analogous to the triadic model with dyad-effects. Combining all the results using an appropriate mean-value expansion confirms the statement of the theorem.
\end{proof}

\newpage

\section{Dyadic model}
\label{sec:dyadic}

In this part we analyse the dyadic link formation model in bipartite networks. The asymptotics of the dyadic link formation model has already been analysed in various flavours by \citet{Graham17} and \citet{Jochmans18}. The latter considers  bipartite networks but its theoretical structure is different from ours, so those results cannot directly be used here. \citet{Graham17}, on the other hand, does not consider a bipartite network. However the results we derive here are essentially the same as Graham's. Therefore, the derivations in this section are provided for sake of completeness. In particular, our aim here is not to present a complete theory for dyadic link formation in a bipartite network but simply to obtain several key results that we use to highlight the fundamental differences between the triadic link formation model with dyad-specific fixed effects and the dyadic link formation model with node-specific effects.

\subsection{The model}
We consider dyadic link formation, determined by the model
\begin{align}
    Y_{ij} = 1\{A_{i} + B_{j} + X_{ij}' \beta_0 - \varepsilon_{ij} > 0\},
    \label{eq:ofe_dyadic}
\end{align}
with the node-specific fixed effects $A_i$ and $B_j$, where $i=1,\ldots,N$ and $j=1,\ldots,N$. The ordered pair $(i,j)$ denotes a dyad consisting of two nodes, with node $i$ belonging to the first part and $j$ to the second part. The set of all $N^2$ dyads is given by $\mathbb D_N = \{1,\ldots,N\}^2$. The fixed effects belonging to the dyad $(i,j)$ are denoted by $\mathrm F _{ij} = (A_i,B_j)$. Similar to before, $X_{ij} \in \mathbb R^P$ is the dyad-specific observable covariate and $\beta_0$ is the associated $P \times 1$ parameter vector. The random shock $\varepsilon_{ij}$ is logistic and conditionally independent across dyads, as formalised in the next assumption.
\begin{assumption}[Likelihood]\label{a:logit_dyadic}
The conditional likelihood of the network $\mathbf Y = \mathbf y$ is
\begin{equation*}
    \mathrm{P}\left(\mathbf Y = \mathbf y | \mathbf X, \mathbf{F}\right) 
    = 
    \prod_{(i,j) \in \mathbb D_N} \mathrm{P} \left(\left. Y_{ij} = y_{ij} \right| X_{ij}, \mathrm{F}_{ij} \right),
\end{equation*}
where 
\begin{align*}
    \mathrm{P} \left(\left. Y_{ij} = y \right| X_{ij}, \mathrm{F}_{ij} \right) 
    = 
    \begin{cases}
    \Lambda\left( X_{ij}' \beta_0 + A_{i} + B_{j} \right) & \text{ if }y = 1, \\
    1-\Lambda\left( X_{ij}' \beta_0 + A_{i} + B_{j} \right) & \text{ if }y=0.
    \end{cases}
\end{align*}
\end{assumption}

\newpage

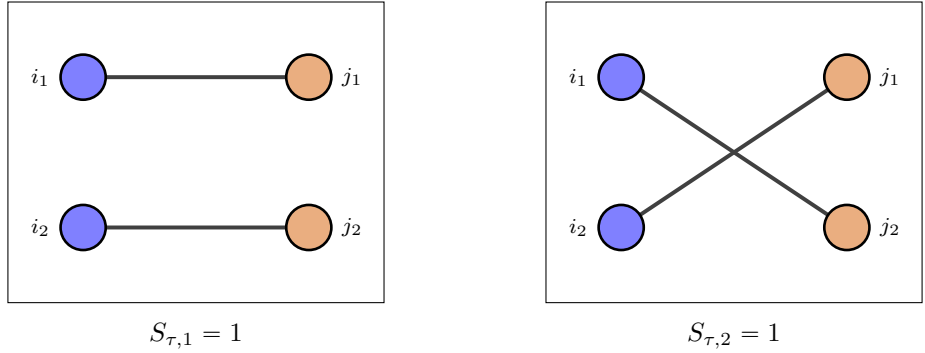
\begin{figure}[t]
    \centering
    \begin{subfigure}[b]{0.45\linewidth}
        \centering
        \begin{tikzpicture}
        \SetVertexStyle[FillOpacity = 0.5]
        \draw (0.0,0.0) rectangle (5,4);
        \Vertex[x = 1, y = 3, color = blue, label = $i_1$, position = left]{i1}
        \Vertex[x = 1, y = 1, color = blue, label = $i_2$, position = left]{i2}
        \Vertex[x = 4, y = 3, color = red, label = $j_1$, position = right]{j1}
        \Vertex[x = 4, y = 1, color = red, label = $j_2$, position = right]{j2}
        
        \Edge(i1)(j1)
        \Edge(i2)(j2)
        \end{tikzpicture}
        \caption*{$ S _{\tau, 1}=1$}
    \end{subfigure}
    \begin{subfigure}[b]{0.45\linewidth}
        \centering
        \begin{tikzpicture}
        \SetVertexStyle[FillOpacity = 0.5]
        
        \draw (0.0,0.0) rectangle (5,4);
        \Vertex[x = 1, y = 3, color = blue, label = $i_1$, position = left]{i1}
        \Vertex[x = 1, y = 1, color = blue, label = $i_2$, position = left]{i2}
        \Vertex[x = 4, y = 3, color = red, label = $j_1$, position = right]{j1}
        \Vertex[x = 4, y = 1, color = red, label = $j_2$, position = right]{j2}
        
        \Edge(i1)(j2)
        \Edge(i2)(j1)
        \end{tikzpicture}
        \caption*{$ S _{\tau, 2}=1$}
    \end{subfigure}
    \caption{Informative tetrads for the dyadic link formation model.}
    \label{fig:identifyingtetrads}
\end{figure}

\subsection{Sufficiency}
In dyadic link formation models, identification is based on tetrads where each node has degree one, yielding a degree sequence of $(1,1,1,1)$. The two identifying tetrads are illustrated in Figure \ref{fig:identifyingtetrads}, which correspond to the tetrads found earlier by \citet{Charbonneau17} and \citet{Jochmans18}. To formalise, let a generic tetrad be given by $\tau = (i_1,j_1,i_2,j_2)$ where $i_1,i_2$ are the nodes from the first part and the remaining nodes belong to the second part. We also let $X_\tau$ and $\mathrm F _\tau$ be the tetrad-specific covariates and fixed effects. Define,
\begin{align*}
    S _{\tau, 1}
    = 
    Y_{i_1 j_1}
    Y_{i_2 j_2}
    \overline Y_{i_1 j_2}
    \overline Y_{i_2 j_1}
    \qquad
    \text{and}
    \qquad
    S _{\tau, 2}
    = 
    \overline Y_{i_1 j_1}
    \overline Y_{i_2 j_2}
    Y_{i_1 j_2}
    Y_{i_2 j_1},
\end{align*}
where $\overline Y_{ij} = 1-Y_{ij}$. The case $S _{\tau,1}=1$ corresponds to the wiring on the left panel of Figure \ref{fig:identifyingtetrads} whereas $S _{\tau,2}=1$ corresponds to the second wiring in the same figure. Then, 
\begin{align}
    S _\tau = S _{\tau,1} + S _{\tau,2},
    \label{eq:Stildedefn_dyadic}
\end{align}
indicates whether the tetrad $\tau$ exhibits one of these two wirings.
The following theorem confirms that these two wirings are informative.

\begin{theorem}[Sufficiency]\label{thm:sufficiency_dyadic} 
If Assumption \ref{a:logit_dyadic} holds, then, for any tetrad $\tau = (i_1,j_1,i_2,j_2)$,
\begin{align*}
    \mathrm{P}( S _{\tau,1} = 1 |  S _{\tau} = 1, \mathrm{F}_{\tau}, X_{\tau}) 
    = \Lambda( W_{\tau}' \beta_0),
\end{align*}     
where $W_{\tau} = 
    \left(X_{i_1 j_1}+X_{i_2 j_2}\right) 
    - 
    \left(X_{i_1 j_2}+X_{i_2 j_1}\right) $.
\end{theorem}
See \citet{Charbonneau17} for a proof.

\subsection{The conditional logit estimator}
Let 
\begin{align*}
    W _{\tau,1} = X_{i_1 j_1}+X_{i_2 j_2}
    \qquad
    \text{and}
    \qquad
    W _{\tau,2} = X_{i_1 j_2}+X_{i_2 j_1},
\end{align*}
and define
\begin{align*}
    p _{\tau, c}(\beta)
    &= 
    \frac
    {\exp( W _{\tau, c}' \beta)}
    {\sum_{c'=1}^2 \exp  ( W _{\tau, c'}' \beta ) }.
\end{align*}
Let also $\mathbf{T}$ denote the set of all tetrads.  The tetrad conditional likelihood function is given by
\begin{align*}
    \widecheck l _N(\beta) 
    &= 
    \frac{1}{|\mathbf{T} |}
    \sum_{\tau \in \mathbf T } 
    l _\tau(\beta),
    \qquad
    \text{where}
    \qquad
    l _\tau(\beta)
    = 
    \sum_{c=1}^2  S _{\tau, c} \log p _{\tau, c}(\beta).
\end{align*}
Consequently, the conditional logit estimator of $\beta_0$ is given by
\begin{align*}
    \widecheck \beta = \arg\max_{\beta \in \mathcal B} \widecheck l _N(\beta).
\end{align*}
The score function is given by 
\begin{align*}
    \widecheck Z _N(\beta)
    =
    \frac{1}{|\mathbf T|}
    \sum_{\tau \in \mathbf T } 
    s _\tau(\beta),
\end{align*} 
where
\begin{align*}
    s _{\tau}(\beta)
    = 
    \sum_{c = 1}^2  
    S _{\tau, c} 
    \left(
        W _{\tau, c} - \overline{ W }_\tau(\beta) 
    \right)
    \qquad
    \text{and}
    \qquad
    \overline{ W }_\tau(\beta) 
    = 
    p _{\tau,1}(\beta) 
    W _{\tau, 1} 
    + 
    p _{\tau,2}(\beta) 
    W_{\tau, 2}.    
\end{align*}

\subsection{Useful lemmas}
Analogous to the corresponding analysis in Sections \ref{app:proofs} and \ref{sec:triadicno}, we first obtain the following general lemmas. Essentially, these are simpler versions of Lemmas \ref{lemma:count} and Lemmas \ref{lemma:main2}/\ref{lemma:main2no}. Therefore, their proofs will only highlight important points but not go into further detail.

\begin{lemma}\label{lemma:count_dyadic} 
For tetrad pairs with $(q_1, q_2)$ common nodes in parts 1 and 2 respectively, where $q_i \in \{0,1,2\}$, the number of such pairs is 
    \begin{align*}
        m_{(q_1, q_2), N} = O(N^{8-(q_1+q_2)})
        \qquad
        \text{as }
        N\to\infty.
    \end{align*}
\end{lemma}
\begin{proof}[Proof of Lemma \ref{lemma:count_dyadic}]
This is essentially a simpler version of Lemma \ref{lemma:count}, the main differences being that now the object of interest contains four rather than six nodes, and the nodes are drawn from two rather than three parts. Following the same logic as in the proof of Lemma \ref{lemma:count}, the first tetrad $\tau$ in a tetrad-pair $(\tau,\tau')$ can be chosen in $O(N^4)$ ways. Depending on the exact pattern of common nodes across the two tetrads, given the first tetrad $\tau$ the second tetrad $\tau'$ can then be chosen in $O(N^{2-q_1}\times N^{2-q_2})$ ways. Then,
\begin{align*}
    m_{(q_1, q_2), N} &= (\text{Number of ways to choose } \tau) \times (\text{Number of ways to choose } \tau' | \tau) \\
    &= O(N^4) \times O(N^{4-(q_1+q_2)}) \\
    &= O(N^{8-(q_1+q_2)}),
\end{align*}
as stated.
\end{proof}

\begin{lemma}\label{lemma:main2_dyadic}
Let $\ell_\tau = \ell(X_\tau, \mathrm{F}_\tau, \varepsilon_\tau)$ be a random vector that is bounded for all $\tau \in \mathbf{T} $. Suppose the following properties hold for each $\tau\in \mathbf{T}$:
(i) $\mathbb{E} [\ell_\tau] = 0 $; 
(ii) $\ell_\tau \indep \ell_{\tau'} | \mathbf{X}, \mathbf{F}$ for all $\tau'$ such that $(\tau, \tau')$ share no common nodes in at least one part; 
(iii) $\ell_\tau \indep \ell_{\tau'}$ for all $\tau'$
such that $(\tau, \tau')$ share no common nodes in any of the two parts;
(iv) for $S_\tau$ as defined in \eqref{eq:Stildedefn_dyadic}, $\ell_\tau=0$ if $S_{\tau}=0$. 
Let $\widecheck \rho _N= \mathrm{P} (Y_{ij}=1)$ for $Y_{ij}$ as defined in \eqref{eq:ofe_dyadic}. 

Then, 
\begin{align*}
    \mathbb{E}\left[
    \left(
    \sum_{\tau \in \mathbf{T}}
    \ell_{\tau}
    \right)^2
    \right]
    =
    \sum_{q_1=0}^{2}\sum_{q_2=0}^{2}
    \mathrm{C}_{(q_1,q_2),N},
\end{align*}
where
\begin{enumerate}
    \item $\mathrm C _{(0,0),N}=0$
    \item For $(q_1,q_2)\in \{(0,1),(1,0),(0,2),(2,0)\}$, $\mathrm{C}_{(q_1,q_2),N} = O(N^{8-(q_1+q_2)}\widecheck{\rho} _N^4)$
    \item  $\mathrm{C}_{(1,1),N} = O(N^6 \widecheck \rho _N^3)$
    \item For $(q_1,q_2)\in \{(2,1),(1,2)\}$: $\mathrm{C}_{(q_1,q_2),N} = O(N^5 \widecheck \rho _N^3)$
    \item $\mathrm{C}_{(2,2),N} = O(N^4 \widecheck \rho _N^2)$
\end{enumerate}

If, in addition, $\mathbb{E}[\ell_\tau | \mathbf{X}, \mathbf{F}] = 0$, then $\mathrm{C}_{(q_1,q_2),N} = 0$ for all cases where any $q_i = 0$. 
\end{lemma}
\begin{proof}[Proof of Lemma \ref{lemma:main2_dyadic}]
    The proof is a simpler version of the Proof of Lemma \ref{lemma:main2} and we provide a sketch here. Remembering equations \eqref{eq:varmdecomp_corrected}, \eqref{eq:defCqN_corrected} and \eqref{eq:pdelta1b_corrected},
    and the related discussions, the rate of each component is determined by two objects: (i) the number of ways in which a particular combination of common nodes $(q_1,q_2)$ can be obtained; and (ii) the probability of two generic tetrads $(\tau,\tau')$ being informative when they share $(q_1,q_2)$ common nodes. The first object is already determined by Lemma \ref{lemma:count_dyadic}, and here we discuss the second object for each possible pair of $(q_1,q_2)$.

    First, we remember that informative tetrads consist of two edges. Second, edges are formed between nodes from different parts. Therefore, for formation of a common edge between two tetrads, we need both $q_1$ and $q_2$ to be at least 1. 

    By assumptions (i) and (iii), when $\tau$ and $\tau'$ share no common nodes, $\mathbb E [\ell_\tau \ell_{\tau'}]= \mathbb E [\ell_\tau ] \mathbb E [\ell_{\tau'}]=0$ and we directly obtain $\mathrm C _{(0,0),N}=0$.
    
    Next, when any $q_i=0$, there will not be a common edge. Consequently, we will need two edges per tetrad, implying $\mathrm P (S _\tau  S _{\tau'}=1)=O(\widecheck \rho _N^4)$ in this case. Combining this with Lemma \ref{lemma:count_dyadic}, we obtain $\mathrm C_{(q_1,q_2),N}=O(N^{8-(q_1+q_2)} \widecheck{\rho}_N^4)$.

    When $(q_1,q_2)=(1,1)$, it is possible for both tetrads to share one common edge. In addition to this, each tetrad will require one more edge, adding up to three edges in total. By Lemma \ref{lemma:count_dyadic}, there are $O(N^6)$ such cases where a tetrad-pair can share two common nodes. This yields $\mathrm C _{(1,1),N} = O(N^6 \widecheck \rho _N^3)$.

    In the next case where one $q_i=2$ and the other equals 1, there can still be at most one common edge. Therefore, the only change comes through Lemma \ref{lemma:count_dyadic}, which yields the stated result.

    Finally, $(q_1,q_2)=(2,2)$ corresponds to the case where the two tetrads are identical. In that case, only two edges will be required, and this scenario corresponds to $O(N^4)$ tetrad-pair configurations, leading to $\mathrm C _{(2,2),N} = O(N^4 \widecheck \rho_N^2)$.

    The very last statement of the Lemma follows from the fact that whenever any $q_i=0$, the two tetrads will be conditionally independent. Then, as before, it follows from an application of the Law of Iterated Expectations that $\mathbb E [\ell_\tau \ell_{\tau'}]=0$, directly yielding $\mathrm C _{(q_1,q_2),N}=0$.
\end{proof}

\subsection{Variance decompositions}\label{sec:dyadicdecomp}

Finally, we obtain the variance decompositions for the model at hand, analogous to those obtained for the triadic link formation models. 

We first obtain the decomposition that would be used in proving the consistency of $\widecheck \beta$, analogous to the decompositions in \eqref{eq:disc1} and \eqref{eq:disc2} that were used in obtaining consistency in the triadic link formation models with dyad- and node-effects, respectively. 
We remember that in the dyadic case the number of tetrads is $O(N^4)$ and the probability of a tetrad being informative is $O(\widecheck  \rho_N^2)$. Then, using Lemma \ref{lemma:main2_dyadic} we have
\begin{align}
        \mathbb{E} 
        \left[
        \left(
            \frac{\sum_{\tau \in \mathbf{T}} 
            \{
                l_\tau(\beta) 
                -
                \mathbb E [ l_\tau(\beta) ]
            \}
            }
            {O(N^4 \widecheck{\rho}_N^2)}
        \right)^2
        \right]
        =&
        O\left(\frac{N^{7} \widecheck \rho _N^4}{N^{8} \widecheck \rho_N^4}\right)
        +
        O\left(\frac{N^{6}\widecheck{\rho}_N^3}{N^{8} \widecheck \rho_N^4}\right)
        \notag
        \\
        &+
        O\left(\frac{N^5\widecheck \rho _N^3}{N^{8} \widecheck \rho_N^4}\right)
        +
        O\left(\frac{N^4 \widecheck \rho _N^2}{N^{8} \widecheck \rho_N^4}\right)
        \notag
        \\
        =&
        O\left(\frac{1}{N}\right)
        +
        O\left(\frac{1}{N^{2} \widecheck \rho_N}\right)
        \notag
        \\
        &+
        O\left(\frac{1}{N^{3} \widecheck \rho_N}\right)
        +
        O\left(\frac{1}{N^4 \widecheck \rho _N^2}\right),
        \label{eq:disc3}
    \end{align}
    which is essentially the same decomposition as obtained on p.1054 of \citet{Graham17}. 
    For consistency in the current setting, one  would require that the right hand side of \eqref{eq:disc3} converge to zero as $N\to\infty$. Aymptotically non-zero average degree (i.e. $N \widecheck \rho _N >0$ for $N$ sufficiently large) is a sufficient condition.

    Next, we consider the variance decomposition of the normalised score (analogous to the decompositions in \eqref{eq:keydecomp} and \eqref{eq:andisc1}). The score contribution of a generic tetrad satisfies the final condition of Lemma \ref{lemma:main2_dyadic}. Therefore, using Lemma \ref{lemma:main2_dyadic} we obtain
    \begin{align}
        \mathrm {Var} 
        \left(
        \sqrt{\frac{N^2}{\widecheck \rho _N^3} }
        \widecheck{Z} _N
        \right)
        &=
        O\left(1\right)
        +
        O\left(\frac{1} {N}\right)
        +
        O\left(\frac{1} {N^2 \widecheck \rho _N}\right),
        \label{eq:dyadicdecomp2}
    \end{align}
    which is, again, essentially the same result as obtained by \citet{Graham17}; see his equation (32). Importantly, as in \eqref{eq:disc3}, asymptotic non-emptiness is sufficient to guarantee that the first term on the right hand side of \eqref{eq:dyadicdecomp2} will be the only asymptotically non-vanishing term as $N\to\infty$.

\newpage

\bibliographystyle{chicago3}
\bibliography{draft0.bbl}

\end{document}